\documentclass[journal]{IEEEtran}

\usepackage{graphicx}                    
\usepackage{amsmath}                     
\usepackage{amsfonts}                    
\usepackage{amssymb}                     
\usepackage{amsthm}                      

\usepackage{xcolor}
\usepackage[normalem]{ulem}              
\usepackage[noadjust,verbose,sort]{cite} 
\usepackage{setspace}                    
\usepackage{tikz} 
\usepackage{textcomp}
\usepackage{algorithm,refcount}
\usepackage{algpseudocode}
\usepackage{tabularx,booktabs}
\usepackage{adjustbox}
\usepackage{array,etoolbox}
\usepackage{comment}
\usepackage{subfigure}

\theoremstyle{remark}\newtheorem{proposition}{Proposition}
\theoremstyle{remark}\newtheorem{theorem}{Theorem}
\theoremstyle{remark}
\theoremstyle{remark}\newtheorem{definition}{Definition}
\theoremstyle{remark}
\theoremstyle{remark}\newtheorem*{remark}{Remark}
\theoremstyle{remark}
\theoremstyle{remark}\newtheorem{appxlemma}{Lemma}
\numberwithin{appxlemma}{section} 
\theoremstyle{remark}
\theoremstyle{remark}\newtheorem{appxcorollary}{Corollary}
\newtheoremstyle{case}{}{}{}{}{}{:}{ }{}
\theoremstyle{case}\newtheorem{case}{{\bf Case}}
\newtheoremstyle{subcase}{}{}{}{}{}{:}{ }{}
\theoremstyle{subcase}

\DeclareMathOperator*{\argmax}{arg\,max}

\newcommand{\ta}{\tilde{a}}

\newcommand{\ts}{\tilde{s}}

\newcommand{\cA}{\mathcal{A}}
\newcommand{\cB}{\mathcal{B}}

\newcommand{\cG}{\mathcal{G}}
\newcommand{\cI}{\mathcal{I}}

\renewcommand{\Pr}{\mathsf{P}}

\renewcommand{\Pr}{\mathsf{P}}

\newcommand{\FIGDIR}{figs}

\newcommand{\mycomment}[1]{%
}%
\begin{document}

\title{Design of a Stochastic Traffic Regulator for End-to-End Network Delay Guarantees}

\author{Massieh Kordi Boroujeny and Brian L. Mark\\
        Dept. of Electrical and Computer Engineering \\
        George Mason University, Fairfax, VA}

\maketitle

\begin{abstract}
Providing end-to-end network delay guarantees in packet-switched networks such as the Internet
is highly desirable for mission-critical and delay-sensitive data transmission, yet it
remains a challenging open problem.   Due to the looseness
of the deterministic bounds, various frameworks for stochastic network calculus have been
proposed to provide tighter, probabilistic bounds on network delay, at least in theory.  However, little attention has been devoted to the  problem of regulating traffic according to stochastic burstiness bounds, which is necessary in order to guarantee the delay bounds
in practice.  We design and analyze a stochastic traffic regulator that can be used in conjunction
with results from stochastic network calculus to provide probabilistic guarantees on end-to-end
network delay.  Numerical results are provided to demonstrate the performance
of the proposed traffic regulator.\footnote{This work was supported in part by the U.S.\ National Science Foundation under Grant No.\ 1717033.  A preliminary version of this work
was presented at the IEEE Int.\ Conf.\ on Communications (ICC'2020)~\cite{KordiICC:2020}.}
\end{abstract}

\begin{IEEEkeywords}
Stochastic network calculus, traffic shaper, end-to-end delay, traffic burstiness bounds.
\end{IEEEkeywords}

\section{Introduction}

Currently, the Internet does not provide end-to-end delay guarantees for traffic flows.  Even if the path
taken by a given traffic flow is fixed, e.g., via mechanisms such as software-defined networking or
multi-protocol label switching, network congestion arising from other flows can result in
highly variable delays.  The variability and random nature of traffic flows in a packet-switched network
make it very challenging to provide any type of performance guarantees.  The standard approach to providing
network performance guarantees consists of two basic elements:
\begin{enumerate}
\item {\em Admission control:}  A new flow should only be admitted to the network if sufficient resources
are available for the new flow, as well as existing flows, to maintain their performance guarantees. 
\item {\em Traffic regulation:}  The traffic flow must be regulated to ensure that it does not use
more resource than what was negotiated by the admission control scheme.  
\end{enumerate}
Admission control relies on a means of characterizing the traffic.  On the other hand,
the random and bursty nature of traffic flows in packet-switched networks make them difficult to characterize.
Even if the flows can be modeled as random arrival processes, the problem of developing a resource
allocation scheme to guarantee end-to-end performance based on such models is practically intractable.

In his seminal work, Cruz~\cite{Cruz1991a,Cruz1991b} proposed the so-called $(\sigma, \rho)$ characterization
of traffic, which imposes a deterministic bound on the burstiness of a traffic flow.  By bounding traffic
flows according to $(\sigma, \rho)$ parameters, Cruz developed a network calculus which determined how
these parameters propagate through network elements and from which end-to-end delay bounds
could be derived.  An important feature of the $(\sigma, \rho)$ characterization is that it could
be enforced by a traffic regulator. In practice, however, the $(\sigma, \rho)$
characterization leads to delay bounds that are very loose, which would lead to low 
network resource utilization. Nevertheless, the $(\sigma, \rho)$ characterization was the basis for further
research into stochastic bounds on traffic burstiness and stochastic network calculus to provide tighter,
probabilistic end-to-end delay guarantees.  

The development of stochastic network calculus and associated performance bounds remains an active topic of
research~\cite{Fidler:2015}.  However, little attention has been devoted to the problem of traffic regulation to ensure
that the input traffic of a network conforms to a stochastic traffic bound.  In the
deterministic network calculus of Cruz, the $(\sigma, \rho)$ traffic regulator is tightly
coupled to the $(\sigma, \rho)$ traffic characterization.  In effect, the $(\sigma, \rho)$ traffic characterization
is defined operationally in terms of a $(\sigma, \rho)$ traffic regulator.  To our knowledge, a traffic
regulator to enforce a stochastic traffic bound has not been addressed
previously, despite the fact that such a regulator
is necessary to provide traffic guarantees in real networks.

In this paper, we  develop a traffic regulator to enforce the so-called
{\em generalized Stochastically Bounded Burstiness} (gSBB) traffic bound in~\cite{Yin2002,JIANG20092011}.
We refer to our proposed regulator as a stochastic $(\sigma^*, \rho)$ regulator, since the burst size parameter can take on one of finite set of values. We describe the design and basic properties of the stochastic $(\sigma^*, \rho)$ regulator and develop practical implementations.  Our analytical
results establish that it enforces the gSBB bound. We demonstrate the operation of the
$(\sigma^*, \rho)$ regulator via numerical examples. 

The remainder of the paper is organized as follows.  
In Section~\ref{sec:background}, we review basic concepts in deterministic and
stochastic network calculus.
In Section~\ref{sec:deterministic_regulator}, we review key properties of the
deterministic $(\sigma, \rho)$ regulator and develop some new results for its analysis, which are applied in Section~\ref{sec:Stochastic_Regulator} to the design and implementation of the proposed stochastic $(\sigma^*, \rho)$
regulator.  Numerical results demonstrating the performance
of the $(\sigma^*, \rho)$ regulator are presented in Section~\ref{sec:numerical}.  
Concluding remarks are given in Section~\ref{sec:conclusion}.

\section{Background on Network Calculus}
\label{sec:background}

\subsection{Deterministic $(\sigma, \rho)$ Network Calculus}

\begin{figure}
\centering
\begin{adjustbox}{width=0.62\columnwidth}
\begin{tikzpicture}
\draw [thick,->](-1,0) -- (0,0);
\draw [thick,->](2,0) -- (5,0);
\node [above] at (3.5,0) {\begin{small}$R_{\rm o}\sim(\sigma+\delta,\rho)$
\end{small}};
\node [below] at (3.5,0) {\begin{small}$\delta=(1 - \rho/C)L_{\max}$
\end{small}};
\node [above] at (-0.5,0) {\begin{small}$R_{\rm i}$\end{small}};
\draw (0,-.5) -- (0,.5) -- (2,.5) -- (2,-.5) -- (0,-.5);
\node at (1,0) {$(\sigma,\rho)$ };
\end{tikzpicture}
\end{adjustbox}
\caption{$(\sigma,\rho)$ regulator with input/output links of capacity $C$.}
\label{fig:Deterministic Regulator}
\end{figure}
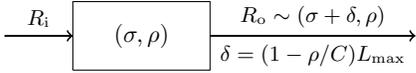

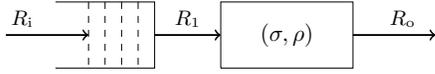
\begin{figure}[t]
    \centering
    \begin{adjustbox}{width=0.65\columnwidth}
   \begin{tikzpicture}
    \draw [thick,->](-2.75,0) -- (-1.5,0);
   \draw [thick,->](-0.5,0) -- (0.5,0);
    \draw [thick,->](2.5,0) -- (3.75,0);
    \draw (-2,.5) -- (-.5,.5) -- (-.5,-.5) -- (-2,-.5);
    \draw [dashed] (-1.5,.5)--(-1.5,-.5);
    \draw [dashed] (-1.25,.5)--(-1.25,-.5);
    \draw [dashed] (-1,.5)--(-1,-.5); 
    \draw [dashed] (-0.75,.5)--(-0.75,-.5);
    \draw (0.5,-.5) -- (0.5,0.5) -- (2.5,0.5) -- (2.5,-.5) -- (0.5,-.5);
    \node at (1.5,0){$(\sigma,\rho)$};
    \node [above] at (-2.5,0) {\begin{small}$R_{\rm i}$\end{small}};
    \node [above] at (0,0) {\begin{small}$R_1$\end{small}};
    \node [above] at (3.25,0) {\begin{small}$R_{\rm o}$\end{small}};
    \end{tikzpicture}
    \end{adjustbox}
    \caption{$(\sigma,\rho)$ traffic shaper with front-end buffer.}
    \label{fig:deterministic_shaper}
\end{figure}

Let $R = \{ R(t) : t \geq 0 \}$ denote a traffic process such that the amount of traffic
arriving in an interval $[s, t]$ is given by
\begin{align}
A(s, t; R) = \int_s^t R(\tau)\;\mathrm{d}\tau .
\end{align}
In this paper, we assume that traffic processes are
in continuous-time, although our development could also be framed in discrete-time.  
A traffic stream $R$ is said to be $(\sigma, \rho)$-bounded, 
denoted as $R \sim (\sigma, \rho)$, if  
\begin{equation}
  \label{eq:deterministic (sigma,rho)}
A(s, t; R)  \leq \rho(t-s)+\sigma, ~~\forall s \in [0, t], 
\end{equation}
where $\sigma, \rho \geq 0$.

In conjunction with traffic characterization~\eqref{eq:deterministic (sigma,rho)}, Cruz~\cite{Cruz1991a} introduced a traffic regulator to enforce conformance to the $(\sigma, \rho)$ parameters.
For an idealized fluid model of input traffic, a $(\sigma,\rho)$ traffic regulator ensures that the {\em output} traffic stream 
$R_{\rm o} \sim (\sigma,\rho)$ and traffic departs the regulator
in the same order as it arrives to the regulator, i.e., the service discipline is
first-come first-served (FCFS).  When the traffic consists of discrete
packets of maximum length $L_{\max}$ and the input/output links to the regulator have finite capacity $C$, the output
traffic stream satisfies $R_{\rm o} \sim (\sigma + \delta, \rho)$, where 
(see Fig.~\ref{fig:Deterministic Regulator})
\begin{align}
    \delta = (1 - \rho/C)L_{\max} .
\label{eq:delta}
\end{align} 

Traffic regulation can be accomplished by the dropping, tagging (as lower priority), or delaying of packets. In the first two cases, the traffic regulator is sometimes referred to as a {\em traffic policer} whereas in the third case it is referred to as a {\em traffic shaper}.  The traffic regulators discussed in this paper will be of the traffic shaper variety.  A traffic shaper includes a front-end buffer, which stores packets that are delayed
in the process of forcing the output traffic to conform to $(\sigma, \rho)$ (see Fig.~\ref{fig:deterministic_shaper}).

The $(\sigma, \rho)$ network calculus originally developed by Cruz has been extended
with more general notions of arrival envelopes, service curves, and 
min-plus algebra~\cite{Chang2000}.  However, practical implementations
of the network calculus have been based on the $(\sigma, \rho)$ characterization and its variants.
Deterministic network calculus has recently been applied to dynamic bandwidth allocation
for Software-Defined Networks (SDNs)~\cite{Lee:2018}.


\subsection{Stochastic Network Calculus}

The $(\sigma, \rho)$ bound in~\eqref{eq:deterministic (sigma,rho)} tends to be
rather loose for bursty traffic.  Similarly, end-to-end delay bounds derived via the deterministic network calculus will be loose in practical networking scenarios, since they are based on a worst-case analysis.  Moreover, the deterministic network
calculus cannot exploit the phenomenon of statistical multiplexing in networks.
These considerations motivated the development of {\em stochastic} traffic burstiness bounds,
and an associated stochastic network calculus to allow the derivation
of stochastic end-to-end delay bounds.
An early proposal for a stochastic traffic burstiness bound was the Exponentially Bounded Burstiness (EBB)
of Yaron and Sidi~\cite{Yaron1994}, which involves an exponential bounding function.  A 
related traffic bound based on moment generating functions was 
proposed by Chang~\cite{Chang1994}. 

In this paper, we focus on the 
{\em generalized} Stochastically Bounded Burstiness (gSBB)
proposed in~\cite{JIANG20092011}.  A  stochastic traffic  process $R$ is gSBB
with upper  rate $\rho$ and  bounding function
$f \in \mathcal{BF}$ if
\begin{align}
 \Pr \{ W_\rho(t; R) \geq \sigma \} \leq f(\sigma),~~~~ \forall t \geq 0,~ \forall \sigma \geq 0,  \label{eq:gSBB_defn}
\end{align}
where $\mathcal{BF}$ denotes the family of positive non-increasing real-valued functions and
$W_\rho(t; R)$ is the virtual workload at time~$t$ of a infinite-buffer 
FCFS (First Come First Served) queue with constant service rate $\rho$ with
input traffic $R$.  The virtual workload is given by
\begin{align}
W_\rho(t; R) = \max_{0\leq s\leq t} [  A(s, t; R) - \rho(t-s) ] .
\label{eq:W_rho}
\end{align}
The gSBB concept is based on Stochastically Bounded Burstiness (SBB)~\cite{Starobinski2000a}, 
which is a direct generalization of EBB.  The gSBB concept has two main advantages over 
SBB:  1)  The class of bounding functions $\mathcal{BF}$ for gSBB is less restrictive; 2) The gSBB is defined in terms of the virtual workload of a queue with constant 
service rate.  The second item is central to the development of our proposed 
stochastic traffic regulator.  

Analogous to the deterministic network calculus, a stochastic network calculus can be developed
based on a given a stochastic traffic burstiness bound~\cite{Yaron1994,Starobinski2000a,Chang2000},  By applying results from the stochastic network calculus based on gSBB (see \cite{JIANG20092011}), the admissibility of a given set of traffic flows with respect to a certain probabilistic end-to-end delay constraint can be determined. 
More general stochastic traffic bounds have since been developed in conjunction with notions of 
statistical arrival envelopes, service curves, and min-plus algebra in the context of stochastic network calculus~\cite{Fidler:2015}.  However, end-to-end delay guarantees via
stochastic network calculus can only be provided if the user traffic streams that offered as input to the network conform to their negotiated traffic burstiness bounds. The stochastic traffic regulator developed in this paper can be applied at the network edge to ensure that a user's traffic stream does not violate the traffic parameter provided
to the admission control unit.  Additional performance benefits can be obtained by applying stochastic traffic regulation in internal network elements.  

\section{Analysis of Deterministic $(\sigma, \rho)$ Regulator}
\label{sec:deterministic_regulator}

In Section~\ref{subsec:IO_analysis} we review results from~\cite{Cruz1991a} and 
then, in Section~\ref{subsec:workload_analysis},
we develop some new results for the $(\sigma, \rho)$ regulator, which we 
shall use in the design and analysis
of the stochastic $(\sigma^*, \rho)$ regulator in Section~\ref{sec:Stochastic_Regulator}.

\subsection{Input/Output Workload Analysis}
\label{subsec:IO_analysis}

Suppose a traffic stream $R$ is offered to an infinite-buffer FCFS system with constant service rate $\rho$.
Clearly, the virtual workload $W_\rho(t; R)$ is a decreasing function of $\rho$.
It can easily be shown that $R \sim (\sigma,\rho)$  if and only if  
\begin{equation}
 W_\rho(t; R)  \leq \sigma,~~~\forall t \geq 0.
 \label{eq:W_rho repeat}
\end{equation}
Equation~\eqref{eq:W_rho repeat} provides a useful alternative characterization
of a $(\sigma, \rho)$-bounded traffic stream.

Now suppose that the input and output traffic links to and
from a $(\sigma,\rho)$ regulator have a finite capacity $C > \rho$.
Consider an input traffic stream $R_{\rm i}$ to the regulator.
Let $s_j$ denote the arrival time of the $j$th packet, 
$t_j$ its departure time, and $L_j$ its length in bits. 
The $j$th packet begins arriving at time $s_j$ and is received completely
at the regulator at time $a_j := s_j + L_j/C$. 
We assume that a packet does not arrive when the previous one is being received.
i.e., $a_j < s_{j+1}$.  

The operation of the regulator
can be described in terms of the workload $W_\rho(s_j;R_{\rm i})$.
At time $s_j$, if $W_\rho(s_j;R_{\rm i})>\sigma$, the regulator delays the packet such that
at its departure time $t_j$, the condition $W_\rho(t_j;R_{\rm o})\leq\sigma$ holds. 
Hence, the departure time of the $j$th packet is derived as~\cite{Cruz1991a}
\begin{align}
\label{eq:delay_regulator}
t_j = [ W_\rho(s_j;R_{\rm i})-\sigma]^+ / \rho+s_j ,
\end{align}
where $[x]^+:=\max\{x,0\}$.   The packet completely departs the regulator
at time 
\begin{align}
    b_j = t_j + L_j/C. \label{eq:bj}
\end{align}
At times other than departures, the workload may not
necessarily be bounded by $\sigma$, but always satisfies~\cite{Cruz1991a} 
\begin{equation}
\label{eq: sigma rho regulator workload bound}
    W_\rho(t;R_{\rm o})\leq\sigma+(1-\rho/C)L_{\max}, ~~~~ \forall t \geq 0,
\end{equation}
Thus, $R_{\rm o} \sim (\sigma+\delta, \rho)$,
where $\delta$, given by \eqref{eq:delta},
can be viewed as  the maximum
error margin in regulating packetized traffic
when the input/output links have capacity $C$ (see Fig.~\ref{fig:Deterministic Regulator}). 

As shown Fig.~\ref{fig:deterministic regulator delay and workload}, when a packet is being
received by the regulator, e.g., during $[s_j, a_j]$, the workload $W_\rho(t; R_{\rm i})$
increases linearly with slope $C-\rho$.  Conversely, during the time between
the complete arrival of a packet and the initial arrival of the next packet to the system,
e.g., during $[a_j, s_{j+1}]$, the workload $W_\rho(t; R_{\rm i})$ decreases linearly with slope $-\rho$.
Similarly, when a packet departs the regulator, e.g., during $[t_j, b_j]$,
the workload $W_\rho(t; R_{\rm o})$
increases linearly with slope $C-\rho$.  When packets are not departing the system,
e.g., during $[b_j, t_{j+1}]$, $W_\rho(t; R_{\rm o})$ decreases linearly with slope $-\rho$.
Assume that the buffer of the regulator is empty at $t=s_1$.
Let 
\begin{align}
\delta_j := (1-\rho/C)L_j
\label{eq:delta_j}
\end{align}
denote the error margin due to regulating the $j$th packet.
We present the governing equations for a $(\sigma,\rho)$  regulator 
in terms of the workloads $W_\rho(t; R_{\rm i})$ and $W_\rho(t; R_{\rm o})$ as follows: 
\begin{align}
W_\rho(t;R_{\rm i}) &= [ W_\rho(a_{j-1};R_{\rm i})
-\rho(t-a_{j-1}) ] ^+, \nonumber\\
& \forall t\in[a_{j-1},s_j],
 \label{eq:Regulator development eq 1}
 \\
 W_\rho(t;R_{\rm i}) & = W_\rho(s_j;R_{\rm i})+(t-s_j)(C-\rho),
 \forall t\in[s_j,a_j],
 \label{eq:Regulator development eq 2}
 \\
 W_\rho(t_j;R_{\rm o}) &= \left\{\begin{array}{ll}
      \sigma , & \text{~if~~} W_\rho(s_j;R_{\rm i})>\sigma ,  \\
      W_\rho(s_j;R_{\rm i}),  & \text{~if~~} W_\rho(s_j;R_{\rm i})\leq \sigma , 
 \end{array}\right. 
 \label{eq:Regulator development eq 3}
 \\
 W_\rho(t;R_{\rm o}) &= W_\rho(t_j;R_{\rm o})+(t-t_j)(C-\rho),
 \forall t\in[t_j,b_j],
 \label{eq:Regulator development eq 4}
 \\
 W_\rho(t;R_{\rm o}) &=W_\rho(b_{j-1};R_{\rm i})-\rho(t-b_{j-1}),\nonumber\\
 &\text{if~~} W_\rho(s_j;R_{\rm i})>\sigma,~~ \forall t\in[b_{j-1},t_j],
 \label{eq:Regulator development eq 5}
\end{align}
for $j=1,2,\ldots$. 

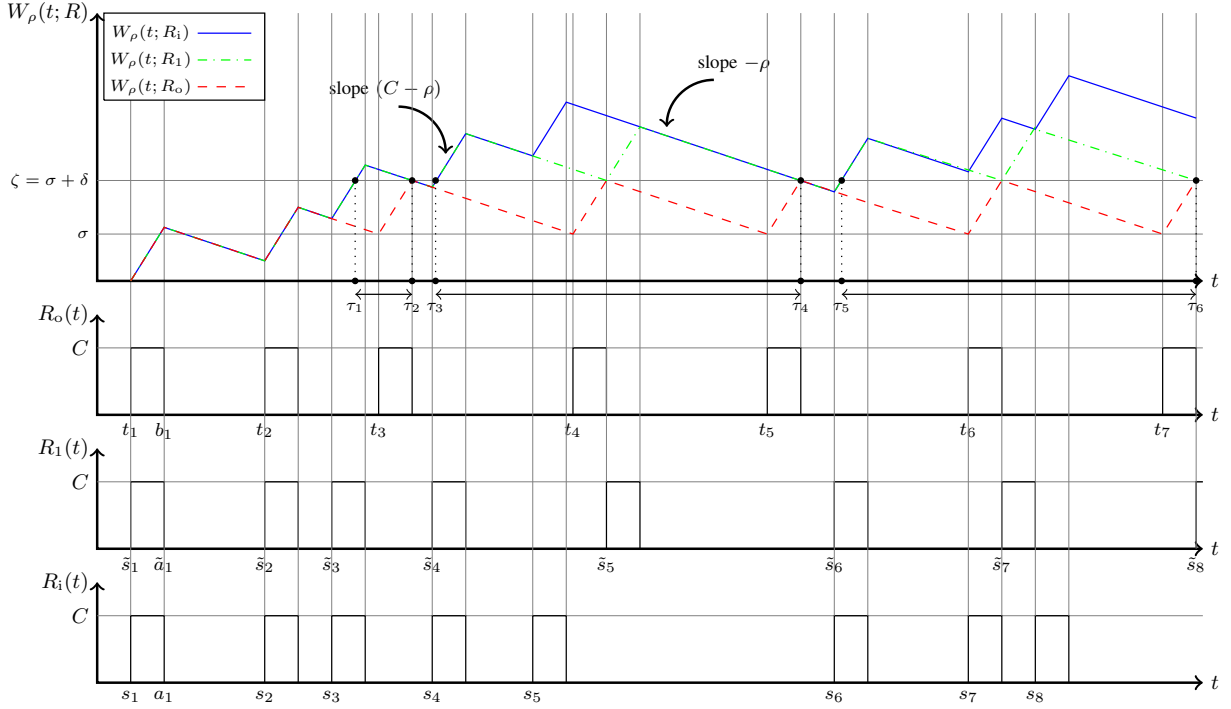
\begin{figure*}
   \centering
    \begin{adjustbox}{width=0.9\textwidth, center=\textwidth}
    \begin{tikzpicture}[scale=0.75, every node/.style={scale=0.75}]
    \draw [thick,<->] (-3.5,-0.5)--(-3.5,-2)--(13,-2);
    \node [right] at (13,-2){\begin{small}$t$\end{small}};
    \draw [help lines] (-3.5,-1)--(13,-1);
    \node [below,left] at (-3.5,-1) {\begin{small}$C$\end{small}};
    \node [above,left] at (-3.5,-0.5) {\begin{small}$R_{\rm i}(t)$\end{small}};
    \draw (-3,-2)--(-3,-1)--(-2.5,-1)--(-2.5,-2);
    \node [below] at (-3,-2) {\begin{footnotesize}$s_1$\end{footnotesize}};
    \node [below] at (-2.5,-2) {\begin{footnotesize}$a_1$\end{footnotesize}};
    \draw (-1,-2)--(-1,-1)--(-0.5,-1)--(-0.5,-2);
    \node [below] at (-1,-2) {\begin{footnotesize}$s_2$\end{footnotesize}};
    \draw (0,-2)--(0,-1)--(0.5,-1)--(0.5,-2);
    \node [below] at (0,-2) {\begin{footnotesize}$s_3$\end{footnotesize}};
    \draw (1.5,-2)--(1.5,-1)--(2,-1)--(2,-2);
    \node [below] at (1.5,-2) {\begin{footnotesize}$s_4$\end{footnotesize}};
    \draw (3,-2)--(3,-1)--(3.5,-1)--(3.5,-2);
    \node [below] at (3,-2) {\begin{footnotesize}$s_5$\end{footnotesize}};
    \draw (7.5,-2)--(7.5,-1)--(8,-1)--(8,-2);
    \node [below] at (7.5,-2) {\begin{footnotesize}$s_6$\end{footnotesize}};
    \draw (9.5,-2)--(9.5,-1)--(10,-1)--(10,-2);
    \node [below] at (9.5,-2) {\begin{footnotesize}$s_7$\end{footnotesize}};
    \draw (10.5,-2)--(10.5,-1)--(11,-1)--(11,-2);
    \node [below] at (10.5,-2) {\begin{footnotesize}$s_8$\end{footnotesize}};
    \draw [thick,<->] (-3.5,1.5)--(-3.5,0)--(13,0);
    \node [right] at (13,0){\begin{small}$t$\end{small}};
    \draw [help lines] (-3.5,1)--(13,1);
    \node [below,left] at (-3.5,1) {\begin{small}$C$\end{small}};
    \node [above,left] at (-3.5,1.5) {\begin{small}$R_1(t)$\end{small}};
    \draw (-3,0)--(-3,1)--(-2.5,1)--(-2.5,0);
    \node [below] at (-3,0) {\begin{footnotesize}$\Tilde{s}_1$\end{footnotesize}};
    \node [below] at (-2.5,0) {\begin{footnotesize}$\tilde{a}_1$\end{footnotesize}};
    \draw (-1,0)--(-1,1)--(-0.5,1)--(-0.5,0);
    \node [below] at (-1,0) {\begin{footnotesize}$\Tilde{s}_2$\end{footnotesize}};
    \draw (0,0)--(0,1)--(0.5,1)--(0.5,0);
    \node [below] at (0,0) {\begin{footnotesize}$\Tilde{s}_3$\end{footnotesize}};
    \draw (1.5,0)--(1.5,1)--(2,1)--(2,0);
    \node [below] at (1.5,0) {\begin{footnotesize}$\Tilde{s}_4$\end{footnotesize}};
    \draw (4.1,0)--(4.1,1)--(4.6,1)--(4.6,0);
    \node [below] at (4.1,0) {\begin{footnotesize}$\Tilde{s}_5$\end{footnotesize}};
    \draw (7.5,0)--(7.5,1)--(8,1)--(8,0);
    \node [below] at (7.5,0) {\begin{footnotesize}$\Tilde{s}_6$\end{footnotesize}};
    \draw (10,0)--(10,1)--(10.5,1)--(10.5,0);
    \node [below] at (10,0) {\begin{footnotesize}$\Tilde{s}_7$\end{footnotesize}};
    \draw (12.9,0)--(12.9,1)--(13,1);
    \node [below] at (12.9,0) {\begin{footnotesize}$\Tilde{s}_8$\end{footnotesize}};
    \draw [thick,<->] (-3.5,3.5)--(-3.5,2)--(13,2);
    \node [right] at (13,2){\begin{small}$t$\end{small}};
    \draw [help lines] (-3.5,3)--(13,3);
    \node [below,left] at (-3.5,3) {\begin{small}$C$\end{small}};
    \node [above,left] at (-3.5,3.5) {\begin{small}$R_{\rm o}(t)$\end{small}};
    \draw (-3,2)--(-3,3)--(-2.5,3)--(-2.5,2);
    \node [left,below] at (-3,2) {\begin{footnotesize}$t_1$\end{footnotesize}};
    \node [below] at (-2.5,2) {\begin{footnotesize}$b_1$\end{footnotesize}};
    \draw (-1,2)--(-1,3)--(-0.5,3)--(-0.5,2);
    \node [left,below] at (-1,2) {\begin{footnotesize}$t_2$\end{footnotesize}};
    \draw (0.7,2)--(0.7,3)--(1.2,3)--(1.2,2);
    \node [below] at (0.7,2) {\begin{footnotesize}$t_3$\end{footnotesize}};
    \draw (3.6,2)--(3.6,3)--(4.1,3)--(4.1,2);
    \node [left,below] at (3.6,2) {\begin{footnotesize}$t_4$\end{footnotesize}};
    \draw (6.5,2)--(6.5,3)--(7,3)--(7,2);
    \node [left,below] at (6.5,2) {\begin{footnotesize}$t_5$\end{footnotesize}};
    \draw (9.5,2)--(9.5,3)--(10,3)--(10,2);
    \node [left,below] at (9.5,2) {\begin{footnotesize}$t_6$\end{footnotesize}};
    \draw (12.4,2)--(12.4,3)--(12.9,3)--(12.9,2);
    \node [left,below] at (12.4,2) {\begin{footnotesize}$t_7$\end{footnotesize}};
   \draw [thick,<-] (1.7,5.9) arc [start angle=0, end angle=90, radius=20pt];
    \node [above] at (0.8,6.6) {\begin{footnotesize}slope $(C-\rho)$\end{footnotesize}};
    \draw [thick,<-] (5,6.3) arc [start angle=180, end angle=90, radius=20pt];
    \node [above] at (6,7) {\begin{footnotesize}slope $-\rho$\end{footnotesize}};
    \draw [help lines] (-3,-1)--(-3,6.7);
    \draw [help lines] (-2.5,-1)--(-2.5,6.7);
    \draw [help lines] (-1,-1)--(-1,8);
    \draw [help lines] (-0.5,-1)--(-0.5,8);
    \draw [help lines] (0,-1)--(0,8);
    \draw [help lines] (0.5,-1)--(0.5,8);
    \draw [help lines] (0.7,3)--(0.7,8);
    \draw [help lines] (1.2,3)--(1.2,8);
    \draw [help lines] (1.5,-1)--(1.5,8);
    \draw [help lines] (2,-1)--(2,8);
    \draw [help lines] (3,-1)--(3,8);
    \draw [help lines] (3.5,-1)--(3.5,8);
    \draw [help lines] (3.6,3)--(3.6,8);
    \draw [help lines] (4.1,1)--(4.1,8);
    \draw [help lines] (4.6,1)--(4.6,8);
    \draw [help lines] (6.5,3)--(6.5,8);
    \draw [help lines] (7,3)--(7,8);
    \draw [help lines] (7.5,-1)--(7.5,8);
    \draw [help lines] (8,-1)--(8,8);
    \draw [help lines] (9.5,-1)--(9.5,8);
    \draw [help lines] (10,-1)--(10,8);
    \draw [help lines] (10.5,-1)--(10.5,8);
    \draw [help lines] (11,-1)--(11,8);
    \draw [help lines] (12.4,3)--(12.4,8);
    \draw [help lines] (12.9,0)--(12.9,8);
    \draw [thick,<->] (-3.5,8)--(-3.5,4)--(13,4);
    \node [right] at (13,4){\begin{small}$t$\end{small}};
    \node [left] at (-3.5,8) {\begin{small}$W_\rho(t;R)$\end{small}};
    \draw [help lines] (-3.5,4.7)--(13,4.7);
    \node [left] at (-3.5,4.7){\begin{scriptsize}$\sigma$\end{scriptsize}};
    \draw [help lines] (-3.5,5.5)--(13,5.5);
    \node [left] at (-3.5,5.5){\begin{scriptsize}$\zeta=\sigma+\delta$\end{scriptsize}};
    \draw [blue] (-3,4)--(-2.5,4.8)--(-1,4.3)--(-0.5,5.1)--(0,4.93)--(0.5,5.73)--(1.2,5.5)--(1.5,5.4)--(2,6.2)--(3,5.867)--(3.5,6.67)--(7.5,5.33)--(8,6.13)--(9.5,5.63)--(10,6.43)--(10.5,6.264)--(11,7.064)--(12.9,6.432);
    \draw [green,dashdotted] (-3,4)--(-2.5,4.8)--(-1,4.3)--(-0.5,5.1)--(0,4.93)--(0.5,5.73)--(1.2,5.5)--(1.5,5.4)--(2,6.2)--(4.1,5.5)--(4.6,6.3)--(7.5,5.33)--(8,6.13)--(10,5.5)--(10.5,6.3)--(10.5,6.264)--(12.9,5.5);
    \draw [red,dashed](-3,4)--(-2.5,4.8)--(-1,4.3)--(-0.5,5.1)--(0,4.93)--(0.7,4.7)--(1.2,5.5)--(3.6,4.7)--(4.1,5.5)--(6.5,4.7)--(7,5.5)--(8,5.17)--(9.5,4.7)--(10.0,5.5)--(12.4,4.7)--(12.9,5.5);
    \draw (-3.4,6.7)--(-1,6.7)--(-1,8)--(-3.4,8)--(-3.4,6.7);
    \draw [blue](-2,7.7)--(-1.2,7.7);
    \draw [green,dashdotted](-2,7.3)--(-1.2,7.3);
    \draw [red,dashed] (-2,6.9)--(-1.2,6.9);
    \node at (-2.7,7.7){\begin{scriptsize}{$W_\rho(t;R_{\rm i})$}\end{scriptsize}};
   \node at (-2.7,7.3){\begin{scriptsize}{$W_\rho(t;R_1)$}\end{scriptsize}};
   \node at (-2.7,6.9){\begin{scriptsize}{$W_\rho(t;R_{\rm o})$}\end{scriptsize}};
    \draw [fill](0.35,5.5) circle [radius=0.04];
    \draw [fill](1.2,5.5) circle [radius=0.04];
    \draw [fill](1.55,5.5) circle [radius=0.04];
    \draw [fill](7,5.5) circle [radius=0.04];
    \draw [fill](7.61,5.5) circle [radius=0.04];
    \draw [fill](12.9,5.5) circle [radius=0.04];
    \draw [dotted] (0.35,5.5)--(0.35,4);
    \draw [dotted] (1.2,5.5)--(1.2,4);
    \draw [dotted] (1.55,5.5)--(1.55,4);
    \draw [dotted] (7,5.5)--(7,4);
    \draw [dotted] (7.61,5.5)--(7.61,4);
    \draw [dotted] (12.9,5.5)--(12.9,4);
    \draw [<->] (0.35,3.8)--(1.2,3.8);
   \draw [<->] (1.55,3.8)--(7,3.8);
    \draw [<->] (7.61,3.8)--(12.9,3.8);
    \node [below] at (0.35,3.8){\begin{scriptsize}{$\tau_1$}\end{scriptsize}};
    \node [below] at (1.2,3.8){\begin{scriptsize}{$\tau_2$}\end{scriptsize}};
    \node [below] at (1.55,3.8){\begin{scriptsize}{$\tau_3$}\end{scriptsize}};
    \node [below] at (7,3.8){\begin{scriptsize}{$\tau_4$}\end{scriptsize}};
    \node [below] at (7.61,3.8){\begin{scriptsize}{$\tau_5$}\end{scriptsize}};
    \node [below] at (12.9,3.8){\begin{scriptsize}{$\tau_6$}\end{scriptsize}};
    \draw [fill](0.35,4) circle [radius=0.04];
    \draw [fill](1.2,4) circle [radius=0.04];
    \draw [fill](1.55,4) circle [radius=0.04];
    \draw [fill](7,4) circle [radius=0.04];
    \draw [fill](7.61,4) circle [radius=0.04];
    \draw [fill](12.9,4) circle [radius=0.04];
    \end{tikzpicture}
    \end{adjustbox}
    \caption{Example of the operation of a $(\sigma,\rho)$ traffic regulator.}
    \label{fig:deterministic regulator delay and workload}
\end{figure*}
Equations~\eqref{eq:Regulator development eq 1}--\eqref{eq:Regulator development eq 5} provide a complete characterization of the virtual workloads of the traffic streams $R_{\rm i}$ and $R_{\rm o}$ and
can be used to construct the corresponding workload
curves in Fig.~\ref{fig:deterministic regulator delay and workload}.

\subsection{Internal Traffic Workload Analysis}
\label{subsec:workload_analysis}

To analyze the stochastic $(\sigma^*, \rho)$ regulator developed in Section~\ref{sec:Stochastic_Regulator}, it will be convenient
to introduce the internal traffic stream $R_1$ shown in Fig.~\ref{fig:deterministic_shaper}
for the $(\sigma, \rho)$ regulator and in Fig.~\ref{fig:stochastic traffic regulator schematic}
for the $(\sigma^*, \rho)$ regulator.  We shall develop some new results
for the $(\sigma, \rho)$ regulator involving the internal stream $R_1$, which
will be useful in the design of the $(\sigma^*, \rho)$ regulator.
Fig.~\ref{fig:deterministic_shaper} can be viewed
as a more detailed depiction of the $(\sigma, \rho)$ regulator shown as a single
box in Fig.~\ref{fig:Deterministic Regulator}. The diagrams 
in Figs.~\ref{fig:deterministic_shaper} and \ref{fig:stochastic traffic regulator schematic}
represent single-server, infinite buffer queueing systems.  The box 
represents the server, which imposes a variable service delay on an arriving packet.  The service delay will be zero
if no shaping is needed.  Only one packet can reside in the server at any given time.  A new
packet~$j$ can arrive to the server at the instant packet~$j-1$ leaves the server.
Packets that arrive when the server is occupied are stored in the front-end buffer in FCFS
order.  The traffic stream $R_1$ consists of the stream of packets arriving to the server.

Let $\tilde{s}_j$ denote the arrival time of the $j$th packet at the buffer and
let $\tilde{a}_j$ denote the complete arrival time to the buffer, i.e.,
$\ta_j := \ts_j + L_j/C$. The server
incurs a delay on the $j$th packet such that it begins departing the buffer at time~$t_j$ and 
completely leaves the regulator at time $b_j$. Since the front-end buffer delays each packet until 
the complete departure time of the previous packet from the regulator, we have
\begin{align}
\label{eq:sTildej}
 \Tilde{s}_j=\max\lbrace s_j,b_{j-1}\rbrace.
\end{align}
Therefore, the operation of $(\sigma, \rho)$ regulator
can also be described in terms of the workload $W_\rho(\tilde{s}_j;R_1)$. In other words, we have the following theorem which is proved in Appendix.
\begin{proposition}
\label{prop:deterministic delay w buffer}
The departure time $t_j$ for the $j$th packet in the $(\sigma, \rho)$ regulator is given by (cf.~\eqref{eq:delay_regulator}): 
\begin{equation}
    t_j=[W_\rho(\tilde{s}_j;R_1)-\sigma]^+/\rho+\tilde{s}_j. 
    \label{eq:tj_R1}
\end{equation}
\end{proposition}

An example sample path of the workloads of traffic streams $R_{\rm i}$, $R_1$, and $R_{\rm o}$
for a deterministic $(\sigma, \rho)$ regulator is shown in 
the top graph of Fig.~\ref{fig:deterministic regulator delay and workload}.
If the input traffic stream $R_{\rm i}$ conforms to the $(\sigma, \rho)$ traffic
burstiness parameter at arrival times, then the workloads of $R_{\rm i}$, $R_1$, and $R_{\rm o}$ will all
coincide, which occurs in the interval $[s_1, s_3]$ in 
the figure.  Within this interval, for packets
$j=1$ and $2$, we have $s_j = \ts_j = t_j$ and
$a_j = \ta_j = b_j$, since both packets arrive when
the workload $W_\rho(t; R_{\rm i}) \leq \sigma$.  
At time $s_3 = \ts_3$, the workloads of $R_1$ and $R_{\rm o}$ diverge because 
packet~$3$ arrives when $W_\rho(t; R_{\rm i}) > \sigma$.
Thus, the packet is delayed in the server and $t_3 > \ts_3$.  
However, the workloads of $R_1$ and $R_{\rm o}$ once again
coincide at time $b_3$, i.e., the complete departure time of packet~$3$
from the regulator.  

The workload curves of $R_1$ and $R_{\rm o}$ form a parallelogram in the
interval $[\ts_3, b_3]$.  The other points of this parallelogram
occur at $\ta_3$, i.e., when packet~$3$ completely
arrives to the server and at $t_3$, i.e., when packet~$3$ starts
to depart the server.  Then the two workload curves coincide
in the interval $[b_3, \ts_4]$.  In general, the workloads of $R_1$ and $R_{\rm o}$
form a (possibly degenerate) parallelogram during
the interval $[\ts_j, b_j]$ and coincide during the interval $[b_j, \ts_{j+1}]$,
for $j=1, 2, \ldots$.  

In Fig.~\ref{fig:deterministic regulator delay and workload},
we see that  the workload curves of $R_{\rm i}$ and $R_1$ coincide until
time $s_5$, which is the start time of the arrival of packet~$5$ to the regulator.
At this time, packet~$4$ is at the server, so packet~$5$ waits until time
time~$\ts_5 > s_5$ to go into service.  At time~$\ta_5$, when packet~$5$ has arrived
completely to the server, the two curves coincide once again.  In the interval
$[s_5, \ta_5]$, the two curves form a parallelogram.  This is not true in
general, but in the interval $[s_j, \ta_j]$ a (possibly degenerate) parallelogram can be formed
in which the sides consists of
$W_\rho(t;R_{\rm i})$ for $t \in [s_j, a_j]$, $W_\rho(t;R_1)$
for $t \in [\ts_j, \ta_j]$, $W_\rho(t;R_{\rm i})$ for $t \in [a_j, \ta_j]$, and $W_\rho(t;R_1)$ for $t \in [s_j, \ts_j]$ for $j=1, 2, \ldots$. Thus, the workload curves of $R_{\rm i}$ and $R_1$ are separated by
a {\em sequence} of possibly degenerate parallelograms.  Each such parallelogram 
corresponds to a packet delayed in the buffer of the regulator.
A similar type of relationship holds between the workload
curves of $R_{\rm i}$ and $R_{\rm o}$.  The workload curves of $R_1$ and $R_{\rm o}$
are separated by at most {\em one} parallelogram because the server can hold at most
one packet.

Based on the above analysis
and Proposition~\ref{prop:deterministic delay w buffer}, the  operation of the $(\sigma, \rho)$ regulator can
characterized in terms of the internal traffic
stream $R_1$ and the output traffic stream $R_{\rm o}$.
Analogous to equations \eqref{eq:Regulator development eq 1}--\eqref{eq:Regulator development eq 5}
the following equations involving $R_1$ can be derived:
\begin{align}
&W_\rho(t;R_{\rm o})= W_\rho(t;R_1)  =[W_\rho(b_{j-1};R_{\rm o})-\rho(t-b_{j-1})]^+, \nonumber
\\
&\quad\qquad\qquad\forall t\in [b_{j-1},\Tilde{s}_j],
 \label{eq:Stochastic Regulator eq 1}
 \\
 &W_\rho(t;R_{\rm 1})  = W_\rho(\Tilde{s}_j;R_{\rm 1})+(t-\tilde{s}_j)(C-\rho),
 ~\forall t\in[\Tilde{s}_j,\Tilde{a}_j],
 \label{eq:Stochastic Regulator eq 2}
 \\
 &W_\rho(t;R_{\rm 1}) =  W_\rho(\tilde{a}_{j};R_{\rm 1})
-\rho(t-\tilde{a}_{j}) , ~\forall t\in[\tilde{a}_{j},b_j]
 \label{eq:Stochastic Regulator eq 3}
 \\
 &W_\rho(t_j;R_{\rm o})  =\left\{\begin{array}{ll}
      \sigma ,  & \text{~if~~} W_\rho(\Tilde{s}_j;R_1)>\sigma, \\
      W_\rho(\tilde{s}_j;R_1),  & \text{~if~~} W_\rho(\Tilde{s}_j;R_1)\leq \sigma ,
  \end{array}\right.
 \label{eq:Stochastic Regulator eq 4}
 \\
 &W_\rho(t;R_{\rm o}) =W_\rho(t_j;R_{\rm o})+(t-t_j)(C-\rho),
  ~\forall t\in[t_j,b_j],
 \label{eq:Stochastic Regulator eq 5}
 \\
 &W_\rho(t;R_{\rm o}) =W_\rho(\Tilde{s}_j;R_{\rm 1})-\rho(t-\Tilde{s}_j),\nonumber
 \\
 &\quad\qquad\qquad\text{~if~~} W_\rho(\Tilde{s}_j;R_{\rm 1})>\sigma,~~ \forall t\in[\Tilde{s}_j,t_j],
 \label{eq:Stochastic Regulator eq 6}
 \end{align}
 for $j=1,2,\ldots$. Equation~\eqref{eq:Stochastic Regulator eq 1} follows
 from the following equality 
 \begin{equation}
 W_\rho(b_{j-1};R_{\rm o})=W_\rho(b_{j-1};R_1),
 \label{eq:W_R1eqW_Ro}
 \end{equation}
 which can be verified using~\eqref{eq:Stochastic Regulator eq 2}-\eqref{eq:Stochastic Regulator eq 6} and~\eqref{eq:bj}. Intuitively, \eqref{eq:W_R1eqW_Ro} holds because
 at most one packet is in the server of the regulator at any given time.
 
\section{Stochastic $(\sigma^*, \rho)$ Regulator}
\label{sec:Stochastic_Regulator}

\begin{figure}
    \centering
    \begin{adjustbox}{width=0.8\columnwidth}
   \begin{tikzpicture}
    \draw [thick,->](-2.75,0) -- (-1.5,0);
   \draw [thick,->](-0.5,0) -- (0.5,0);
    \draw [thick,->](4,0) -- (5.5,0);
    \draw (-2,.5) -- (-.5,.5) -- (-.5,-.5) -- (-2,-.5);
    \draw [dashed] (-1.5,.5)--(-1.5,-.5);
    \draw [dashed] (-1.25,.5)--(-1.25,-.5);
    \draw [dashed] (-1,.5)--(-1,-.5); 
    \draw [dashed] (-0.75,.5)--(-0.75,-.5);
    \draw (0.5,-2.5) -- (0.5,1) -- (4,1) -- (4,-2.5) -- (0.5,-2.5);
    \node at (2.25,-0.5){\begin{footnotesize}$
     \sigma^*(j)=\max \left \{ \sigma \in \Sigma \right \} $\end{footnotesize}};
    \node at (1.5,-1){\footnotesize{such that:}};
    \node at (2.25,-1.5){\begin{footnotesize}$\Pr\left\{W_{\rho}(t;R_{\rm o})\geq\gamma\right\} \leq f(\gamma),$\end{footnotesize}};
    \node at (2.25,-2){\begin{footnotesize}$\forall~t\geq 0,~~\forall \gamma\in[0,T]$. \end{footnotesize}};
    \node at (2.25,0.5){$(\sigma^*,\rho)$};
    \node [above] at (-2.5,0) {\begin{footnotesize}$R_{\rm i}$\end{footnotesize}};
    \node [above] at (0,0) {\begin{footnotesize}$R_1$\end{footnotesize}};
    \node [above] at (5,0) {\begin{small}$R_{\rm o}$\end{small}};
    \end{tikzpicture}
   \end{adjustbox}
    \caption{Idealized stochastic $(\sigma^*,\rho)$ traffic regulator.}
    \label{fig:stochastic traffic regulator schematic}
\end{figure}
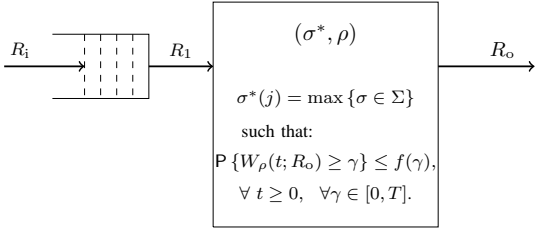

The $(\sigma,\rho)$ regulator may enforce a loose
deterministic bound on the traffic or incur unnecessarily
large delays on the traffic. To address these issues, we propose
a {\em stochastic} traffic regulator with the aim of enforcing
a probabilistic bound on the output traffic:
\begin{align}
    \label{eq: Stochastic (sigma,rho)}
\Pr\left\{W_{\rho}(t;R_{\rm o})\geq\gamma\right\} \leq f(\gamma), ~~~~\forall t \geq 0, ~\forall \gamma\in[0,T],
\end{align}
where $f$ is a non-increasing positive bounding function and $T$
is a limit on the tail distribution of the workload (see \cite{KordiCISS2019}). 
As $T \rightarrow \infty$, \eqref{eq: Stochastic (sigma,rho)} becomes equivalent to gSBB in~\eqref{eq:gSBB_defn}.

\subsection{Operational Principles}

For a stationary and ergodic input traffic stream, 
we shall show that tight enforcement of~\eqref{eq: Stochastic (sigma,rho)} can 
be achieved under steady-state conditions using a regulator with a constant 
rate parameter $\rho$ and a variable burstiness parameter $\sigma^*$ which is chosen from a finite set $\Sigma$
for each arriving packet.  We refer to such a regulator as
a stochastic $(\sigma^*, \rho)$ regulator.  
A schematic of an idealized $(\sigma^*, \rho)$ regulator is shown in Fig.~\ref{fig:stochastic traffic regulator schematic}.  The input and output links of the regulator
are assumed to have capacity $C$.
A buffer at the front-end of the regulator delays incoming packets until all previous packets
have departed, thus ensuring a FCFS service discipline.
Let $R_{\rm i}$ and $R_{\rm o}$ denote, respectively, the input traffic to and output traffic from
the regulator.  We denote the internal traffic stream departing from the front-end buffer as
$R_1$.  Let $s_j$ and $\ts_j$ denote, respectively, the arrival and departure time of the $j$th packet at 
the buffer.  

For each packet~$j$, the $(\sigma^*, \rho)$ regulator chooses a burstiness
parameter~$\sigma^*(j)$ such that a delay $d_j$ is incurred, where (cf.\ \eqref{eq:delay_regulator})
\begin{align}
d_j = t_j - s_j = [W_\rho(s_j; R_{\rm i}) - \sigma^*(j) ]^+ / \rho ,
\label{eq:tj_stochastic}
\end{align} 
and $t_j$ denotes the time at which the packet starts departing the traffic regulator.  The packet completely
leaves the regulator at time $b_j$.
The front-end buffer acts as in the deterministic $(\sigma, \rho)$ regulator (see Section~\ref{sec:deterministic_regulator}); therefore $\Tilde{s}_j$ can be derived from~\eqref{eq:sTildej}.
As in a deterministic $(\sigma,\rho)$ traffic regulator, the rate parameter $\rho$ must be greater than or equal to the long-term average input traffic rate, i.e., 
\begin{equation}
\rho \geq \lim_{t \rightarrow\infty} \frac{1}{t-s} \int_s^t R_{\rm i}(\tau)\;\mathrm{d}\tau , ~~~ \forall s \geq 0,
\end{equation}
to avoid incurring an unbounded packet delay in the long-term.
\subsection{Overshoot Probability and Overshoot Ratio}

\input{figs/Overshoot_Increment}

To design a practical $(\sigma^*, \rho)$ regulator, the {\em overshoot probability}
$\Pr\left\{W_{\rho}(t;R_{\rm o})\geq\gamma\right\}$ in \eqref{eq: Stochastic (sigma,rho)}
can be approximated by a time-averaged {\em overshoot ratio} assuming that the input traffic 
$R_{\rm i}$ is stationary and ergodic. 

\begin{definition}
\label{def:Overshoot Duration}
Given a threshold value $\zeta >0$ and a traffic stream $R$,
an {\em overshoot interval} with respect to $R$ and $\zeta$ is a maximal interval of time $\eta$
such that $W_\rho(\tau; R) \geq \zeta$ for all $ \tau \in \eta$.  Let $|\eta|$ denote the
length of interval $\eta$.  Let $\mathcal{O}(t)$ denote the set of overshoot intervals
contained in $[0, t]$.  Then the {\em overshoot duration} up to time~$t$ is defined as 
\begin{align}
O_{\zeta}(t;R) = \sum_{\eta \in \mathcal{O}(t)} |\eta| .
 \label{eq:total overshoot duration} 
\end{align}
\end{definition}
In Fig.~\ref{fig:deterministic regulator delay and workload}, the overshoot set with respect to threshold value $\zeta$ until the end of time domain depicted in the figure consists of three intervals $[\tau_1,\tau_2]$, $[\tau_3,\tau_4]$ and $[\tau_5,\tau_6]$. 
Given a time interval $[a, b]$, let $W_1 = W_\rho(a; R_{\rm o})$
and $W_2 = W_\rho(b; R_{\rm o})$.
We define the increment in overshoot duration when the workload of
the output process is {\em increasing}
due to a packet departure from the regulator as follows:
\begin{equation}
 \alpha(a,b,\zeta)=\left\{\begin{array}{ll}
b - a   ,              &    \zeta \leq W_1 ,\\
(W_2-\zeta)/(C-\rho) ,   &   W_1\leq  \zeta \leq W_2\\
0   ,                    &   W_2 < \zeta .
 \end{array}\right.   
 \label{eq:alpha_increment}   
\end{equation}
We define the increment in overshoot duration when the workload is {\em decreasing}
due to the packet inter-departure time as follows:
\begin{equation}
 \beta(a,b, \zeta)=\left\{\begin{array}{ll}
b - a ,                &   \zeta \leq W_2 , \\
(W_1- \zeta)/\rho,    &   W_2 \leq \zeta \leq W_1 , \\
0 ,                      &   W_1< \zeta .
 \end{array}\right.   
 \label{eq:beta_increment}   
\end{equation}
Figure~\ref{fig:alpha_function} illustrates
$\alpha(a, b, \zeta)$ and $\beta(a, b, \zeta)$.
The following proposition shows how to compute 
$O_{\zeta}(t;R_{\rm o})$ at time $t=b_j$
for packet~$j$.  
\begin{proposition}
\begin{align*}
O_{\zeta}(b_1;R_{\rm o}) &= \alpha(t_1,b_1,\zeta) \\ 
O_{\zeta}(b_j;R_{\rm o}) &= 
O_{\zeta}(b_{j-1};R_{\rm o}) \! + \!  \beta(b_{j-1},t_j,\zeta) \!  + \! \alpha(t_j,b_j,\zeta),
\end{align*}
for $j=2,3, \ldots$.
\label{prop:overshoot}
\end{proposition}

We define the {\em overshoot ratio} of the regulator at time~$t$ with respect to
a threshold $\zeta$ by
\begin{align}
    o_\zeta(t) = O_{\zeta}(t;R_{\rm o})/ t.
    \label{eq:overshootratio}
\end{align}
If the input traffic $R_{\rm i}$ is stationary and ergodic, the overshoot ratio asymptotically
approaches the overshoot probability, i.e.,
\begin{equation}
o_{\zeta}(t)
\sim \Pr\left\{W_{\rho}(t;R_{\rm o})  > \zeta \right \} ~\mbox{as}~t \rightarrow \infty .     
\label{eq:Prob_approx}
\end{equation}
Using the overshoot probability as a proxy for the
overshoot probability in~\eqref{eq: Stochastic (sigma,rho)}, we design a $(\sigma^*, \rho)$ regulator
that selects the burstiness parameters $\sigma^*(j)$, $j = 1, 2, \ldots$,
form a set $\Sigma$ such that 
\begin{align}
    \label{eq:Stochastic (sigma,rho) for jth packet}
 o_\gamma(t) \leq f(\gamma),~~\forall ~t \in [b_{j-1},b_j],~ \forall \gamma \in [0, T] ,
\end{align}
while minimizing the incurred packet delay.

\subsection{Piecewise-Linear Bounding Function}

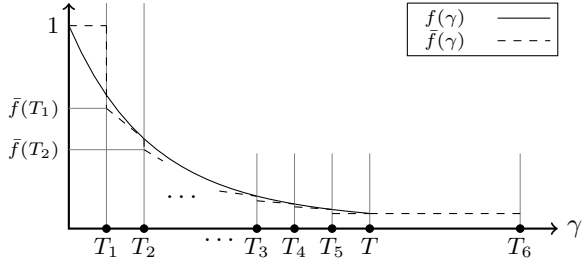
\begin{figure}
\centering
\begin{tikzpicture}
\draw [thick,<->] (6.5,0)--(0,0)--(0,3);
\node [left] at (0,2.7) {\begin{small}$1$\end{small}};
\node [right] at (6.5,0) {$\gamma$};
\draw [domain=0:4] plot (\x,{2.7*(0.8*exp(-1*\x)+0.2*exp(-0.3*\x))});
\draw [help lines] (6,0) -- (6,1);
\draw [help lines] (4,0) -- (4,1);
\draw [help lines] (3.5,0) -- (3.5,1);
\draw [help lines] (3,0) -- (3,1);
\draw [help lines] (2.5,0) -- (2.5,1);
\draw [help lines] (1,0) -- (1,3);
\draw [help lines] (0.5,0) -- (0.5,3);
\draw [fill](6,0) circle [radius=0.05];
\draw [fill](4,0) circle [radius=0.05];
\draw [fill](3.5,0) circle [radius=0.05];
\draw [fill](3,0) circle [radius=0.05];
\draw [fill](2.5,0) circle [radius=0.05];
\draw [fill](1,0) circle [radius=0.05];
\draw [fill](0.5,0) circle [radius=0.05];
\node [below] at (6,0){\begin{small}{$T_6$}\end{small}};
\node [below] at (4,0){\begin{small}{$T$}\end{small}};
\node [below] at (3.5,0){\begin{small}{$T_5$}\end{small}};
\node [below] at (3,0){\begin{small}{$T_4$}\end{small}};
\node [below] at (2.5,0){\begin{small}{$T_3$}\end{small}};
\node [below] at (1,0){\begin{small}{$T_{2}$}\end{small}};
\node [below] at (0.5,0){\begin{small}{$T_{1}$}\end{small}};
\node [below] at (2,0) {$\ldots$};
\draw [dashed] (0,2.7)--(0.5,2.7);
\draw [dashed] (0.5,2.7)--(0.5,1.6);
\draw [dashed] (0.5,1.6)--(1,1.195);
\draw [dashed] (1,1.195)--(1,1.05);
\draw [dashed] (1,1.05)--(1.25,0.9);
\draw [dashed] (2,0.5)--(2.5,0.43);
\draw [dashed] (2.5,0.37)--(3,0.327);
\draw [dashed] (3,0.29)--(3.5,0.254);
\draw [dashed] (3.5,0.20)--(4,0.2);
\draw [dashed] (4,0.2)--(6,0.2);
\draw [help lines] (0,1.6)--(0.5,1.6);
\node [left] at(0,1.6) {\scriptsize{$\bar{f}(T_1)$}};
\draw [help lines] (0,1.05)--(1,1.05);
\node [left] at(0,1.05) {\scriptsize{$\bar{f}(T_2)$}};
\node at (1.5,0.43) {$\ldots$};
\draw (4.5,3)--(6.5,3)--(6.5,2.3)--(4.5,2.3)--(4.5,3);
\draw (5.7,2.8)--(6.4,2.8);
\node at (5,2.8) {\begin{scriptsize}$f(\gamma)$\end{scriptsize}};
\draw [dashed] (5.7,2.5)--(6.4,2.5);
\node at (5,2.5) {\begin{scriptsize}$\bar{f}(\gamma)$ \end{scriptsize}};
\end{tikzpicture}
\caption{Piecewise linear approximating function for $f(\gamma)$, $M=6$.}
\label{fig:Regulator Bounding f}
\end{figure}

Next, we address the issues of selecting the set $\Sigma$ of burstiness parameter values
and verification of the condition~\eqref{eq:Stochastic (sigma,rho) for jth packet}. We replace the bounding function $f$
by a piecewise-linear function $\bar{f}$ defined
in terms of a set of values $T_1 < T_2 < \ldots < T_M$ and 
the value $\delta$ given by \eqref{eq:delta}
satisfying the following constraints:
\begin{align}
T-T_{M-1} \geq \delta; ~T_{M} \gg T , ~T_1 \geq \delta, 
~T_{i+1}-T_{i}\geq \delta,
    \label{eq:T_i}
\end{align}
for $i=1,2,\ldots,M-1$. For given $T$ and $\delta$, the maximum possible value
of $M$ is given by
\begin{align}
    M_{\max} = \lfloor T / \delta \rfloor - 1 .
\label{eq:M_max_value}
\end{align}
The values $\{ T_1, \ldots, T_M \}$
determine the set of burstiness parameter values
\begin{align}
    \Sigma = \{ \sigma_i := T_i - \delta: i = 1, \ldots, M \} .
\label{eq:sigma_i}
\end{align}
Note that $\sigma_1 < \sigma_2 < \ldots < \sigma_M$.

Without loss of generality, we assume $f(0) = 1$.
The function $\bar{f}$ is designed to be a close
lower bound to $f$ in the interval $[T_1, T]$
and an upper bound to $f$ in the interval $[0, T_1)$.
For the definition of $\bar{f}$ presented here,
we shall assume that $M$ is given by
\eqref{eq:M_max_value}.\footnote{For technical reasons,
a slightly different definition of $\bar{f}$ 
is used for smaller values of $M$ in the proofs of Theorems 1--3.}
In particular, we set $\bar{f}(\gamma) = f(0) = 1$ for $\gamma \in [0, T_1)$.
Since $\bar{f} \geq f$ in this interval, traffic regulation with respect to
$\bar{f}$ may result in violation of \eqref{eq: Stochastic (sigma,rho)}. However,
the violation probability is upper bounded by $T_1/T$, which can be made 
arbitrarily small by suitable choices of $T_1$ and/or $T$.  
We also set $\bar{f}(\gamma) = \bar{f}(T)$ for $T_{M}> \gamma \geq T_{M-1}$, and we choose a large value for $T_M$ such that the burst size of the output traffic is not limited by the stochastic $(\sigma^*,\rho)$ regulator. 

In the interval $[T_i, T_{i+1})$ let 
\begin{align}
    g_i(\gamma) := f(T_{i+1}) + \omega_i (\gamma - T_{i+1})
\end{align}
represent the line connecting the points $(T_i, f(T_i))$ and $(T_{i+1}, f(T_{i+1}))$
with slope 
\begin{align}
\omega_i := \frac{f(T_{i+1}) - f(T_i)}{T_{i+1} - T_i}
\end{align} 
for $i=1, \ldots, M-2$.  If $f(\gamma) \geq g_i(\gamma)$
for all $\gamma \in [T_i, T_{i+1})$ we set $\bar{f} = g_i$ in this interval.  Otherwise, we
set $\bar{f} = h_i$ on $[T_i, T_{i+1})$, where
\begin{align}
    h_i(\gamma) := f(T_{i+1}) + f'(T_{i+1}) (\gamma - T_{i+1}) .
\end{align}
This ensures that $\bar{f} \leq f$ on $[T_1, T_{M-1})$.  We then set $\bar{f}(\gamma) = f(T)$ for
$\gamma \in [T_{M-1}, T_{M}]$ and $\bar{f}(\gamma) = 0$ for $\gamma> T_{M}$.  To summarize,
we define
\begin{align}
\bar{f}(\gamma) :=\left\{\begin{array}{ll}
    1, &  \gamma \in [0, T_1), \\
   f(T_{i+1}) \! + \! m_i(\gamma \! - \! T_{i+1}),  &  \gamma \in [T_{i}, T_{i+1} ) , \\
    f(T) , &  \gamma \in [T_{M-1}, T_M ],  \\
    0,         &  \gamma > T_M,  
\end{array}\right.
\label{eq:f_bar}
\end{align}
for $i=1, \ldots, M-2$ and the slopes $m_i$ are given by
\begin{align}
    m_i =
    \left \{
      \begin{array}{ll}
         \omega_i,        & \mbox{if $f \geq g_i$ on $[T_i, T_{i+1})$} , \\
         f'(T_{i+1}), & \mbox{otherwise} ,
  \end{array}
  \right . 
\end{align}
for $i=1, \ldots, M-2$.

\subsection{Canonical $(\sigma^*, \rho)$ Regulator}

Based on the definition of $\bar{f}$ in \eqref{eq:f_bar}, we modify the constraint
in \eqref{eq:Stochastic (sigma,rho) for jth packet} to hold only for $\gamma \in [T_1, T]$, i.e.,
\begin{align}
      \label{eq:canonical_constraint}
 o_\gamma(t) \leq f(\gamma),~~\forall ~t \in [b_{j-1},b_j],~ \forall \gamma \in [T_1, T] .
\end{align}
Towards a practical implementation, we further replace the bounding 
function $f$ by $\bar{f}$ to obtain the following burstiness constraint:
\begin{align}
     o_\gamma(t) \leq \bar{f}(\gamma),~~\forall ~t \in [b_{j-1},b_j],~ \forall \gamma \in [T_1, T] .
     \label{eq:approx_stoch_regulator_constraint}
\end{align}
To incur minimal packet delay, $\sigma^*(j)$ should be chosen
as the largest value in $\Sigma$ such that the constraint \eqref{eq:approx_stoch_regulator_constraint}
is maintained.  We then define a {\em canonical} $(\sigma^*, \rho)$ regulator as follows:
 \begin{align}
   \cA_j &= \{ \sigma \in \Sigma :  o_\gamma(t) \leq \bar{f}(\gamma),
        \forall t \in [b_{j-1},b_j(\sigma)], \nonumber \\
       &~~~~~~~ \forall \gamma \in [T_1, T ] \} \nonumber \\
   \sigma^*(j) &= \left \{ 
         \begin{array}{ll}
                \sigma_{\max \cA_j}, & \mbox{if $\cA_j \neq \emptyset$}, \\
                \sigma_1        , & \mbox{otherwise}.
         \end{array}
   \right .
   \label{eq:canonical_regulator}
 \end{align}
Equations~\eqref{eq:tj_R1}-\eqref{eq:Stochastic Regulator eq 6} servcan be used as the governing equations 
for a stochastic $(\sigma^*,\rho)$ in which $\sigma$ is replaced by $\sigma^*(j)$
according to~\eqref{eq:canonical_regulator}.
The canonical regulator cannot be implemented directly, since the 
condition in \eqref{eq:canonical_regulator} cannot be verified practically
for all values of $t \in [b_{j-1}, b_j]$ and $\gamma \in [T_1, T]$. 
Next, we develop practical implementations 
of the canonical $(\sigma^*, \rho)$ regulator.

\subsection{Basic Implementation}
\label{subsec:Basic}

We aAssume that $T_M$ is chosen sufficiently large such that for every packet~$j$ the set 
\begin{equation}
\cB_j = \left \{ 1 \leq \ell \leq M : \sigma_\ell \geq W_\rho(\tilde{s}_j ;  R_{\rm 1} ) \right \},
\label{eq:cBj}
\end{equation}
is non-empty. Let 
\begin{align}
\cI_j = \left \{  2 \leq \ell \leq \min \cB_j: o_{T_{\ell-1}}(b_j(\sigma_\ell)) \leq \bar{f}(T_{\ell})\right \}
 \label{eq:Ij}
\end{align}
where $t_j(\sigma_\ell)$ and $b_j(\sigma_\ell)$ are given by~\eqref{eq:tj_R1}  and \eqref{eq:bj}, respectively.
Let
\begin{align}
  \sigma^*(j)   = \left \{
    \begin{array}{ll}
        \sigma_{\max \cI_j},  &  \mbox{if}~\cI_j \neq \emptyset , \\
        \sigma_1,                   & \mbox{otherwise.}
    \end{array}
    \right .
\label{eq:basic_implementation}
\end{align}
Equations~\eqref{eq:cBj}--\eqref{eq:basic_implementation} ardefine used to develop a approximatepractical 
implementations of the canonical $(\sigma^*, \rho)$ regulator given by~\eqref{eq:canonical_regulator}.
FIn this implementation, for a given value of $\sigma_\ell \in \Sigma$,
the condition in \eqref{eq:Ij}~\eqref{eq:canonical_constraint} is checked only
at $t = b_j(\sigma_\ell)$ and $\gamma = T_{\ell-1}$. Therefore, as shown in Section~\ref{sec:numerical}, the constraint in~\eqref{eq:canonical_constraint} may be violated for some values of $t$. However, tThese violations will not occur for sufficiently large values of $t$.

\begin{theorem} \label{thm:overshoot_soln1}
The $(\sigma^*, \rho)$ regulator defined by
\eqref{eq:cBj}--\eqref{eq:basic_implementation} produces an output traffic traffic
stream that satisfies~\eqref{eq:canonical_constraint}
for sufficiently large~$t$.
\end{theorem}

The proof of Theorem~\ref{thm:overshoot_soln1} can be found in~Appendix~\ref{appx:overshoot_soln1}.
A pseudo-code implementation of the stochastic $(\sigma^*,\rho)$ regulator is 
given in Algorithm~\ref{alg:Stochastic regulator}.
The input traffic stream $R_{\rm i}$ is represented as a sequence $\{ (s_1,L_1), \ldots, (s_N,L_N)\}$, where the $s_i$'s are the arrival times of the packets and the $L_i$'s are the packet lengths. The 
$(\sigma^*, \rho)$ regulator consists of the rate $\rho$, the bounding
function $f$, the time period $T$ over which the bound is applied, the 
set $\Sigma$, and the values $\{ T_1, \ldots, T_M \}$
which determine the piecewise-linear bounding function $\bar{f}$.  
The input and output links for 
the regulator are assumed to be of capacity $C > \rho$.  The output traffic stream $R_{\rm o}$ 
is represented by the sequence $\{ (t_1,L_1), \ldots, (t_N,L_N) \}$, where the $t_i$'s are packet departure times.  The {\bf for} loop starting in line~11
finds the largest $\ell \in \{2, \ldots, k=\min \cB_j \}$ such that the inequality
in \eqref{eq:Ij} is satisfied with $\sigma = \sigma_\ell$.  If such $\sigma_\ell$ exists,
then $\sigma^*(j) = \sigma_\ell$; otherwise, $\sigma^*(j) = \sigma_1$, in accordance
with~\eqref{eq:basic_implementation}.
 
Computation of  the departure time, $t_j$, of the $j$th
packet requires updates to $o_{T_{i}}(b_j)$ for $i=1,\ldots,M-1$. Once 
$t_j$ is determined, the values of $o_{T_i}(b_j)$,
for $i=1,\ldots,M-1$, need to be updated. Thus, the overall computational
complexity is $O(M)$ per packet. Using a parallel implementation of the {\bf for} loop at line 11,
the complexity per packet can be reduced to constant time, $O(1)$. 

\alglanguage{pseudocode}
\begin{algorithm}[t]
\caption{$(\sigma^*, \rho)$ stochastic regulator} 
\begin{algorithmic}[1]
\Require $R_{\rm i} = \{ (s_1,L_1), \ldots, (s_N,L_N) \}$; \Comment{Input traffic stream}
\Require $\rho$;  $f(\cdot)$; $T$; $M$; $L_{\max}$; $C$  \Comment{Regulator parameters}
\Ensure  $R_{\rm o} = \{ t_1, t_2, \ldots , t_N \}$ \Comment{Output traffic stream}
\State $\delta \gets (1 - \rho/C)L_{\max}$
\State Compute $T_i, \sigma_i$ for $i=1,2,\ldots,M$ \Comment{\eqref{eq:T_i},~\eqref{eq:sigma_i}}
\State Compute $\bar{f}(\cdot)$ \Comment{\eqref{eq:f_bar}}
\State  $t_1 \gets \tilde{s}_1 \gets s_1$; $b_1 \gets t_1 + L_1/C$ 
\State $W_\rho(\tilde{s}_1;R_1) \gets W_\rho(t_1;R_{\rm o}) \gets 0$
\State Compute $W_\rho(b_1;R_{\rm o})$ \Comment{\eqref{eq:Stochastic Regulator eq 5}}
\State Compute $o_{T_{i}}(b_1)$; $i=1,2,\ldots,M-1$ \Comment{Prop.~\ref{prop:overshoot}}
\For{$j= 1, \ldots, N$} \Comment{Packet~$j$ arrives at time $s_j$}
    \State Compute $\tilde{s}_j$, $W_\rho(\Tilde{s}_j;R_1)$, $\cB_j$ \Comment{\eqref{eq:sTildej},~\eqref{eq:Stochastic Regulator eq 1},~\eqref{eq:cBj}} \label{alg1:line10}
    \State $\mbox{found} \leftarrow \mbox{\bf false}$; $k \gets \min \cB_j$
    \For{$\ell = k, \ldots, 2$} \Comment{$k \geq 2$}
        \State  $\sigma \gets \sigma_\ell$; Compute $t_j(\sigma)$, $b_j(\sigma)$  \Comment{\eqref{eq:tj_R1},~\eqref{eq:bj}}
        \State Compute  $W_\rho(t_j; R_{\rm o})$, $W_\rho(b_j; R_{\rm o})$ 
            \Comment{\eqref{eq:Stochastic Regulator eq 4},~\eqref{eq:Stochastic Regulator eq 5}}\label{alg1:line13}
        \State Compute $o_{T_{\ell-1}}(b_j)$  \Comment{Prop.~\ref{prop:overshoot}}
        \If {$o_{T_{\ell-1}}(b_j) \leq \bar{f}(T_{\ell}$)}
        \Comment{\eqref{eq:Ij}}
            \State $\mbox{found} \gets \mbox{\bf true}$; {\bf break}
        \EndIf
    \EndFor
    \If{{\bf not}~$\mbox{found}$}
    	\State $\sigma \gets \sigma_{1}$; Compute $t_j(\sigma)$, $b_j(\sigma)$ \Comment{\eqref{eq:tj_R1},~\eqref{eq:bj}}
    \EndIf
	\State Compute $o_{T_{i}}(b_j)$; $i \! = \! 1, 2, \ldots, M \! - \! 1$ \Comment{Prop.~\ref{prop:overshoot}}
\EndFor
\end{algorithmic}
\label{alg:Stochastic regulator} 
\end{algorithm}

\subsection{Modified Implementations}
\label{subsec: Modified}

The requirement of sufficient large $t$ in Theorem~\ref{thm:overshoot_soln1} can be
avoided by modifying the definition of $\cI_j$ in \eqref{eq:Ij} to include additional
checks. Let $\cB_j$ be as defined in~\eqref{eq:cBj}. We 
re-define $\mathcal{I}_j$ as follows:
\begin{align}
\cI_j = \big \lbrace  \! 2 \leq \ell \leq \min \cB_j : & ~o_{T_{i}}(b_j(\sigma_\ell)) \leq \bar{f}(T_{i})-\epsilon_{i,j}(\sigma_\ell),
\nonumber \\
&~~~~~~~~ \forall i = 1, \ldots, \ell-1 \big \rbrace,
\label{eq:Ij_2}
\end{align}
where 
\begin{align}
    \label{eq:epsilonij}
    \epsilon_{i,j}(\sigma_\ell)\! := \!\left\lbrace\begin{array}{ll}
      \! \frac{W_\rho(b_{j}(\sigma_\ell);R_{\rm o})-T_i}{\rho b_{j}(\sigma_\ell)}(1-\bar{f}(T_i)),
        &  \!i \!= 1,\ldots,\ell\!-\!2 ,\\
      \!\bar{f}(T_{\ell-1})-\bar{f}(T_\ell),  & \!i\!=\ell-1 .
    \end{array}\right.
\end{align}
The modified definition of $\cI_j$ in \eqref{eq:Ij_2}
involves additional checks for the $j$th packet,
which may result in a smaller value of $\sigma^*(j)$ and
hence higher delay incurred on the packet. Interestingly, our numerical simulations show
that this results in slightly smaller {\em average} delay incurred on the input traffic. This can be explained as follows.  By incurring more delay on {\em some} input 
packets at an earlier stage, 
the output traffic may be better shaped to the desired bound; therefore, 
on average, less delay will need
to be incurred on future packets.

The overshoot ratio $o_{T_i}(t)$ at $t=b_j(\sigma_\ell)$ is checked against  $\bar{f}(T_i)-\epsilon_{i,j}(\sigma_\ell)$ rather than
$\bar{f}(T_i)$, for $i = 1, \ldots, \ell-2$.
The reasoning behind this stricter condition is illustrated in Fig.~\ref{fig:OverShootRatioIncrease}. 
In choosing $\sigma^*(j)=\sigma_\ell$, the overshoot ratios $o_{T_i}(t)$, for $i = 1, \ldots, \ell-2$, will be increasing functions of $t$, as shown in Fig.~\ref{fig:OverShootRatioIncrease}, up to time $t = t_{j+1}(i)$, 
which is defined as the time at which  
\begin{equation}
    W_\rho(t;R_{\rm o})=T_i~~\text{for}~~i=1,2,\ldots,T_{\ell-2},
\end{equation}
and the $(j \! + \! 1)$st packet arrives late enough such that $W_\rho(s_{j+1};R_{\rm o})=0$. Enforcing the condition in~\eqref{eq:Ij_2} with the lower values $\bar{f}(T_i)-\epsilon_{i,j}(\sigma_\ell)$ ensures that the overshoot ratio stays less than $\bar{f}(T_i)$ for all $t\geq b_{j}$.
In this implementation, for a given value of $\sigma_\ell \in \Sigma$,
the condition~\eqref{eq:canonical_constraint} is checked only
at $t = b_j(\sigma_\ell)$ and for $\gamma \in \{ T_1, \ldots, T_{\ell-1}\}$. These extra checks compared to Algorithm~\ref{alg:Stochastic regulator}, as stated in the following theorem and shown in Section~\ref{sec:numerical}, guarantee that there will be no violation of the constraint~\eqref{eq:canonical_constraint}.  

\begin{theorem}
The $(\sigma^*, \rho)$ regulator defined by 
\eqref{eq:cBj}, \eqref{eq:Ij_2}, and~\eqref{eq:basic_implementation}
produces an output traffic
stream that satisfies~\eqref{eq:canonical_constraint} for all $t \geq 0$.
\label{thm:overshoot_soln2}
\end{theorem}

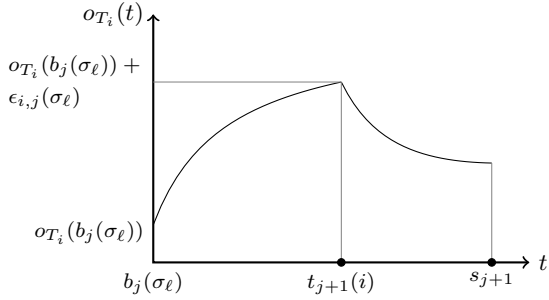
\begin{figure}
\centering
\begin{tikzpicture}
\draw [thick,<->] (5,0)--(0,0)--(0,3.3);
\draw [domain=0:2.5] plot (\x, {(5+31.6*\x)/(10+10*\x)});
\draw [domain=2.5:4.5] plot (\x, {30/(12.5+12*(\x-2.5))+0.25*(\x-2.5)});
\draw [help lines] (2.5,0) -- (2.5,2.4);
\draw [help lines] (4.5,0) -- (4.5,1.3219);
\draw [fill](2.5,0) circle [radius=0.05];
\draw [fill](4.5,0) circle [radius=0.05];
\draw [help lines] (0,2.4) -- (2.5,2.4);
\node [below] at (0,0) {\begin{footnotesize}$b_{j}(\sigma_\ell)$\end{footnotesize}};
\node [below] at (2.5,0) {\begin{footnotesize}$t_{j+1}(i)$\end{footnotesize}};
\node [below] at (4.5,0) {\begin{footnotesize}$s_{j+1}$\end{footnotesize}};
\node [left] at (0,0.4) {\begin{footnotesize}$o_{T_i}(b_j(\sigma_\ell))$\end{footnotesize}};
\node [left, text width=1.8cm] at (0,2.4) {\begin{footnotesize}$o_{T_i}(b_j(\sigma_\ell))+\epsilon_{i,j}(\sigma_\ell)$\end{footnotesize}};
\node [left] at (0,3.3) {\begin{small}$o_{T_i}(t)$\end{small}};
\node [right] at (5,0) {\begin{small}$t$\end{small}};
\end{tikzpicture}
\caption{Overshoot ratio $o_{T_i}(t)$ for $t>b_j$, when $W_\rho(s_{j+1};R_{\rm o})=0$.}
\label{fig:OverShootRatioIncrease}
\end{figure}

\alglanguage{pseudocode}
\begin{algorithm}
\caption{Replacement for lines 14--17 of Algorithm~\ref{alg:Stochastic regulator}} 
\begin{algorithmic}[1] 
\setcounterref{ALG@line}{alg1:line13}
\State Compute $o_{T_{i}}(b_j(\sigma_\ell))$; $i=1,\ldots,\ell-1$ \Comment{Prop.~\ref{prop:overshoot}}
\State Compute $\epsilon_{i,j}(\sigma_\ell)$ ; $i=1,\ldots,\ell-1$\Comment{\eqref{eq:epsilonij}}
\If {$o_{T_{i}}(b_j) \leq \bar{f}(T_{i})-\epsilon_{i,j}(\sigma_\ell)~\forall i\in\{1,\ldots,\ell \! - \! 1\}$}
\State $\mbox{found} \gets \mbox{\bf true}$; {\bf break}
\EndIf
\end{algorithmic}
\label{alg:Stochastic regulator modified} 
\end{algorithm}

See Appendix~\ref{appx:overshoot_soln2} for a proof of Theorem~\ref{thm:overshoot_soln2}.
 By modifying Algorithm~\ref{alg:Stochastic regulator} in accordance with Theorem~\ref{thm:overshoot_soln2},
we obtain an alternative implementation that
satisfies~\eqref{eq:canonical_constraint} for all $t \geq 0$
at the expense of some additional computation. 
The modified implementation is obtained by replacing lines 15--18 
in Algorithm~\ref{alg:Stochastic regulator} with the pseudo-code 
shown in  Algorithm~\ref{alg:Stochastic regulator modified}.
In lines 15 and 16, $\ell-1$ values of $o_{T_{i}}(b_j(\sigma_\ell))$ and $\epsilon_{i,j}(\sigma_\ell)$ need to be computed. Therefore, the complexity of the {\bf for} loop at line 12
in Algorithm~\ref{alg:Stochastic regulator} is $O(M^2)$ and the overall complexity of the modified algorithm is $O(M^2)$ per packet. Parallel implementations of the {\bf for} loop in line 11
of Algorithm~\ref{alg:Stochastic regulator}, and lines 15 and 16 in Algorithm~\ref{alg:Stochastic regulator modified}, can bring the overall time complexity down to $O(1)$ per packet.

With further algorithmic 
modifications, the complexity of Algorithm~\ref{alg:Stochastic regulator modified} 
can be reduced to $O(M)$, i.e., the same time complexity 
as Algorithm~\ref{alg:Stochastic regulator}.  Let $\cB_j$ again be as in \eqref{eq:cBj}.
Let $k = \min \cB_j$ and
\begin{align}
\mathcal{J}_j\!=\!\left\lbrace 1 \leq \ell \leq  k-1\!:\!o_{T_{\ell}}(b_j(\sigma_k)) \! \leq\! \bar{f}(T_{\ell})-\epsilon_{\ell,j}(\sigma_k)\right\rbrace\!,  
\label{eq:Jj}
\end{align}
where $\epsilon_{i,j}(\sigma_k)$ is defined in~\eqref{eq:epsilonij}. 
If $1\in\mathcal{J}_j$ let
\begin{align}
m = \max \left \{ \ell\in\mathcal{J}_j  : i \in \mathcal{J}_j, ~\forall 1 \leq i \leq \ell  \right \},
\label{eq:m}
\end{align}
and let
\begin{align}
 \mathcal{K}_j=\left\lbrace 2\leq \ell \leq  m+1: o_{T_{\ell-1}}(b_j(\sigma_\ell)) \leq \bar{f}(T_{\ell})\right\rbrace,  
 \label{eq:Kj}
\end{align}
where $b_j(\sigma_\ell)$ and $o_{T_{\ell-1}}(b_j(\sigma_\ell))$ given as follows:
\begin{align}
 &b_j(\sigma_\ell) =\tilde{s}_j+ (W_\rho(\tilde{s}_j;R_1)-\sigma_\ell)/\rho + L_j/C ,
 \label{eq:bj(sigma_l)} \\
 &b_j(\sigma_\ell)o_{T_{\ell-1}}(b_j(\sigma_\ell)) =  b_j(\sigma_k)o_{T_{\ell-1}}(b_j(\sigma_k))
 \!\nonumber\\
 &\hspace{12em}+ \!  (W_\rho(\tilde{s}_j;R_1) \! - \! \sigma_\ell)/\rho .
 \label{eq:Overshootattj(sigma_l)} 
\end{align}
We now present a third implementation of the 
canonical $(\sigma^*, \rho)$ regulator given by
\begin{align}
    \sigma^*(j) =\left\lbrace\begin{array}{ll}
        \sigma_{\max \mathcal{K}_j},  &  \mbox{if}~1\in\mathcal{J}_j~~\text{and}~~ \mathcal{K}_j\neq \emptyset, \\
        \sigma_1, & \text{otherwise}.
    \end{array}\right.
    \label{eq:implementation_3}
\end{align}

\begin{theorem}
\label{thm:reducing complexity}
The $(\sigma^*, \rho)$ regulator defined by 
\eqref{eq:cBj} and \eqref{eq:Jj}--\eqref{eq:implementation_3} 
produces the same output stream as the $(\sigma^*, \rho)$ regulator
of Theorem~\ref{thm:overshoot_soln2} for a given input stream and hence the output
stream satisfies~\eqref{eq:canonical_constraint} for all $t \geq 0$.
\end{theorem}

\alglanguage{pseudocode}
\begin{algorithm}[h]
\caption{Replacement for lines 11--18 of Algorithm~\ref{alg:Stochastic regulator}} 
\begin{algorithmic}[1] 
\setcounterref{ALG@line}{alg1:line10}
\State $m \gets 0$; $\mbox{found} \leftarrow \mbox{\bf false}$
\For{$\ell = 1, \ldots, k-1$} \Comment{$k \geq 2$}
   \State  $\sigma \gets \sigma_k$; Compute $t_j(\sigma)$, $b_j(\sigma)$, \Comment{\eqref{eq:tj_R1},~\eqref{eq:bj}}
   \State Compute $\epsilon_{\ell,j}(\sigma)$, $o_{T_{\ell}}(b_j)$    \Comment{\eqref{eq:epsilonij},~Prop.~\ref{prop:overshoot}}
   \If {$o_{T_{\ell}}(b_j) > \bar{f}(T_{\ell})-\epsilon_{\ell,j}(\sigma)$} \Comment{\eqref{eq:Jj}}
       \State {\bf break}
   \EndIf
   \State $m\gets m+1$
\EndFor
\For {$\ell=m+1,\ldots,2$}
   \State $\sigma \gets \sigma_\ell$; Compute $b_j(\sigma)$,  $o_{T_{\ell-1}}(b_j)$ 
        \Comment{\eqref{eq:bj(sigma_l)},~\eqref{eq:Overshootattj(sigma_l)}}
   \If {$o_{T_{\ell-1}}(b_j) \leq \bar{f}(T_{\ell})$}  \Comment{\eqref{eq:Ij_2}}
        \State $\mbox{found} \gets \mbox{\bf true}$; {\bf break}
   \EndIf
\EndFor
\end{algorithmic}
\label{alg:Stochastic regulator modified reduced complexity} 
\end{algorithm}

A proof of Theorem~\ref{thm:reducing complexity} is given in Appendix~\ref{appx:reducing complexity}.
The $(\sigma^*, \rho)$ regulator corresponding to 
Theorem~\ref{thm:reducing complexity} can be implemented by
replacing lines 11-18 in Algorithm~\ref{alg:Stochastic regulator} with the lines shown in
in Algorithm~\ref{alg:Stochastic regulator modified reduced complexity}.  The {\bf for} loops at lines 11 and
19 in Algorithm~\ref{alg:Stochastic regulator modified reduced complexity} both have complexity $O(M)$.
Therefore, the overall complexity of Algorithm~\ref{alg:Stochastic regulator modified reduced complexity}
is $O(M)$ per packet. Similar to Algorithm~\ref{alg:Stochastic regulator} with a suitable 
parallel implementation, the complexity per packet can be further reduced to $O(1)$.

\section{Numerical Results}
\label{sec:numerical}

We consider a system in which the packets sizes $L_j$ are drawn randomly according to
\begin{equation}
L_j\sim \mbox{U} \{ L_{\rm min}, L_{\min} \! + \! 1, \ldots, L_{\rm max} \}  ,
\label{eq:packet size distribution}
\end{equation}
where $\mbox{U}(\mathcal{A})$ denotes a uniform distribution over the set $\mathcal{A}$.
The inter-arrival times of the packets, $s_{j+1}-s_j$, are determined as follows:
\begin{align}
s_{j+1}-s_j\sim U_j+ L_j/C ,
\label{eq:interarrival equation}
\end{align}
where $U_j \sim \mbox{Exp}(\lambda)$, i.e., $\{ U_j \}$ is an i.i.d.\ sequence of
exponentially distributed random variables with parameter $\lambda$.
By adopting~\eqref{eq:interarrival equation} to model the inter-arrival times,
we ensure that packets are received after the previous ones have been fully received, i.e.,
the packets will not overlap with each other.
In a system described by~\eqref{eq:packet size distribution}--\eqref{eq:interarrival equation}, $\rho^{-1}W_\rho(s_j;R_{\rm i})$
is equal to the waiting time 
experienced by the $j$th customer in a $G/G/1$ system in which 
the service time of the $j$th customer is given by 
$S_j=(\rho^{-1}-C^{-1})L_j$ and the inter-arrival time between the $j$th and
$(j+1)$st customer is $U_j$~\cite{Cruz1991a,Kleinrock:1976}. 

In this example, we set $L_{\rm min}=5$, $L_{\rm max}=10$, and $\lambda=0.25$,  and 
$\rho=0.65$. We use the following bounding function: 
\begin{equation}
 f(\sigma) := \left\{\begin{array}{ll}
    -2.5 \times 10^{-3} \sigma + 1,  & 0 \leq \sigma \leq 40 ,   \\
    -5\times 10^{-3} \sigma + 1.1,  & 40 < \sigma \leq T=200 . 
 \end{array} \right.
 \label{eq:f(sigma)}
\end{equation}
In Fig.~\ref{fig:Stochastic Traffi Regulator effect}, $\bar{f}$ is defined by approximating
$f$ by a piecewise-linear function 
according to~\eqref{eq:f_bar} 
with $M=20$, $T_M=400$ and $T_{i+1}-T_{i}=20$ for $i=1,\ldots,M-2$. Note that, as $f(\gamma)$ is also piecewise-linear, $\bar{f}(\gamma)=f(\gamma)$ for $\gamma\in[T_1,T]$. Observe that the output traffic is shaped to satisfy the desired bound. 

\begin{figure}
\centering
\includegraphics[scale=0.45]{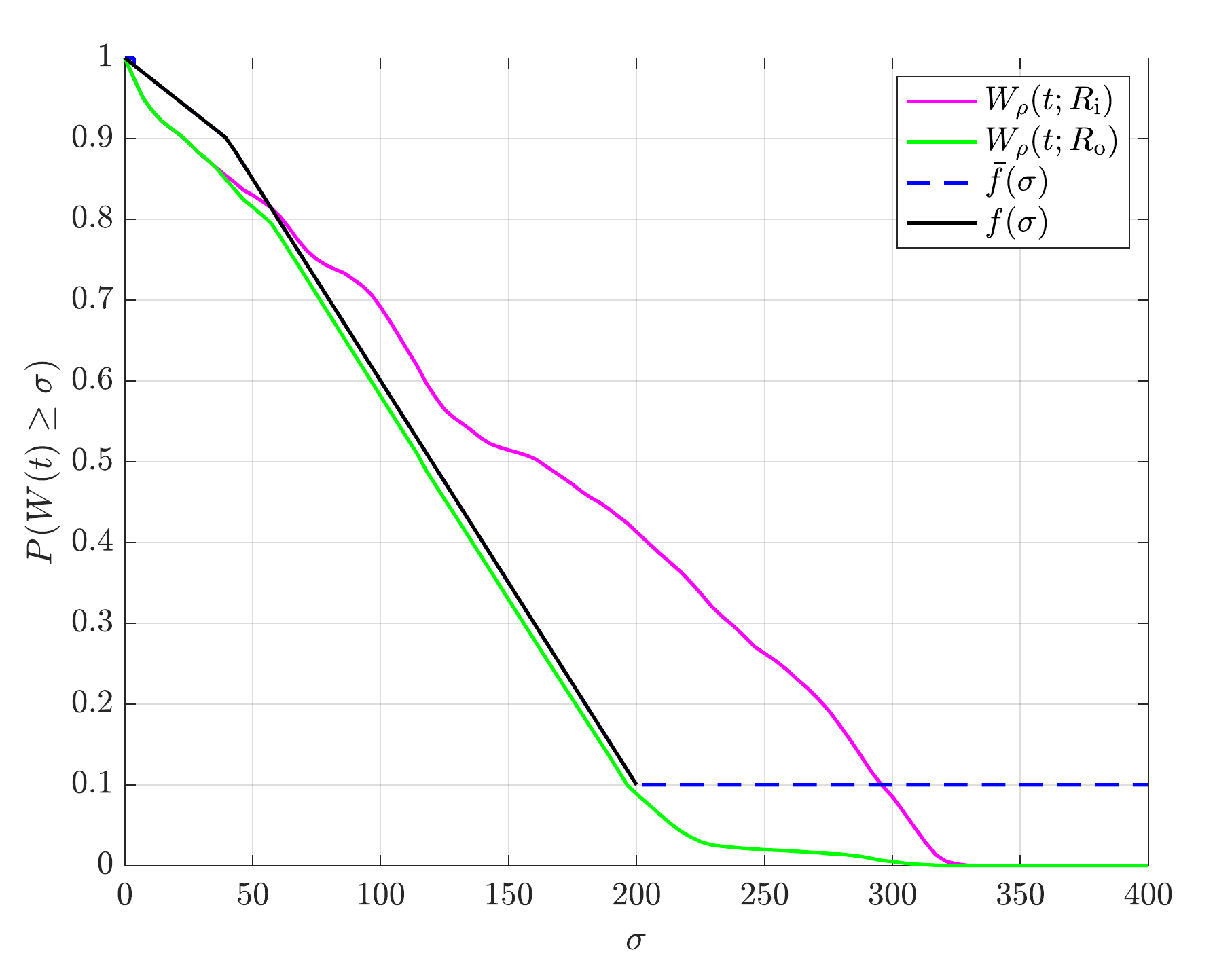}
\caption{Performance of the stochastic $(\sigma^*, \rho)$ traffic regulator.}
\label{fig:Stochastic Traffi Regulator effect}
\end{figure}

Using the same model for inter-arrival and packet lengths, we have investigated the impact
of the parameter $M$ on traffic shaping of the input traffic. From Fig.~\ref{fig:stochastic traffic regulator number of stages}, we see that as $M$ is increased, a tighter fit of the output traffic
to the desired bound can be achieved. In our example, the maximum possible value of $M$,
given by \eqref{eq:M_max_value}, is $M_{\max}=56$, for which a very tight fit to the bound 
is achieved.
Figs.~\ref{fig:Stochastic Traffi Regulator effect} and~\ref{fig:stochastic traffic regulator number of stages} were obtained using Algorithm~\ref{alg:Stochastic regulator modified reduced complexity}. 

Table~\ref{table:stochastic traffic regulator stages} presents
 the average delay and standard deviation of the delay for the packets 
 using Algorithms~\ref{alg:Stochastic regulator} and~\ref{alg:Stochastic regulator modified reduced complexity}. Note that as $M$ increases the average delay decreases and the standard deviation of the packet delay also decreases. These results are expected, since
an increase in $M$ implies that the delay incurred on a packet 
can increase in smaller increments, resulting in smaller overall variance. In addition,
a larger value of $M$ results in a smaller average delay 
since there are more smaller choices of delay
for a packet in order to maintain the burstiness bound. Algorithm~\ref{alg:Stochastic regulator modified reduced complexity} slightly outperforms Algorithm~\ref{alg:Stochastic regulator} for larger values of $M$,
in particular, $M=56$, as shown in Table~\ref{table:stochastic traffic regulator stages}.

\begin{figure}
\centering
\includegraphics[scale=0.45]{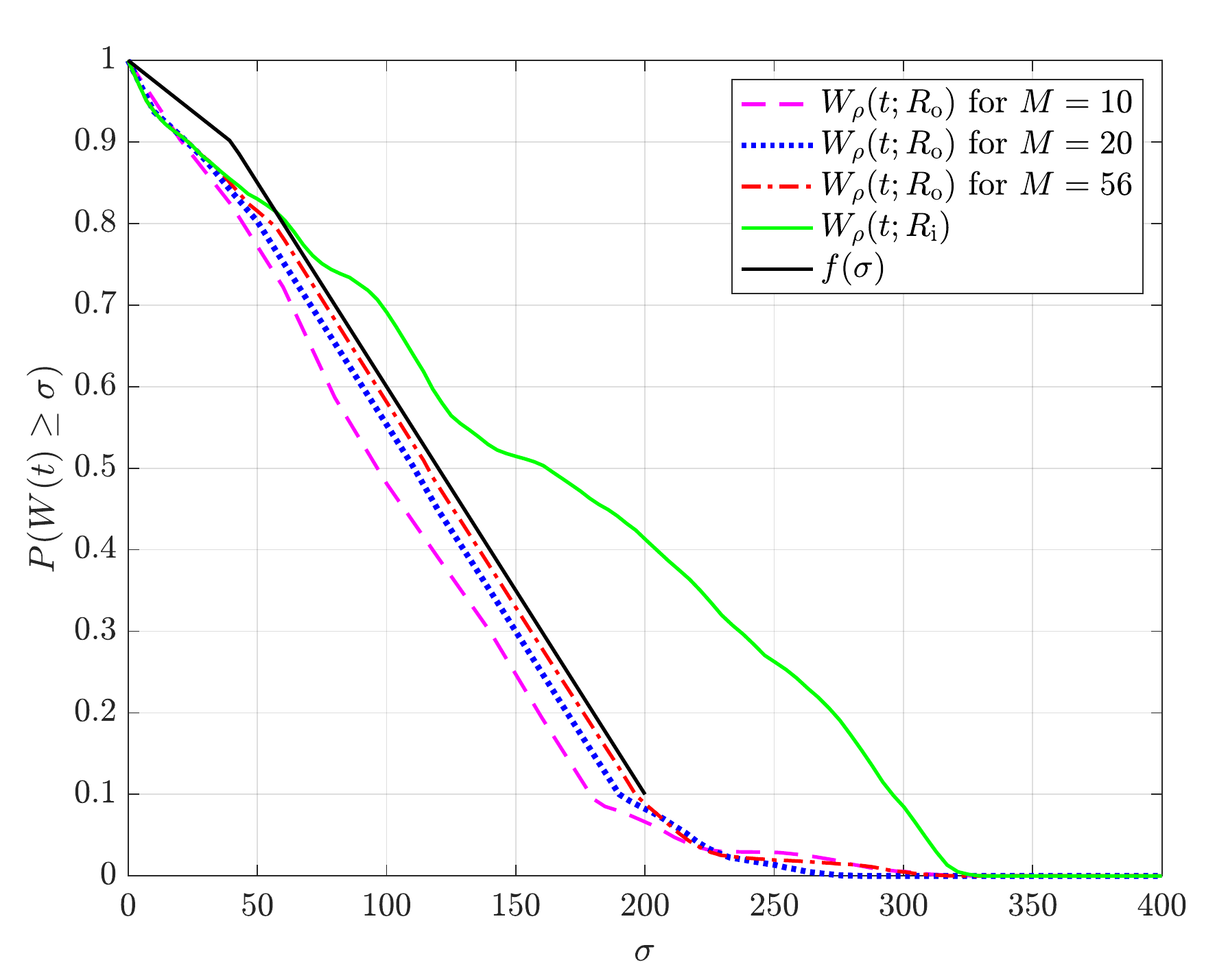}
\caption{Traffic regulator performance with different $M$ values.}
\label{fig:stochastic traffic regulator number of stages}
\end{figure}

\begin{table}
\centering
\caption{Traffic shaping delay with different $M$ values for Algorithm~\ref{alg:Stochastic regulator} and Algorithm~\ref{alg:Stochastic regulator modified reduced complexity}.}
\label{table:stochastic traffic regulator stages}
\begin{tabular}{|c|cc|cc|}
 \cline{2-5} 
\multicolumn{1}{c|}{} &  \multicolumn{2}{c}{Average Delay}& \multicolumn{2}{|c|}{Std.\ Dev.\ of Delay}\\
    \hline
    $M$ &Alg.\ 1 & Alg.\ 3 & Alg.\ 1 & Alg.\ 3 \\
    \hline
     10 & 89&  89 & 115&115 \\
     20 & 78&78  & 109&109   \\
     56 & 72&71  & 100&99\\
     \hline
\end{tabular}
\end{table}

The main advantage of Algorithm~\ref{alg:Stochastic regulator modified reduced complexity} is that
the constraint on the bounding function~$f$ is guaranteed to hold for all values of $t \geq 0$, whereas
some violations may occur using Algorithm~\ref{alg:Stochastic regulator} for small values of $t$.
On the other hand, Algorithm~\ref{alg:Stochastic regulator} is somewhat simpler from an implementation
point of view.
In Fig.~\ref{fig:OvershootRatioFunctioninTime} the overshoot ratio $o_{T_{17}}(t)$ vs.\ $t$ is shown for Algorithms~\ref{alg:Stochastic regulator} and~\ref{alg:Stochastic regulator modified reduced complexity}. In Fig.~\ref{fig:OvershootRatioFunctioninTime}, some violations of~\eqref{eq:Stochastic (sigma,rho) for jth packet} occur with Algorithm~\ref{alg:Stochastic regulator} but there are no violations with
Algorithm~\ref{alg:Stochastic regulator modified reduced complexity}.

\begin{figure}
\centering
\includegraphics[scale=0.45]{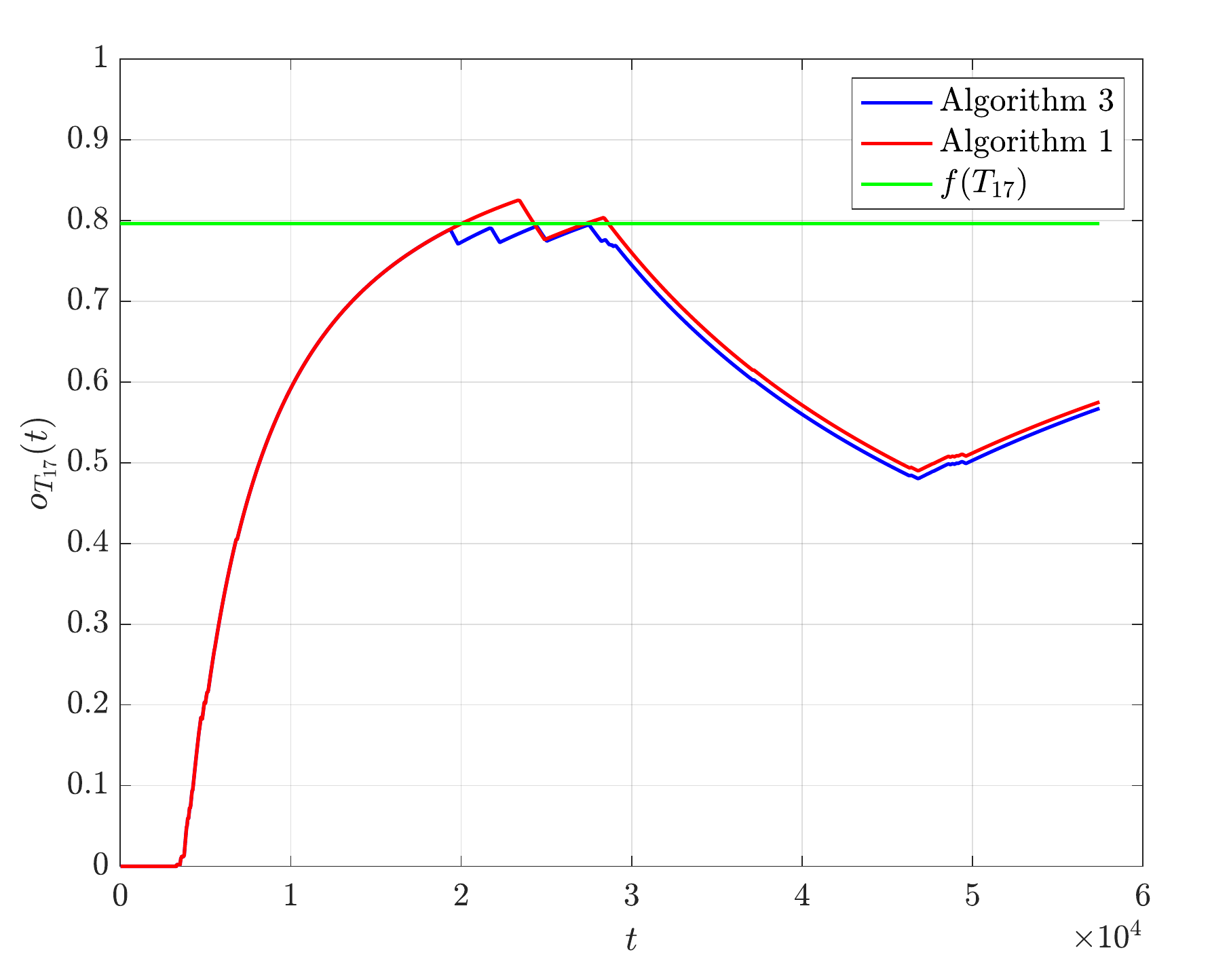}
\caption{Overshoot ratio $o_{T_{17}}(t)$ is shown vs.\ $t$ when $M=56$ for Algorithms~\ref{alg:Stochastic regulator} and~\ref{alg:Stochastic regulator modified reduced complexity}.}
\label{fig:OvershootRatioFunctioninTime}
\end{figure}

\section{Conclusion}
\label{sec:conclusion}

The stochastic traffic regulator developed in this paper addresses an open problem in the application
of stochastic network calculus to real networks.  The
validity of the stochastic end-to-end delay bounds determined via stochastic network calculus relies
on the assumption that all input traffic streams conform to certain stochastic burstiness bounds~\cite{Yaron1994,Starobinski2000a,JIANG20092011}.  Without a means of enforcing such bounds,
on the input traffic, the end-to-end delay bounds cannot be guaranteed in practice.  Given an input
traffic stream, our proposed $(\sigma^*, \rho)$ regulator inserts delays, as necessary, 
to ensure that the output traffic stream conforms to the following stochastic burstiness constraint:
when the traffic is offered to a queue with constant service rate $\rho$, the virtual workload of the queue
exceeds a threshold $\gamma$ with probability is less than $f(\gamma)$ for all $\gamma$ in a given range, 
where $f$ is a non-increasing bounding function. This is essentially the gSBB characterization from~\cite{JIANG20092011}.

Our proposed canonical $(\sigma^*, \rho)$ regulator consists
of a finite set $\Sigma = \{ \sigma_1, \ldots, \sigma_M \}$ of burstiness parameter
values and an associated piecewise-linear function $\bar{f}$ that
approximates the bounding function~$f$.  The operation of the $(\sigma^*, \rho)$ regulator is similar
to that of a deterministic $(\sigma, \rho)$ regulator, except that the
burstiness parameter is variable and is computed for each arriving packet in such a way as 
to achieve the desired bound.  Practical implementations of the canonical regulator were developed and their key properties were established.  Assuming stationarity and ergodicity
of the input traffic, all of the implementations enforce comformance of the output traffic to the stochastic burstiness bound~\eqref{eq: Stochastic (sigma,rho)} in steady-state.  Algorithm~\ref{alg:Stochastic regulator modified reduced complexity}
has the additional property that the overshoot ratio, i.e., the fraction of
time that virtual workload exceeds $\gamma$ is less than $f(\gamma)$ at {\em all} times
for all $\gamma$ in the considered range.  

The performance characteristics of our $(\sigma^*, \rho)$ regulator 
implementations were demonstrated through numerical examples
using a particular bounding function.  With larger $M$, the approximation of $f$ by $\bar{f}$
becomes more accurate and consequently, the mean and standard deviation of packet delay decreases.
The $(\sigma^*, \rho)$ regulator could also be applied in conjunction with 
the phase-type bounds proposed in~\cite{KordiCISS2019}.
A method for fitting a phase-type bounding function to a given
traffic source is developed in~\cite{KordiCISS2020}.

\appendices
\section{Proof of Proposition~\ref{prop:deterministic delay w buffer}}
\label{appx:deterministic delay w buffer}
\numberwithin{equation}{section}
\setcounter{equation}{0}

We first establish the following lemma\footnote{For notational convenience
we drop the subscript $\rho$ when referring to workload functions $W_\rho(\cdot; \cdot)$.}.

\begin{appxlemma}
\label{appxlem:workload}
\begin{align}
    W(\tilde{s}_j;R_1)=W(s_j;R_{\rm i})-(\tilde{s}_j-s_j)\rho
    \label{eq:sj_and_stildej}
\end{align}
\begin{proof}
We prove~\eqref{eq:sj_and_stildej} using induction. For $j=1$, i.e., the first
packet arrival, $\tilde{s}_1=s_1$ and $W(\ts_j; R_1) = W(s_j; R_{\rm i})=0$, so~\eqref{eq:sj_and_stildej}
holds in this case.  Assuming~\eqref{eq:sj_and_stildej} is valid for the $j$th packet, we now verify 
that it holds for the $(j+1)$st packet.
Note that in the interval $(\ts_j, \ta_j)$, the workload function $W(t; R_1)$ increases linearly with slope $C - \rho$
by an amount $\delta_j$ (see \eqref{eq:delta_j}) and then decreases linearly with slope $-\rho$ in the
interval $(\ta_j, \ts_{j+1})$ (see \eqref{eq:Stochastic Regulator eq 2} and \eqref{eq:Stochastic Regulator eq 3}).  Hence,
\begin{align}
W(\ts_{j+1}, R_1) =  W(\ts_j , R_1) + \delta_j - \rho(\ts_{j+1} - \ta_j) .
\label{eq:W_ts_R1}
\end{align}
By a similar argument (see \eqref{eq:Regulator development eq 1}-\eqref{eq:Regulator development eq 2}),
\begin{align}
W(s_{j+1}, R_{\rm i}) =  W(s_j , R_{\rm i}) + \delta_j - \rho(s_{j+1} - a_j) .
\label{eq:W_s_Ri}
\end{align}
Next, we apply first \eqref{eq:sj_and_stildej} and then \eqref{eq:W_s_Ri} into \eqref{eq:W_ts_R1} and re-arrange
terms to obtain
\begin{align}
 W(\ts_{j+1}, R_1) &=  
 W(s_{j+1}, R_{\rm i}) - \rho(\ts_{j+1} - s_{j+1}) \nonumber \\
 & ~~ + [(\ta_j - \ts_j) - (a_j - s_j)] \rho .
 \label{eq:W_ts_j1}
\end{align}
The last term in \eqref{eq:W_ts_j1} vanishes, since $\ta_j - \ts_j = a_j - s_j = L_j/C$.  Thus,
we have established \eqref{eq:sj_and_stildej} using mathematical induction.
\end{proof}
\end{appxlemma}
\begin{proof}[Proof of Proposition~\ref{prop:deterministic delay w buffer}]
First, suppose $W(\ts_j; R_1) \leq \sigma$.  Since
$W(\ts_j; R_1) = W(\ts_j; R_{\rm o})$  (see~\eqref{eq:Stochastic Regulator eq 4}), 
we have $W(\ts_j; R_{\rm o}) \leq \sigma$, i.e., the 
$(\sigma, \rho)$ constraint is satisfied by the output process at time $\ts_j$.
This implies that the $j$th packet departs the regulator starting
at time $t_j = \ts_j$, which confirms \eqref{eq:tj_R1} in this case.

Next, suppose $W(\ts_j; R_1) > \sigma$.  Then $W(s_j; R_{\rm i}) \geq W(s_j; R_1) \geq W(\ts_j; R_1) > \sigma$.
Thus,  in this case, we can remove the $[\cdot]^+$ operator
in both \eqref{eq:delay_regulator} and \eqref{eq:tj_R1}.
Applying Lemma~\ref{appxlem:workload} to the right-hand side of \eqref{eq:tj_R1}, we have
\begin{align*}
    [ & W(\tilde{s}_j;R_1)-\sigma]/\rho+\tilde{s}_j
    \\
    &=[ W(s_j;R_{\rm i})-(\tilde{s}_j-s_j)\rho-\sigma ] / \rho+\tilde{s}_j\\
    &= [ W(s_j;R_{\rm i}) - \sigma ]/\rho+s_j = t_j .
\end{align*}
This completes the proof of Proposition~\ref{prop:deterministic delay w buffer}.
\end{proof}

\section{Proof of Theorem~\ref{thm:overshoot_soln2}}
\label{appx:overshoot_soln2}

\numberwithin{equation}{section}
\setcounter{equation}{0}

The proof of Theorem~\ref{thm:overshoot_soln2} is based on the following two lemmas.
\begin{appxlemma}
\label{appxlem:overshoot_soln2_discreet T_i}
Let $\cB_j$ be as defined in~\eqref{eq:cBj}, let $k=\min \cB_j$, and $\mathcal{I}_j$ is defined as
\begin{align}
\cI_j = \big \lbrace  \! 2 \leq \ell \leq k : & ~o_{T_{i}}(b_j(\sigma_\ell)) \leq \bar{f}(T_{i})-\epsilon_{i,j}(\sigma_\ell),
\nonumber \\
&~~~~~~~~ \forall i = 1, \ldots, \ell-1 \big \rbrace,
\label{eq:Ij appendix}
\end{align}
, where 
\begin{align}
    \label{eq:epsilonij appendix}
    \epsilon_{i,j}(\sigma_\ell)\!=\!\left\lbrace\begin{array}{ll}
      \!\frac{W_\rho(b_{j}(\sigma_\ell);R_{\rm o})-T_i}{\rho b_{j}(\sigma_\ell)}(1-\bar{f}(T_i))   &  \!\!i \in\! \{ 1,\ldots,\ell\!-\!2 \}\\
      \!\bar{f}(T_{\ell-1})-\bar{f}(T_\ell)  & \!\!i=\ell-1,
    \end{array}\right.
\end{align}
where $t_j(\sigma_\ell)$ and $b_j(\sigma_\ell)$ are given by~\eqref{eq:tj_R1}  and \eqref{eq:bj}, respectively.  Set
\begin{align}
  i^*  = \left \{
    \begin{array}{ll}
        \max \cI_j,  &  \cI_j \neq \emptyset , \\
        1,                   & \mbox{otherwise.}
    \end{array}
    \right .
\label{eq:istar appendix}
\end{align}
If the burst parameter
$\sigma^*(j)$ is set as follows:
\begin{align}
    \sigma^*(j) = \sigma_{i^*} ,
 \label{eq:sigma_j_soln1 appendix}
\end{align}
then 
\begin{align}
&o_{T_i}(t)\leq  \bar{f}(T_i),\nonumber\\ &\forall~t\in[b_{j-1},b_j(\sigma_{i^*})] ~~~\forall~i\in \{1,2,\ldots,M\}.
\label{eq:OvershootRatioBoundforCurrentPacket}
\end{align}

\begin{proof}
We proof~\eqref{eq:OvershootRatioBoundforCurrentPacket} using induction. Note that, for $j=1$ we have $\mathcal{I}_j=\emptyset$ and according to~\eqref{eq:Stochastic Regulator eq 5}
\begin{align}
O_{T_i}(t;R_{\rm o})=0~~~\forall~i\in \{1,2,\ldots,M\}.    
\end{align}
Therefore~\eqref{eq:OvershootRatioBoundforCurrentPacket} holds for $j=1$. Lets assume~\eqref{eq:OvershootRatioBoundforCurrentPacket} is valid for the $j$th packet, we now verify it holds for the $(j+1)$th packet.
We assume $\sigma^*(j)=\sigma_m$, where $m\in \mathcal{I}_j$. Therefore $b_j=b_j(\sigma_m)$ and according to~\eqref{eq:Ij appendix} and assumption~\eqref{eq:OvershootRatioBoundforCurrentPacket} 
\begin{align}
&o_{T_i}(b_j)\leq  \bar{f}(T_i), ~~\forall~i\in \{1,2,\ldots,M\}, 
\label{eq:OvershootIneqaulity1}
\\
&o_{T_{i}}(b_j) \leq \bar{f}(T_{i})-\epsilon_{i,j}(\sigma_m),
~~\forall i \in \{ 1,2,\ldots,m-1 \},
\label{eq:OvershootIneqaulity2}
\end{align}
where $\epsilon_{i,j}(\sigma_m)$ is defined in~\eqref{eq:epsilonij appendix}. As $\sigma_{m+1}> W(b_{j};R_{\rm o})\geq W(t_{j+1};R_{\rm o})$, therefore $\sigma^*(j+1)$ can be chosen from $\{\sigma_1,\ldots,\sigma_{m},\sigma_{m+1}\}$. We define $\hat{W}(t;R_{\rm o})$ as the decreasing workload with slope $\rho$ from $W(b_j;R_{\rm o})$ as shown in Fig.~\ref{fig:WorkloadDecline}. Also, if $\sigma^*(j+1)$ is set as $\sigma^*(j+1)=\sigma_\ell$ for $\ell\in\{1,2,\ldots,m+1\}$ as in Fig.~\ref{fig:WorkloadDecline}, we have 
\begin{figure}[t]
    \centering
        \begin{adjustbox}{scale=1}
        \begin{tikzpicture}
           \draw [thick,<->] (5,0)--(0,0)--(0,4.5);
            \node [left] at (0,4.5) 
            {\begin{small}
            $W(t;R_{\rm o})$
            \end{small}};
            \node [left] at (0,3.4) 
            {\begin{footnotesize}
            $W(b_{j};R_{\rm o})$
            \end{footnotesize}};
            \node [right] at (5,0) {\small{$t$}};
            \node [below] at (0.7,0) 
            {\begin{scriptsize}
            $b_{j}$
            \end{scriptsize}};
            \node [left,below] at (2.62,0) 
            {\begin{scriptsize}
            $t_{j+1}$
            \end{scriptsize}};
            \node [right,below] at (3.22,0) 
            {\begin{scriptsize}
            $b_{j+1}$
            \end{scriptsize}};
            \node [right,below] at (1.37,0) 
            {\begin{scriptsize}
            $t_{j+1}(m)$
            \end{scriptsize}};
            \node [left] at (0,2.8) 
            {\begin{footnotesize}
            $T_{m}$
            \end{footnotesize}};
            \node [left] at (0,1.5) 
            {\begin{footnotesize}
            $\sigma^*(j+1)$
            \end{footnotesize}};
            \node [left] at (0,3.7) 
            {\begin{footnotesize}
            $T_{\ell}$
            \end{footnotesize}};
            \node [below] at (4.5,0) 
            {\begin{scriptsize}
            $t_{j+1}(0)$
            \end{scriptsize}};
            \node [left] at (0,4) 
            {\begin{footnotesize}
            $\sigma_{\ell+1}$
            \end{footnotesize}};
            \draw [dashed] (0.3,1.8)--(0.5,2.6);
            \draw (0.5,2.6)--(0.7,3.4);
            \draw (0.7,3.4)--(2.82,1.5);
            \draw [dashdotted] (0.7,3.4)--(4.5,0);
            \draw (2.82,1.5)--(3.02,2.3);
            \draw [help lines] (0,4)--(2.9,4);
            \draw [help lines] (0,3.7)--(2.9,3.7);
            \draw [help lines] (2.82,0)--(2.82,1.5);
            \draw [help lines] (3.02,0)--(3.02,2.3);
            \draw [help lines] (0,2.8) -- (1.37,2.8);
            \draw [help lines] (1.37,0) -- (1.37,2.8);
            \draw [help lines] (0.7,0) -- (0.7,3.4);
            \draw [help lines] (0,3.4) -- (0.7,3.4);
            \draw [help lines] (0,1.5) -- (2.82,1.5);
            \draw [fill](2.82,0) circle [radius=0.03];
            \draw [fill](3.02,0) circle [radius=0.03];
            \draw [fill](4.5,0) circle [radius=0.03];
            \draw [fill](1.37,0) circle [radius=0.03];
            \draw (2.9,4.5)--(5,4.5)--(5,3.7)--(2.9,3.7)--(2.9,4.5);
            \draw (4.25,4.3)--(4.8,4.3);
            \draw [dashdotted](4.25,3.9)--(4.8,3.9);
            \node [left] at (4.25,4.3) 
            {\begin{scriptsize}
            $W(t;R_{\rm o})$
            \end{scriptsize}};
            \node [left] at (4.25,3.9) 
            {\begin{scriptsize}
            $\hat{W}(t;R_{\rm o})$
            \end{scriptsize}};
            \end{tikzpicture}
        \end{adjustbox}
        \caption{$W(t;R_{\rm o})$ and $\hat{W}(t;R_{\rm o})$ when $\sigma^*(j+1)\in\{\sigma_1,\ldots,\sigma_m\}$.}
        \label{fig:WorkloadDecline}
    \end{figure}
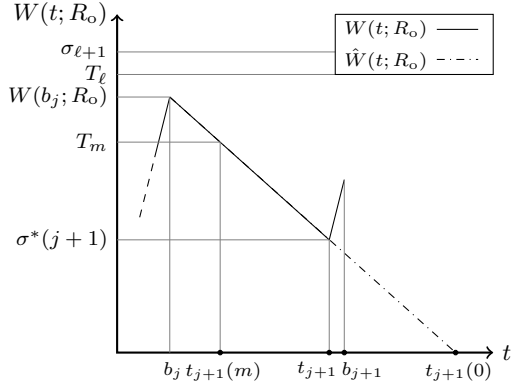
\begin{equation}
W(t;R_{\rm o})=\hat{W}(t;R_{\rm o}) , ~~~~\forall~t\in[b_j,t_{j+1}(\sigma_\ell)].
\end{equation}
We define $t_{j+1}(i)$ as the time that 
\begin{equation}
    \hat{W}(t_{j+1}(i);R_{\rm o})=T_i,~~\text{for}~~i\in\{0,1,\ldots,m\},
\end{equation}
with $T_0=0$. Therefore, according to~\eqref{eq:Stochastic Regulator eq 1}, we have
\begin{equation}
 t_{j+1}(i)=b_{j}+\frac{W(b_{j};R_{\rm o})-T_i}{\rho}.  
 \label{eq:tji}
\end{equation}
According to Proposition~\ref{prop:overshoot} and Fig.~\ref{fig:WorkloadDecline}, we have
\begin{align}
&O_{T_i}(t;R_{\rm o})= O_{T_i}(b_j;R_{\rm o})
\nonumber\\
&~~~\forall~t\in[b_j,b_{j+1}],~i\in\{m,m+1,\ldots,M\}.
\label{eq:OverShootDevelopment1}
\end{align}
Therefore, according to~\eqref{eq:OvershootIneqaulity1}
\begin{align}
 &o_{T_i}(t)\leq  \bar{f}(T_i),  
 \nonumber\\
&~~~\forall~t\in[b_j,b_{j+1}],~i\in\{m,m+1,\ldots,M\}.
\label{eq:OvershootIneqaulityNew1}   
\end{align}
On the other hand, as $\sigma^*(j+1)$ is chosen according to~\eqref{eq:istar appendix} we can have two following subcases 
\setcounter{case}{0}
\begin{case}
$\sigma^*(j+1)=\sigma_n\in\{\sigma_1,\ldots,\sigma_m\}$.
\end{case}
In this case, we have 
\begin{align}
&O_{T_i}(t;R_{\rm o})\!= \!\left\lbrace \begin{array}{ll}
    \!\!O_{T_i}(b_j;R_{\rm o})\!+\!(t-b_j) & \!\!\forall~t\in[b_j,t_{j+1}(i)], \\
    \!\!O_{T_i}(t_{j+1}(i);R_{\rm o}) & \!\!\forall~t\in(t_{j+1}(i),b_{j+1}],
\end{array}\right.
\label{eq:OverShootDevelopment2}
\end{align}
for $i\in\{n,n+1,\ldots,m-1\}$. Therefore, according to~\eqref{eq:tji} and~\eqref{eq:OverShootDevelopment2}, for $i\in\{n,n+1,\ldots,m-1\}$
\begin{align*}
\max_{t \in[b_{j},b_{j+1}]}& o_{T_i}(t)=o_{T_i}(t_{j+1}(i)).  
\end{align*}
But as we have~\eqref{eq:OvershootIneqaulity2} therefore,
\begin{align}
&o_{T_i}(t)\leq  \bar{f}(T_i),  
 \nonumber\\
&~~~\forall~t\in[b_j,b_{j+1}],~i\in\{n,n+1,\ldots,m-1\}.
\label{eq:OvershootIneqaulityNew2}       
\end{align}
It can be easily verified according to Proposition~\ref{prop:overshoot}, for $i\in\{1,2,\ldots,n-1\}$ we have 
\begin{align}
O_{T_i}(t;R_{\rm o})= 
    O_{T_i}(b_j;R_{\rm o})+(t-b_j),  ~~\forall~t\in[b_j,b_{j+1}].
\label{eq:OverShootDevelopment3}
\end{align}
Therefore, for $i\in\{1,2,\ldots,n-1\}$
\begin{align}
\max_{t \in[b_{j},b_{j+1}]}& o_{T_i}(t)=o_{T_i}(b_{j+1})
\label{eq:maxOvershootRatioj+1}
\end{align}
But as $\sigma^*(j+1)=\sigma_n$ is chosen using~\eqref{eq:istar appendix} we have
\begin{align}
&o_{T_{i}}(b_{j+1}) \leq \bar{f}(T_{i})-\epsilon_{i,j+1}(\sigma_n)
\nonumber\\
&\hspace{8em} \forall i \in \{ 1,2,\ldots,n-1 \},
\label{eq:OvershootIneqaulityassumption}
\end{align}
Therefore according to~\eqref{eq:maxOvershootRatioj+1} and~\eqref{eq:OvershootIneqaulityassumption}
\begin{align}
&o_{T_i}(t)<\bar{f}(T_i)
\nonumber\\
&~~~\forall~t\in[b_j,b_{j+1}],~i\in\{1,2,\ldots,n-1\}.
\label{eq:OvershootIneqaulityNew3}    
\end{align}
Therefore, for this case using~\eqref{eq:OvershootIneqaulityNew1},~\eqref{eq:OvershootIneqaulityNew2} and~\eqref{eq:OvershootIneqaulityNew3} 
\begin{align}
&o_{T_i}(t)<\bar{f}(T_i)
\nonumber\\
&~~~\forall~t\in[b_j,b_{j+1}],~i\in\{1,2,\ldots,M\}.
\label{eq:OvershootIneqaulityNew4}    
\end{align}
\begin{case}
$\sigma^*(j+1)=\sigma_{m+1}$.
\end{case}
In this case workload can be as in Fig.~\ref{fig:WorkloadDecline2} and can have subcases I, II, and III. In all subcases as 
\begin{align*}
W(t;R_{\rm o})<T_i~~\forall~t\in[b_j,b_{j+1}],~~i\in\{m+1,\ldots,M\}.
\end{align*}
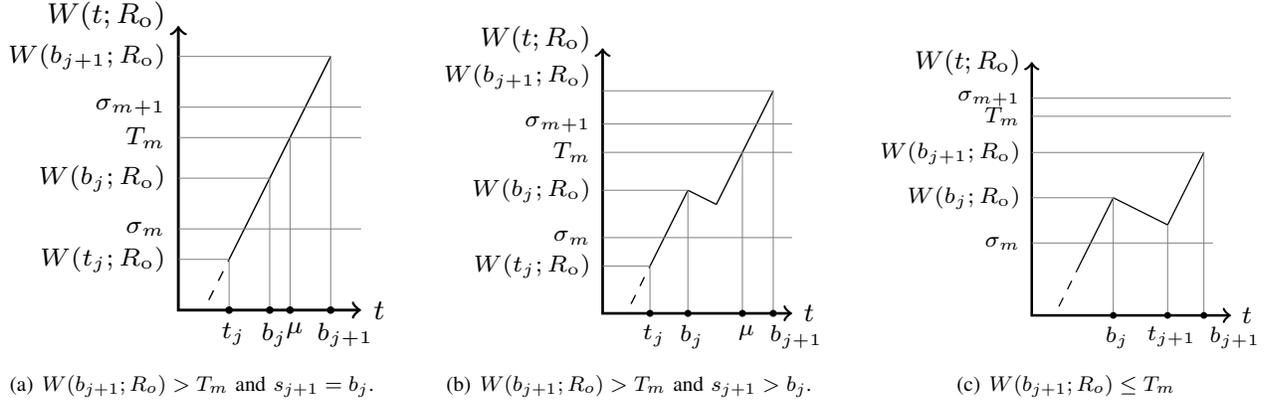
\begin{figure*}[t]
    \subfigure[$W(b_{j+1};R_o)>T_m$ and $s_{j+1}=b_j$.]{
    \begin{adjustbox}{width=0.3\textwidth}
        \begin{tikzpicture}
           \draw [thick,<->] (1.8,0)--(0,0)--(0,2.8);
            \node [left] at (0,2.9) 
            {\begin{footnotesize}
            $W(t;R_{\rm o})$
            \end{footnotesize}};
            \node [left] at (0,1.3)
            {\begin{scriptsize}
            $W(b_{j};R_{\rm o})$
            \end{scriptsize}};
            \node [left] at (0,0.5)
            {\begin{scriptsize}
            $W(t_{j};R_{\rm o})$
            \end{scriptsize}};
            \node [left] at (0,2.5)
            {\begin{scriptsize}
            $W(b_{j+1};R_{\rm o})$
            \end{scriptsize}};
            \node [right] at (1.8,0) {\footnotesize{$t$}};
            \node [below] at (0.9,0) 
            {\begin{scriptsize}
            $b_{j}$
            \end{scriptsize}};
            \node [below] at (0.5,0) 
            {\begin{scriptsize}
            $t_{j}$
            \end{scriptsize}};
            \node [right,below] at (1.6,0) 
            {\begin{scriptsize}
            $b_{j+1}$
            \end{scriptsize}};
            \node [left] at (0,1.7) 
            {\begin{scriptsize}
            $T_{m}$
            \end{scriptsize}};
            \node [left] at (0,2) 
            {\begin{scriptsize}
            $\sigma_{m+1}$
            \end{scriptsize}};
            \node [left] at (0,0.8) 
            {\begin{scriptsize}
            $\sigma_{m}$
            \end{scriptsize}};
            \node [below] at (1.1,0) 
            {\begin{scriptsize}
            $\mu$
            \end{scriptsize}};
            \draw [dashed] (0.3,0.1)--(0.5,0.5);
            \draw (0.5,0.5)--(0.9,1.3);
            \draw (0.9,1.3)--(1.5,2.5);
            \draw [help lines] (1.1,0)--(1.1,1.7);
            \draw [help lines] (0,1.7)--(1.8,1.7);
            \draw [help lines] (0,2)--(1.8,2);
            \draw [help lines] (0.9,0)--(0.9,1.3);
            \draw [help lines] (1.5,0)--(1.5,2.5);
            \draw [help lines] (0,1.3) -- (0.9,1.3);
            \draw [help lines] (0,2.5) -- (1.5,2.5);
            \draw [help lines] (0.5,0) -- (0.5,0.5);
            \draw [help lines] (0,0.5) -- (0.5,0.5);
            \draw [help lines] (0,0.8) -- (1.8,0.8);
            \draw [fill](0.5,0) circle [radius=0.03];
            \draw [fill](0.9,0) circle [radius=0.03];
            \draw [fill](1.5,0) circle [radius=0.03];
            \draw [fill](1.1,0) circle [radius=0.03];
            \end{tikzpicture}
        \end{adjustbox}
        \label{fig:WorkloadDecline2Subcase1}
        }
        \subfigure[$W(b_{j+1};R_o)>T_m$ and $s_{j+1}>b_j$.]{
    \begin{adjustbox}{width=0.3\textwidth}
        \begin{tikzpicture}
           \draw [thick,<->] (2,0)--(0,0)--(0,2.8);
            \node [left] at (0,2.9) 
            {\begin{footnotesize}
            $W(t;R_{\rm o})$
            \end{footnotesize}};
            \node [left] at (0,1.3)
            {\begin{scriptsize}
            $W(b_{j};R_{\rm o})$
            \end{scriptsize}};
            \node [left] at (0,0.5)
            {\begin{scriptsize}
            $W(t_{j};R_{\rm o})$
            \end{scriptsize}};
            \node [left] at (0,2.5)
            {\begin{scriptsize}
            $W(b_{j+1};R_{\rm o})$
            \end{scriptsize}};
            \node [right] at (2,0) {\footnotesize{$t$}};
            \node [below] at (0.9,0) 
            {\begin{scriptsize}
            $b_{j}$
            \end{scriptsize}};
            \node [below] at (0.5,0) 
            {\begin{scriptsize}
            $t_{j}$
            \end{scriptsize}};
            \node [right,below] at (2,0) 
            {\begin{scriptsize}
            $b_{j+1}$
            \end{scriptsize}};
            \node [left] at (0,1.7) 
            {\begin{scriptsize}
            $T_{m}$
            \end{scriptsize}};
            \node [left] at (0,2) 
            {\begin{scriptsize}
            $\sigma_{m+1}$
            \end{scriptsize}};
            \node [left] at (0,0.8) 
            {\begin{scriptsize}
            $\sigma_{m}$
            \end{scriptsize}};
            \node [below] at (1.475,0) 
            {\begin{scriptsize}
            $\mu$
            \end{scriptsize}};
            \draw [dashed] (0.3,0.1)--(0.5,0.5);
            \draw (0.5,0.5)--(0.9,1.3);
            \draw (0.9,1.3)--(1.2,1.15);
            \draw (1.2,1.15)--(1.8,2.35);
            \draw [help lines] (1.475,0)--(1.475,1.7);
            \draw [help lines] (0,1.7)--(2,1.7);
            \draw [help lines] (0,2)--(2,2);
            \draw [help lines] (0.9,0)--(0.9,1.3);
            \draw [help lines] (1.8,0)--(1.8,2.35);
            \draw [help lines] (0,1.3) -- (0.9,1.3);
            \draw [help lines] (0,2.35) -- (1.8,2.35);
            \draw [help lines] (0.5,0) -- (0.5,0.5);
            \draw [help lines] (0,0.5) -- (0.5,0.5);
            \draw [help lines] (0,0.8) -- (2,0.8);
            \draw [fill](0.5,0) circle [radius=0.03];
            \draw [fill](0.9,0) circle [radius=0.03];
            \draw [fill](1.475,0) circle [radius=0.03];
            \draw [fill](1.8,0) circle [radius=0.03];
            \end{tikzpicture}
        \end{adjustbox}
        \label{fig:WorkloadDecline2Subcase2}
        }
    \centering
        \subfigure[$W(b_{j+1};R_o)\leq T_m$]{
        \begin{adjustbox}{width=0.3\textwidth}
        \begin{tikzpicture}
           \draw [thick,<->] (2.2,0)--(0,0)--(0,2.8);
            \node [left] at (0,2.8) 
            {\begin{footnotesize}
            $W(t;R_{\rm o})$
            \end{footnotesize}};
            \node [left] at (0,1.3) 
            {\begin{scriptsize}
            $W(b_{j};R_{\rm o})$
            \end{scriptsize}};
            \node [left] at (0,1.8) 
            {\begin{scriptsize}
            $W(b_{j+1};R_{\rm o})$
            \end{scriptsize}};
            \node [right] at (2.2,0) {\footnotesize{$t$}};
            \node [below] at (0.9,0) 
            {\begin{scriptsize}
            $b_{j}$
            \end{scriptsize}};
            \node [left,below] at (1.5,0) 
            {\begin{scriptsize}
            $t_{j+1}$
            \end{scriptsize}};
            \node [right,below] at (2.2,0) 
            {\begin{scriptsize}
            $b_{j+1}$
            \end{scriptsize}};
            \node [left] at (0,2.2) 
            {\begin{scriptsize}
            $T_{m}$
            \end{scriptsize}};
            \node [left] at (0,2.4) 
            {\begin{scriptsize}
            $\sigma_{m+1}$
            \end{scriptsize}};
            \node [left] at (0,0.8) 
            {\begin{scriptsize}
            $\sigma_{m}$
            \end{scriptsize}};
            \draw [dashed] (0.3,0.1)--(0.5,0.5);
            \draw (0.5,0.5)--(0.9,1.3);
            \draw (0.9,1.3)--(1.5,1);
            \draw (1.5,1)--(1.9,1.8);
            \draw [help lines] (0,0.8) -- (2,0.8);
            \draw [help lines] (0.9,0) -- (0.9,1.3);
            \draw [help lines] (1.5,0) -- (1.5,1);
            \draw [help lines] (1.9,0) -- (1.9,1.8);
            \draw [help lines] (0,1.3) -- (0.9,1.3);
            \draw [help lines] (0,2.2) -- (2.2,2.2);
            \draw [help lines] (0,2.4) -- (2.2,2.4);
            \draw [help lines] (0,1.8) -- (1.9,1.8);
            \draw [fill](1.5,0) circle [radius=0.03];
            \draw [fill](1.9,0) circle [radius=0.03];
            \draw [fill](0.9,0) circle [radius=0.03];
            \end{tikzpicture}
        \end{adjustbox}
        \label{fig:WorkloadDecline2Subcase3}
        }
        \caption{Two subcases I, II, III of $W_\rho(t;R_{\rm o})$, when $\sigma^*(j)=\sigma_m$ and $\sigma^*(j+1)=\sigma_{m+1}$.}
        \label{fig:WorkloadDecline2}
    \end{figure*}
Therefore, according to Proposition~\ref{prop:overshoot} and Fig.~\ref{fig:WorkloadDecline2}, we have
\begin{align*}
&O_{T_i}(t;R_{\rm o})= O_{T_i}(b_j;R_{\rm o})
\nonumber\\
&~~~\forall~t\in[b_j,b_{j+1}],~i\in\{m+1,\ldots,M\}.
\end{align*}
Therefore, according to~\eqref{eq:OvershootIneqaulity1}
\begin{align}
 &o_{T_i}(t)\leq  \bar{f}(T_i),  
 \nonumber\\
&~~~\forall~t\in[b_j,b_{j+1}],~i\in\{m+1,\ldots,M\}.
\label{eq:OvershootIneqaulityCase21}   
\end{align}
On the other hand, for all subcases as 
\begin{align*}
W(t;R_{\rm o})>T_i~~\forall~t\in[b_j,b_{j+1}],~~i\in\{1,2,\ldots,m-1\}.
\end{align*}
Therefore, according to Proposition~\ref{prop:overshoot} and Fig.~\ref{fig:WorkloadDecline2}, we have
\begin{align*}
&O_{T_i}(t;R_{\rm o})= O_{T_i}(b_j;R_{\rm o})+(t-b_j)
\nonumber\\
&~~~\forall~t\in[b_j,b_{j+1}],~i\in\{1,2,\ldots,m-1\}.
\end{align*}
Therefore, for $i\in\{1,2,\ldots,m-1\}$
\begin{align}
\max_{t \in[b_{j},b_{j+1}]}& o_{T_i}(t)=o_{T_i}(b_{j+1})
\label{eq:maxOvershootRatioCase2j+1}
\end{align}
But as $\sigma^*(j+1)=\sigma_{m+1}$ is chosen using~\eqref{eq:istar appendix} we have 
\begin{align}
&o_{T_{i}}(b_{j+1}) \leq \bar{f}(T_{i})-\epsilon_{i,j+1}(\sigma_{m+1})
\nonumber\\
&\hspace{8em} \forall i \in \{ 1,2,\ldots,m-1 \},
\label{eq:OvershootIneqaulityassumptionCase2}
\end{align}
Therefore according to~\eqref{eq:maxOvershootRatioCase2j+1} and~\eqref{eq:OvershootIneqaulityassumptionCase2}
\begin{align}
&o_{T_i}(t)<\bar{f}(T_i)
\nonumber\\
&~~~\forall~t\in[b_j,b_{j+1}],~i\in\{1,2,\ldots,m-1\}.
\label{eq:OvershootIneqaulityCase23}    
\end{align}
In subcases I and II , as 
\begin{equation*}
    W(t;R_{\rm o})\geq T_m~~~~\forall~t\in[\mu,b_{j+1}],
\end{equation*}
where $\mu$ is defined as follows:
\begin{equation*}
 W(\mu;R_{\rm o})= T_m,~~~\mu\in[b_j,b_{j+1}].   
\end{equation*}
Therefore, 
\begin{align*}
O_{T_m}(t;R_{\rm o})= \left\lbrace\begin{array}{ll}
   O_{T_m}(b_j;R_{\rm o})  & t\in[b_j,\mu] \\
   O_{T_m}(\mu;R_{\rm o})+(t-\mu)  & t\in(\mu,b_{j+1}]
\end{array}\right.
\end{align*}
Hence, 
\begin{align*}
\argmax_{t \in[b_{j},b_{j+1}]}& o_{T_m}(t)\in\{b_{j},b_{j+1}\}.
\end{align*}
But as $\sigma^*(j+1)=\sigma_{m+1}$ is chosen using~\eqref{eq:istar appendix}, therefore 
\begin{align*}
o_{T_m}(b_{j+1}) &\leq \bar{f}(T_{m})-\epsilon_{m,j+1}(\sigma_{m+1})
=\bar{f}(T_{m+1})<\bar{f}(T_{m})
\end{align*}
Also as we have~\eqref{eq:OvershootIneqaulity1} therefore,
\begin{align}
o_{T_m}(t)  \leq \bar{f}(T_{m}),~~~\forall~t\in[b_j,b_{j+1}].
\label{eq:OvershootIneqaulityCase22}
\end{align}
For subcase III, on the other hand 
\begin{align*}
W(t;R_{\rm o})<T_m~~\forall~t\in[b_j,b_{j+1}].
\end{align*}
Therefore, it can easily be shown 
\begin{align}
o_{T_m}(t)  \leq \bar{f}(T_{m})~~~\forall~t\in[b_j,b_{j+1}].
\label{eq:OvershootIneqaulityCase24}
\end{align}
Therefore, using~\eqref{eq:OvershootIneqaulityCase21},~\eqref{eq:OvershootIneqaulityCase23},~\eqref{eq:OvershootIneqaulityCase22} and~\eqref{eq:OvershootIneqaulityCase24} we have 
\begin{align}
o_{T_m}(t) \leq \bar{f}(T_{m})
&~\forall~t\in[b_j,b_{j+1}],\nonumber\\
&~i\in\{1,\ldots,M\}.
\label{eq:OvershootIneqaulityCase2final}
\end{align}
\end{proof}
\end{appxlemma}

\begin{appxlemma}
\label{appxlem:linearupperbound}
If 
\begin{align}
    o_{T_{m}}(t) \leq c_1;~~o_{T_{m+1}}(t)\leq c_2,
    \label{eq:Ij_3at2T}
\end{align}
for $m\in\{1,2,\ldots,M-2\}$, and $t\in[b_{j-1},b_j(\sigma_\ell)]$. Then 
\begin{align}
    o_{\gamma}(t)\leq \bar{f}(T_m)-(\gamma-T_m)\frac{c_1-c_2}{T_{m+1}-\sigma_{m+1}}, 
    \label{eq:Ij_3at1stInterval}
\end{align}
for $\forall\gamma\in[\sigma_{m+1},T_{m+1})$ and 
\begin{equation}
o_{\gamma}(t) < c_1,~~~~~ \forall\gamma\in[T_m,\sigma_{m+1}). 
    \label{eq:Ij_3at2ndInterval}    
\end{equation}
\end{appxlemma}
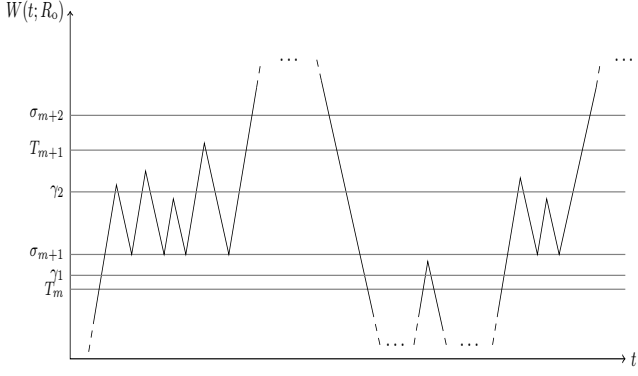
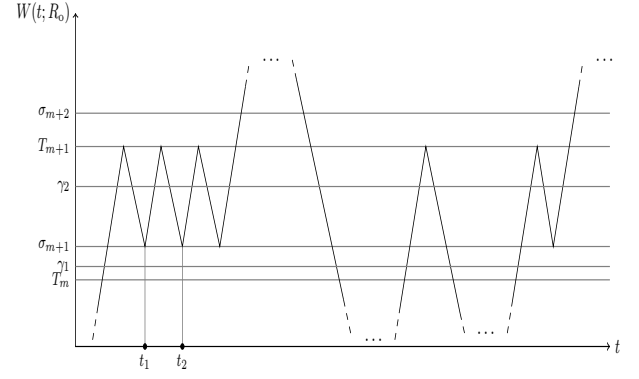
\begin{figure*}
\centering
    \subfigure[$O_{\gamma_1,T_{m+1}}(t)$ is not at its maximum value for a constant $O_{T_{m},T_{m+1}}(t)=c$]{
    \centering
    \begin{adjustbox}{width=0.47\textwidth,height=2in}
    \begin{tikzpicture}
    \draw [<->] (0,5)--(0,0)--(12,0);
    \draw [help lines] (0,1)--(12,1);
    \draw [help lines] (0,1.5)--(12,1.5);
    \draw [help lines] (0,3)--(12,3);
    \draw [help lines] (0,3.5)--(12,3.5);
    \draw [dashed](0.4,0.1)--(0.5,0.5);
    \draw (0.5,.5)--(1,2.5)--(1.33,1.5)--(1.63,2.7)--(2.03,1.5)--(2.23,2.3)--(2.5,1.5)--(2.9,3.1)--(3.43,1.5)--(4.03,3.9);
    \draw [dashed] (4.03,3.9)--(4.13,4.3);
    \node at (4.73,4.3) {$\ldots$};
    \draw [dashed] (5.33,4.3)--(5.43,4);
    \draw (5.43,4)--(6.53,0.7);
    \draw [dashed] (6.53,0.7)--(6.73,0.1);
    \node at (7.03,0.2) {$\ldots$};
    \draw [dashed] (7.43,0.2)--(7.53,0.6);
    \draw (7.53,0.6)--(7.73,1.4)--(8.03,0.5);
    \draw [dashed] (8.03,0.5)--(8.13,0.2);
    \node at (8.63,0.2) {$\ldots$};
    \draw [dashed] (9.13,0.2)--(9.23,0.6);
    \draw (9.23,0.6)--(9.73,2.6)--(10.1,1.5)--(10.3,2.3)--(10.57,1.5)--(11.37,3.9);
    \draw [dashed] (11.37,3.9)--(11.47,4.3);
    \node at (11.97,4.3) {$\ldots$};
    \draw [help lines] (0,2.4)--(12,2.4);
    \draw [help lines] (0,1.2)--(12,1.2);
    \node [left] at (0,2.4) {\footnotesize{$\gamma_2$}};
    \node [left] at (0,1.2) {\footnotesize{$\gamma_1$}};
    \node [left] at (0,5) {\small{$W(t;R_{\rm o})$}};
    \node [right] at (12,0) {\small{$t$}};
    \node [left] at (0,1) {\footnotesize{$T_m$}};
    \node [left] at (0,1.5) {\footnotesize{$\sigma_{m+1}$}};
    \node [left] at (0,3) {\footnotesize{$T_{m+1}$}};
    \node [left] at (0,3.5) {\footnotesize{$\sigma_{m+2}$}};
    \end{tikzpicture}
    \end{adjustbox}
    \label{fig:workload_fluctuationCase1}
    }
    \hspace{1em}
    \subfigure[$O_{\gamma_1,T_{m+1}}(t)$ is at its maximum value for a constant $O_{T_{m},T_{m+1}}(t)=c$]{
    \centering
    \begin{adjustbox}{width=0.45\textwidth,height=2in}
    \begin{tikzpicture}
    \draw [<->] (0,5)--(0,0)--(12.5,0);
    \draw [help lines] (0,1)--(12.5,1);
    \draw [help lines] (0,1.5)--(12.5,1.5);
    \draw [help lines] (0,3)--(12.5,3);
    \draw [help lines] (0,3.5)--(12.5,3.5);
    \draw [dashed](0.4,0.1)--(0.5,0.5);
    \draw (0.5,.5)--(1.125,3)--(1.625,1.5)--(2,3)--(2.5,1.5)--(2.875,3)--(3.375,1.5)--(3.975,3.9);
    \draw [dashed] (3.975,3.9)--(4.075,4.3);
    \node at (4.575,4.3) {$\ldots$};
    \draw [dashed] (5.07,4.3)--(5.17,4);
    \draw (5.17,4)--(6.27,0.7);
    \draw [dashed] (6.27,0.7)--(6.47,0.1);
    \node at (6.97,0.1) {$\ldots$};
    \draw [dashed] (7.47,0.1)--(7.57,0.5); 
    \draw (7.57,0.5)--(8.195,3)--(8.995,0.5);
    \draw [dashed] (8.995,0.5)--(9.1,0.2);
    \node at (9.6,0.2) {$\ldots$};
    \draw [dashed] (10.1,0.2)--(10.2,0.6);
    \draw (10.2,0.6)--(10.8,3)--(11.175,1.5)--(11.775,3.9);
    \draw [dashed](11.775,3.9)--(11.875,4.3);
    \node at (12.375,4.3) {$\ldots$};
    \draw [help lines] (0,2.4)--(12.5,2.4);
    \draw [help lines] (0,1.2)--(12.5,1.2);
    \node [left] at (0,2.4) {\footnotesize{$\gamma_2$}};
    \node [left] at (0,1.2) {\footnotesize{$\gamma_1$}};
    \node [left] at (0,5) {\small{$W(t;R_{\rm o})$}};
    \node [right] at (12.5,0) {\small{$t$}};
    \node [left] at (0,1) {\footnotesize{$T_m$}};
    \node [left] at (0,1.5) {\footnotesize{$\sigma_{m+1}$}};
    \node [left] at (0,3) {\footnotesize{$T_{m+1}$}};
    \node [left] at (0,3.5) {\footnotesize{$\sigma_{m+2}$}};
    \draw [help lines] (1.625,0)--(1.625,1.5);
    \draw [help lines] (2.5,0)--(2.5,1.5);
    \node [below] at (1.625,0) {\footnotesize{$t_1$}};
    \node [below] at (2.5,0) {\footnotesize{$t_2$}};
    \draw [fill](1.625,0) circle [radius=0.04];
    \draw [fill](2.5,0) circle [radius=0.04];
    \end{tikzpicture}
    \end{adjustbox}
    \label{fig:workload_fluctuationCase2}
    }
    \caption{Fluctuation of the $W(t;R_{\rm o})$ between $T_{m}$ and $T_{m+1}$}
    \label{fig:workload_fluctuation}
\end{figure*}
\begin{proof}
We prove this Lemma for two following cases
\setcounter{case}{0}
\begin{case}
$\gamma\in[\sigma_{m+1},T_{m+1})$.
\end{case}
We know according to~\eqref{eq:Ij_3at2T}
\begin{align}
 O_{T_{m}}(t;R_{\rm o})\leq t c_1;~~O_{T_{m+1}}(t;R_{\rm o})\leq t c_2. 
\label{eq:overshootbounds}
\end{align}
For the simplification of the proof we extend the concept of the overshoot to the overshoot duration with respect to two threshold values.
\begin{definition}
\label{def:Overshoot Duration for two threshold}
Given two threshold values $\zeta_2>\zeta_1 >0$ and a traffic stream $R$,
a {\em limited overshoot interval} with respect to $R$, $\zeta_1$ and $\zeta_2$ is a maximal interval of time $\kappa$
such that $\zeta_2>W_\rho(\tau; R) \geq \zeta_1$ for all $ \tau \in \kappa$.  Let $|\kappa|$ denote the
length of interval $\kappa$.  Let $\mathcal{O}_{[\zeta_1, \zeta_2)}(t)$ denote the set of limited overshoot intervals
contained in $[0, t]$.  Then the {\em limited overshoot duration} up to time~$t$ is defined as 
\begin{align}
O_{\zeta_1,\zeta_2}(t;R) = \sum_{\kappa \in \mathcal{O}_{[\zeta_1, \zeta_2)}(t)} |\kappa| .
 \label{eq:total overshoot duration for two threshold} 
\end{align}
\end{definition}
According to the Definition~\ref{def:Overshoot Duration} and~\ref{def:Overshoot Duration for two threshold} it is obvious that 
\begin{equation*}
 O_{\zeta_2}(t;R)=O_{\zeta_1,\zeta_2}(t;R)+O_{\zeta_1}(t;R)  
\end{equation*}
On the other hand, for a fixed $O_{T_{m}}(t;R)-O_{T_{m+1}}(t;R)=O_{T_m,T_{m+1}}(t;R)=c$, for any $\gamma\in[T_M,T_{M+1}]$, $O_{\gamma}(t;R_{\rm o})$ is maximized when  $O_{\gamma,T_{m+1}}(t;R)$ is at its maximum value. But we should note that as shown in Fig.~\ref{fig:deterministic regulator delay and workload} and according to equations~\eqref{eq:Stochastic Regulator eq 4}-\eqref{eq:Stochastic Regulator eq 6} the workload $W(t;R_{\rm o})$ can fluctuate between $T_M$ and $T_{M+1}$ as shown in Fig.~\ref{fig:workload_fluctuation}. It can be seen by comparing Fig.~\ref{fig:workload_fluctuationCase1} and~\ref{fig:workload_fluctuationCase2} that $O_{\gamma,T_{m+1}}(t;R)$ is greater in Fig.~\ref{fig:workload_fluctuationCase2} in compare to Fig.~\ref{fig:workload_fluctuationCase1}. In other words, we should have the fluctuation of $W(t;R_{\rm o})$ between $T_{m}$ and $T_{m+1}$ in units of the complete fluctuation as shown in Fig.~\ref{fig:workload_completefluctuation}. By considering the increasing slope of $W(t;R_{\rm o})$ as $C-\rho$ and the decreasing slope of $-\rho$, it can easily be shown that
\begin{figure}
    \centering
    \begin{tikzpicture}
    \draw [<->] (0,4)--(0,0)--(3,0);
    \draw [help lines] (0,1)--(3,1);
    \draw [help lines] (0,1.5)--(3,1.5);
    \draw [help lines] (0,3)--(3,3);
    \draw [help lines] (0,3.5)--(3,3.5);
    \draw (0.75,1.5)--(1.125,3)--(1.625,1.5);
    \draw [help lines] (0,2.4)--(3,2.4);
    \draw [help lines] (0,1.2)--(3,1.2);
    \node [left] at (0,2.4) {\footnotesize{$\gamma_2$}};
    \node [left] at (0,1.2) {\footnotesize{$\gamma_1$}};
    \node [left] at (0,4) {\small{$W(t;R_{\rm o})$}};
    \node [right] at (3,0) {\small{$t$}};
    \node [left] at (0,1) {\footnotesize{$T_m$}};
    \node [left] at (0,1.5) {\footnotesize{$\sigma_{m+1}$}};
    \node [left] at (0,3) {\footnotesize{$T_{m+1}$}};
    \node [left] at (0,3.5) {\footnotesize{$\sigma_{m+2}$}};
    \draw [help lines] (0.75,0)--(0.75,1.5);
    \draw [help lines] (1.625,0)--(1.625,1.5);
    \node [below] at (0.75,0) {\footnotesize{$t_1$}};
    \node [below] at (1.625,0) {\footnotesize{$t_2$}};
    \draw [help lines] (0.975,0)--(0.975,2.4);
    \draw [help lines] (1.325,0)--(1.325,2.4);
    \node [below] at (0.975,0) {\footnotesize{$\tau_1$}};
    \node [below] at (1.325,0) {\footnotesize{$\tau_2$}};
    \draw [fill](0.75,0) circle [radius=0.04];
    \draw [fill](0.975,0) circle [radius=0.04];
    \draw [fill](1.325,0) circle [radius=0.04];
    \draw [fill](1.625,0) circle [radius=0.04];
    \end{tikzpicture}
    \caption{One unit of complete fluctuation of $W(t;R{\rm o})$ between $T_{m}$ and $T_{m+1}$}
    \label{fig:workload_completefluctuation}
    \end{figure}
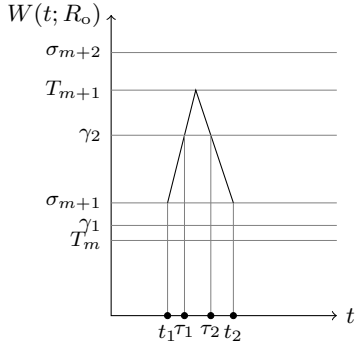
\begin{align}
&\Delta t=t_2-t_1=(T_{m+1}-\sigma_{m+1})\left(\frac{1}{\rho}+\frac{1}{C-\rho}\right) 
\label{eq:CompleteFluctuationDuration}
\\
&\Delta \tau=\tau_2-\tau_1=(T_{m+1}-\gamma_2)\left(\frac{1}{\rho}+\frac{1}{C-\rho}\right) 
\label{eq:FluctuationDuration}
\end{align}
Therefore, if $O_{T_m,T_{m+1}}(t;R)=c$, in order to maximize $O_{\gamma,T_{m+1}}(t;R)$ we should have $n_1$ complete fluctuation interval, where $n_1$ is
\begin{equation}
    n_1=\left\lfloor\frac{c}{\Delta t}\right\rfloor.
    \label{eq:n}
\end{equation}
For example in Fig.~\ref{fig:workload_fluctuationCase2}, $n_1=6$.
Therefore, 
\begin{equation}
O_{\gamma,T_{m+1}}(t;R)\leq \frac{c\Delta \tau}{\Delta t}
\label{eq:O_gamma_Tm+1_bound}
\end{equation}
This upper bound is tight and can happen for $M$ being at its maximum value, such that $\sigma_{m+1}\approx T_m$, and for the fluctuation of $W(t;R_{\rm o})$ as in Fig.~\ref{fig:workload_fluctuationCase2}. Hence,
\begin{align}
&o_{\gamma}(t)=\frac{O_{\gamma,T_{m+1}}(t;R)+O_{T_{m+1}}(t;R)}{t}\nonumber\\
&\leq \frac{c\Delta \tau}{t\Delta t}+o_{T_{m+1}}(t)=o_{T_{m}}(t)
\nonumber\\
&-(\gamma-\sigma_{m+1})\frac{o_{T_{m}}(t)-o_{T_{m+1}}(t)}{T_{m+1}-\sigma_{m+1}}
\label{eq:O_gamma_bound}
\end{align}
It can easily be shown with the constraint of~\eqref{eq:Ij_3at2T}, we have 
\begin{equation}
 \frac{O_{\gamma}(t;R)}{t}\leq c_1- (\gamma-\sigma_{m+1})\frac{c_1-c_2}{T_{m+1}-\sigma_{m+1}}
\end{equation}
\begin{case}
$\gamma\in[T_m,\sigma_{m+1})$.
\end{case}
It can be seen in Fig.~\ref{fig:workload_fluctuationCase1} and Fig.~\ref{fig:workload_fluctuationCase2}, that
\begin{equation}
 O_{\gamma}(t;R)  < O_{T_m}(t;R) 
 \label{eq:overshootapproximate}
\end{equation}
Therefore, 
\begin{equation}
 o_{\gamma}(t)< o_{T_m}(t)\leq f(T_m)  
 \label{eq:overshootratioapproximate}
\end{equation}
If we choose $M$ at its maximum level the bounding function can be a linear function between any two points $T_m$ and $T_{m+1}$.
\end{proof}

\begin{appxcorollary}
\label{appxcorollary:1}
Let 
\begin{align}
    o_{T_{m}}(t) \leq  \bar{f}(T_{m});~~o_{T_{m+1}}(t) \leq \bar{f}(T_{m+1}),
\end{align}
for $m\in\{1,2,\ldots,M-2\}$, and $t\in[b_{j-1},b_j(\sigma_\ell)]$ and $M$ is chosen as the maximum value such that $T_i\approx\sigma_{i+1}$ for $i\in\{1,2,\ldots,M-1\}$. Then 
\begin{align}
    o_{\gamma}(t)\leq \bar{f}(T_{m})-(\gamma-T_m)\frac{\bar{f}(T_{m})-\bar{f}(T_{m+1})}{T_{m+1}-T_m},
\end{align}
for $\forall\gamma\in[T_{m},T_{m+1})$. 
\end{appxcorollary}
\begin{remark}
For the case that $M$ is not chosen as the maximum possible value and $T_i<\sigma_{i+1}$ for some $i\in\{1,2,\ldots,M-1\}$, by slightly modifying the definition of the $\bar{f}(\gamma)$, we can get a result similar to Corollary~\ref{appxcorollary:1}. 
In this modification, in the interval $[\sigma_{i+1}, T_{i+1})$ let 
\begin{align}
    l_i(\gamma) := f(T_{i+1}) + \omega_i (\gamma - T_{i+1})
\end{align}
represent the line connecting the points $(\sigma_{i+1}, f(\sigma_{i+1}))$ and $(T_{i+1}, f(T_{i+1}))$
with slope 
\begin{align}
\hat{\omega}_i := \frac{f(T_{i+1}) - f(\sigma_{i+1})}{T_{i+1} - \sigma_{i+1}}
\end{align} 
for $i=1, \ldots, M-2$.  If $f(\gamma) \geq l_i(\gamma)$
for all $\gamma \in [\sigma_{i+1}, T_{i+1})$ we set $\bar{f} = l_i$ in this interval.  Otherwise, we
set $\bar{f} = h_i$ on $[\sigma_{i+1}, T_{i+1})$, where
\begin{align}
    h_i(\gamma) = f(T_{i+1}) + f'(T_{i+1}) (\gamma - T_{i+1}) .
\end{align}
On the other hand, in the interval $[T_i,\sigma_{i+1})$ we set $\bar{f} = f(\sigma_{i+1})$. Similarly, we  set $\bar{f}(\gamma) = f(T)$ for
$\gamma \in [T_{M-1}, T_{M}]$ and $\bar{f}(\gamma) = 0$ for $\gamma > T_{M}$.  To summarize,
we define
\begin{align}
\bar{f}(\gamma) :=\left\{\begin{array}{ll}
    1, &  \gamma\in[0,T_1) , \\
   f(\sigma_{i+1}),  &  \gamma\in[T_i,\sigma_{i+1}) , \\
   f(T_{i+1}) + \hat{m}_i(\gamma-T_{i+1}),  &  \gamma\in [\sigma_{i+1},T_{i+1}) ,\\
    f(T) , &  \gamma\in [T_{M-1},T_{M}], \\
    0,         &  \gamma>T_M ,  
\end{array}\right.
\label{eq:f_barnew}
\end{align}
where the slopes $\hat{m}_i$ are defined by
\begin{align}
    \hat{m}_i =
    \left \{
      \begin{array}{ll}
         \hat{\omega}_i,        & \mbox{if $f \geq h_i$ on $[\sigma_{i+1}, T_{i+1})$} , \\
         f'(T_{i+1}), & \mbox{otherwise},
  \end{array}
  \right . 
\end{align}
for $i = 1, \ldots, M-2$. This modified $\Bar{f}(\gamma)$ is shown in Fig.~\ref{fig:fbarnew}.
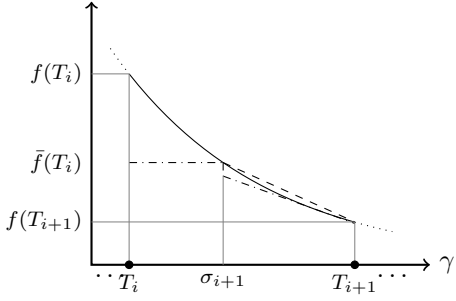
\begin{figure}
\centering
\begin{tikzpicture}
\draw [thick,<->] (4.5,0)--(0,0)--(0,3.5);
\node [right] at (4.5,0) {$\gamma$};
\draw [domain=0.5:3.5] plot (\x,{1.5*(0.8*exp(-1*(\x/2-1)))});
\draw [domain=0.25:0.5][dotted] plot (\x,{1.5*(0.8*exp(-1*(\x/2-1)))});
\draw [domain=3.5:4][dotted] plot (\x,{1.5*(0.8*exp(-1*(\x/2-1)))});
\draw [domain=0.5:3.5][dashed] (1.75,1.36)--(3.5,0.57);
\draw [domain=1.75:3.5][dashdotted] plot (\x,{-0.36*(\x-3.5)+0.55});
\draw [dashdotted] (1.75,1.36)--(1.75,1.18);
\draw [dashdotted] (0.5,1.36)--(1.75,1.36);
\draw [help lines] (3.5,0) -- (3.5,0.57);
\draw [help lines] (0.5,0) -- (0.5,2.54);
\draw [help lines] (0,0.57) -- (3.5,0.57);
\draw [help lines] (0,2.54) -- (0.5,2.54);
\draw [help lines] (1.75,0) -- (1.75,1.18);
\draw [fill](3.5,0) circle [radius=0.05];
\draw [fill](0.5,0) circle [radius=0.05];
\node [below] at (3.5,0){\begin{footnotesize}{$T_{i+1}$}\end{footnotesize}};
\node [below] at (0.5,0){\begin{footnotesize}{$T_{i}$}\end{footnotesize}};
\node [left] at (0,0.59){\begin{footnotesize}{$f(T_{i+1})$}\end{footnotesize}};
\node [left] at (0,2.54){\begin{footnotesize}{$f(T_{i})$}\end{footnotesize}};
\node [left] at (0,1.36){\begin{footnotesize}{$\bar{f}(T_{i})$}\end{footnotesize}};
\node [below] at (1.75,0){\begin{footnotesize}{$\sigma_{i+1}$}\end{footnotesize}};
\node [below] at (0.25,0) {$\ldots$};
\node [below] at (4,0) {$\ldots$};
\end{tikzpicture}
\caption{Modified definition of $\bar{f}(\gamma)$ when $M<M_{\max}$.}
\label{fig:fbarnew}
\end{figure}
\end{remark}
\begin{appxcorollary}
\label{appxcorollary:2}
Let 
\begin{align}
    o_{T_{m}}(t) \leq  \bar{f}(T_{m});~~o_{T_{m+1}}(t) \leq \bar{f}(T_{m+1}),
\end{align}
for $m\in\{1,2,\ldots,M-2\}$, and $t\in[b_{j-1},b_j(\sigma_\ell)]$ and $\bar{f}(\gamma)$ is defined as~\eqref{eq:f_barnew}. Then 
\begin{align}
    o_{\gamma}(t)\leq \bar{f}(\sigma_{m+1})-(\gamma-T_m)\frac{\bar{f}(\sigma_{m+1})-\bar{f}(T_{m+1})}{T_{m+1}-\sigma_{m+1}}, 
\end{align}
for $\forall\gamma\in[\sigma_{m+1},T_{m+1})$ and  
\begin{align}
    o_{\gamma}(t)\leq \bar{f}(\sigma_{m+1}),
\end{align}
for $\forall\gamma\in[T_{m},\sigma_{m+1})$.
\end{appxcorollary}
\begin{proof}[Proof of Theorem~\ref{thm:overshoot_soln2}]
In Lemma~\ref{appxlem:overshoot_soln2_discreet T_i} we showed if $\sigma^*(j)$ is chosen using~\eqref{eq:sigma_j_soln1 appendix} then 
\begin{align}
&o_{T_i}(t)\leq \bar{f}(T_i)
\nonumber\\
&~~~\forall~t\in[b_{j-1},b_{j}(\sigma^*(j))],~i\in\{1,2,\ldots,M\}.
\label{eq:overshootdiscreetTbound}
\end{align}
On the other hand, we showed in Corollary~\ref{appxcorollary:1} that if we have ~\eqref{eq:overshootdiscreetTbound} and $M=M_{\max}$, defined in~\eqref{eq:M_max_value}, then 
\begin{align}
    o_{\gamma}(t)\leq \bar{f}(T_{i})-(\gamma-T_i)\frac{\bar{f}(T_{i})-\bar{f}(T_{i+1})}{T_{i+1}-T_i},
\end{align}
for $\forall\gamma\in[T_{i},T_{i+1})$ and $\forall i\in\{1,\ldots,M\}$. On the other hand, in Corollary~\ref{appxcorollary:2} we showed, if if we have ~\eqref{eq:overshootdiscreetTbound} and $M<M_{\max}$, then with the modified definition of $\Bar{f}(\gamma)$ in~\eqref{eq:f_barnew}, \begin{align}
    o_{\gamma}(t)\leq \bar{f}(\sigma_{i+1})-(\gamma-T_i)\frac{\bar{f}(\sigma_{i+1})-\bar{f}(T_{i+1})}{T_{i+1}-\sigma_{i+1}}, 
\end{align}
for $\forall\gamma\in[\sigma_{i+1},T_{i+1})$ and  
\begin{align}
    o_{\gamma}(t)\leq \bar{f}(\sigma_{i+1}),
\end{align}
for $\forall\gamma\in[T_{i},\sigma_{i+1})$ and $\forall i\in\{1,\ldots,M\}$. Therefore, for $M=M_{\max}$ if for all $i \in \{ 1,2,\ldots,M-1 \}$ and all $\gamma \in [T_i,T_{i+1}]$,
\begin{align}
f(\gamma)\geq   \bar{f}(T_{i})-(\gamma-T_i)\frac{\bar{f}(T_{i})-\bar{f}(T_{i+1})}{T_{i+1}-T_i},
\label{eq:f_fbar_inequality1}
\end{align}
with $\bar{f}(\gamma)$ defined in~\eqref{eq:f_bar}, and for $M<M_{\max}$ if for all $i \in \{ 1,2,\ldots,M-1 \}$ and all $\gamma \in [\sigma_{i+1},T_{i+1})$
\begin{align}
    f(\gamma)\geq \bar{f}(\sigma_{i+1})-(\gamma-T_i)\frac{\bar{f}(\sigma_{i+1})-\bar{f}(T_{i+1})}{T_{i+1}-\sigma_{i+1}}, 
    \label{eq:f_fbar_inequality2}
\end{align}
and if for all $\gamma \in [T_{i},\sigma_{i+1})$
\begin{align}
    f(\gamma)\geq \bar{f}(\sigma_{i+1}),
    \label{eq:f_fbar_inequality3}
\end{align}
with $\bar{f}(\gamma)$ defined in~\eqref{eq:f_barnew}, then 
\begin{align}
o_{\gamma}(t) \leq f(\gamma),
~~\forall~t\in[b_{j-1},b_{j}(\sigma^*(j))],~~ \forall \gamma \in [T_1, T].
\end{align}
But definition of $\bar{f}(\gamma)$ assures inequalities in~\eqref{eq:f_fbar_inequality1}--\eqref{eq:f_fbar_inequality3}. 
\end{proof}

\section{Proof of Theorem~\ref{thm:reducing complexity}}
\label{appx:reducing complexity}
\numberwithin{equation}{section}
\setcounter{equation}{0}
In order to prove Theorem~\ref{thm:reducing complexity} we first establish the following lemma:
\begin{appxlemma}
\label{appxlem:4}
Let $\mathcal{B}_j$ be as define in~\eqref{eq:cBj}, let $k=\min\mathcal{B}_j$ and $\mathcal{J}_j$ be as defined in~\eqref{eq:Jj}. If $m\in\mathcal{J}_j$ and $m<k-1$ then 
\begin{align}
\frac{O_{T_{m}}(b_j(\sigma_\ell);R_{\rm o})}{b_j(\sigma_\ell)}+&
\frac{W_\rho(b_{j}(\sigma_\ell);R_{\rm o})-T_m}{\rho b_{j}(\sigma_\ell)}(1-\bar{f}(T_m))
\nonumber\\
&  \leq \bar{f}(T_m), 
\label{eq:OvershootRatioPlusEpsInequality}
\end{align}
for $\forall\ell\in\{m+1,\ldots,k\}$.
\end{appxlemma}
\begin{proof}
Note that, 
\begin{align}
\frac{W_\rho(b_{j}(\sigma_\ell);R_{\rm o})-T_m}{\rho b_{j}(\sigma_\ell)}(1-f(T_m))=\epsilon_{m,j}(\sigma_\ell),    
\end{align}
for $\forall\ell\in\{m+2,\ldots,k\}$. According to~\eqref{eq:delay_regulator}, and as it is shown in Fig.~\ref{fig:OvershootRatioEquality}, it can easily be verified that
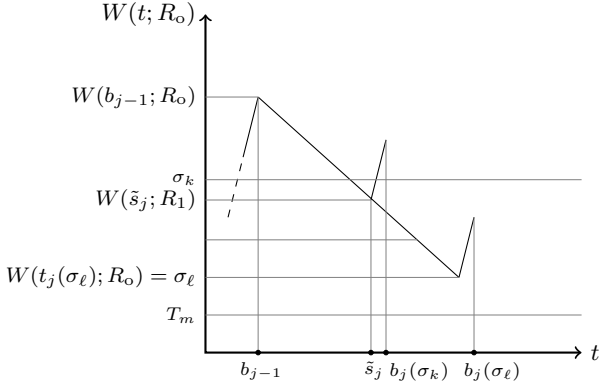
\begin{figure}[t]
    \centering
        \begin{adjustbox}{scale=1}
        \begin{tikzpicture}
           \draw [thick,<->] (5,0)--(0,0)--(0,4.5);
            \node [left] at (0,4.5) 
            {\begin{small}
            $W(t;R_{\rm o})$
            \end{small}};
            \node [left] at (0,3.4) 
            {\begin{footnotesize}
            $W(b_{j-1};R_{\rm o})$
            \end{footnotesize}};
            \node [right] at (5,0) {\small{$t$}};
            \node [below] at (0.7,0) 
            {\begin{scriptsize}
            $b_{j-1}$
            \end{scriptsize}};
            \node [right,below] at (2.2,0) 
            {\begin{scriptsize}
            $\Tilde{s}_j$
            \end{scriptsize}};
            \node [left] at (0,2.03) 
            {\begin{footnotesize}
            $W(\Tilde{s}_j;R_1)$
            \end{footnotesize}};
            \node [left] at (0,1) 
            {\begin{footnotesize}
            $W(t_j(\sigma_\ell);R_{\rm o})=\sigma_\ell$
            \end{footnotesize}};
            \node [below] at (3.77,0) 
            {\begin{scriptsize}
            $b_{j}(\sigma_\ell)$
            \end{scriptsize}};
            \node [left] at (0,2.3) 
            {\begin{scriptsize}
            $\sigma_k$
            \end{scriptsize}};
            \node [below] at (2.8,0) 
            {\begin{scriptsize}
            $b_j(\sigma_k)$
            \end{scriptsize}};
            \node [left] at (0,0.5) 
            {\begin{scriptsize}
            $T_m$
            \end{scriptsize}};
            \draw [dashed] (0.3,1.8)--(0.5,2.6);
            \draw (0.5,2.6)--(0.7,3.4);
            \draw (0.7,3.4)--(3.37,1);
            \draw (2.2,2.03)--(2.4,2.83);
            \draw (3.37,1)--(3.57,1.8);
            \draw [help lines] (0,1)--(3.37,1);
            \draw [help lines] (3.57,0)--(3.57,1.8);
            \draw [help lines] (2.4,0)--(2.4,2.83);
            \draw [help lines] (2.2,0)--(2.2,2.03);
            \draw [help lines] (0,2.3)--(5,2.3);
            \draw [help lines] (0,2.03)--(2.2,2.03);
            \draw [help lines] (0.7,0) -- (0.7,3.4);
            \draw [help lines] (0,3.4) -- (0.7,3.4);
            \draw [help lines] (0,1.5) -- (2.82,1.5);
            \draw [help lines] (0,0.5) -- (5,0.5);
            \draw [fill](2.2,0) circle [radius=0.03];
            \draw [fill](0.7,0) circle [radius=0.03];
            \draw [fill](2.4,0) circle [radius=0.03];
            \draw [fill](3.57,0) circle [radius=0.03];
            \end{tikzpicture}
        \end{adjustbox}
        \caption{Relation between $b_j(\sigma_\ell)$ and $b_j(\sigma_k)$ for $\ell\in\{m+2,\ldots,k\}$ with $\min{\mathcal{B}_j}=k$.}
        \label{fig:OvershootRatioEquality}
    \end{figure}
\begin{equation}
    b_j(\sigma_\ell)=b_j(\sigma_k)+\frac{W(\tilde{s}_j;R_1)-\sigma_\ell}{\rho} \label{eq:bjincrease}
\end{equation}
Therefore, using Proposition~\ref{prop:overshoot}, we have
\begin{equation}
O_{T_{m}}(b_j(\sigma_\ell);R_{\rm o})=O_{T_{m}}(b_j(\sigma_k);R_{\rm o})+\frac{W(\tilde{s}_j;R_1)-\sigma_\ell}{\rho} 
\label{eq:OvershootIncrease}
\end{equation}
Also we have
\begin{equation}
W_\rho(b_{j}(\sigma_\ell);R_{\rm o})=W_\rho(b_{j}(\sigma_k);R_{\rm o})- (W(\tilde{s}_j;R_1)-\sigma_\ell)
\label{eq:WorkloadDecrease}
\end{equation}
Therefore, using~\eqref{eq:bjincrease},~\eqref{eq:OvershootIncrease} and~\eqref{eq:WorkloadDecrease} and $m\in \mathcal{J}_j$ after some manipulations we can show~\eqref{eq:OvershootRatioPlusEpsInequality} holds.
\end{proof}
\begin{proof}[Proof of Theorem~\ref{thm:reducing complexity}]
Based on $m$, defined in~\eqref{eq:m}, we have two following cases: 
\setcounter{case}{0}
\begin{case}
$m=k-1$.
\end{case}
In this case, as
$m\in\mathcal{J}_j$, therefore according to~\eqref{eq:Jj}
\begin{equation}
\frac{O_{T_{\ell}}(b_j(\sigma_k);R_{\rm o})}{b_j(\sigma_k)}\leq \bar{f}(T_{\ell})-\epsilon_{\ell,j}(\sigma_k), 
\label{eq:JjCase1}
\end{equation}
for $\ell=1,2,\ldots,k-1$.
Therefore, according to~\eqref{eq:epsilonij}
\begin{equation}
\frac{O_{T_{m}}(b_j(\sigma_k);R_{\rm o})}{b_j(\sigma_k)}\leq \bar{f}(T_m)-\epsilon_{m,j}(\sigma_k)=\bar{f}(T_k)    
\end{equation}
Hence, according to~\eqref{eq:Kj}, $m+1\in\mathcal{K}_j$. Hence, $\sigma^*(j)$ derived using Theorem~\ref{thm:reducing complexity} is $\sigma^*(j)=\sigma_k$. On the other hand, according to~\eqref{eq:JjCase1} and~\eqref{eq:Ij_2}, $\sigma^*(j)$ derived using Theorem~\ref{thm:overshoot_soln2} will be also $\sigma^*(j)=\sigma_k$.
\begin{case}
$m<k-1$.
\end{case}
 Lets assume $\sigma^*(j)$ derived using Theorem~\ref{thm:reducing complexity} and $\sigma^*(j)=\sigma_n$. We will show $\sigma^*(j)$ derived using Theorem~\ref{thm:overshoot_soln2} will be also $\sigma^*(j)=\sigma_n$. In this case according to~\eqref{eq:Kj} and~\eqref{eq:implementation_3}, if $\ell\in\{1,2,\ldots,n-1\}$ then $\ell \in\mathcal{J}_j$ and $\ell<k-1$. Therefore, according to Lemma~\ref{appxlem:4}
\begin{align}
 \frac{O_{T_{\ell}}(b_j(\sigma_n);R_{\rm o})}{b_j(\sigma_n)}&\leq \bar{f}(T_\ell)-
 \nonumber\\
&\frac{W_\rho(b_{j}(\sigma_n);R_{\rm o})-T_\ell}{\rho b_{j}(\sigma_n)}(1-\bar{f}(T_\ell)), 
\label{eq:Condition1Case2}
\end{align}
for $l=1,2,\ldots,n-1$. On the other hand, as $n\in\mathcal{K}_j$, according to~\eqref{eq:Kj} and~\eqref{eq:epsilonij}
\begin{equation}
\frac{O_{T_{n-1}}(b_j(\sigma_n);R_{\rm o})}{b_j(\sigma_n)}\leq \bar{f}(T_{n-1})=\bar{f}(T_n)-\epsilon_{n-1,j}(\sigma_{n}) 
\label{eq:Condition2Case2}
\end{equation}
Therefore, according to~\eqref{eq:Condition1Case2},~\eqref{eq:Condition2Case2} and~\eqref{eq:Condition1Case2}, $\sigma^*(j)$ derived using Theorem~\ref{thm:overshoot_soln2} will be also $\sigma^*(j)=\sigma_n$.
\end{proof}

\section{Proof of Theorem~\ref{thm:overshoot_soln1}}
\label{appx:overshoot_soln1}

\subsection{Proof of Theorem~\ref{thm:overshoot_soln1}, Part I}
\label{appx_subsec:overshoot_soln1_I}

In this section we prove the following lemma, which is a preliminary version of Theorem~\ref{thm:overshoot_soln1}. Then using the results in this appendix, we prove Theorem~\ref{thm:overshoot_soln1} in the next section. We also provide some details about the practical implementation of Algorithm~\ref{alg:Stochastic regulator} in the next section.

\begin{appxlemma}
\label{appxlem:overshoot_soln1}
Assume that $T_M$ is chosen sufficiently large such that for every packet~$j$ the set 
\begin{equation}
\cB_j = \left \{ 1 \leq \ell \leq M : \sigma_\ell \geq W_\rho(\tilde{s}_j ;  R_{\rm 1} ) \right \},
\end{equation}
is non-empty. 
Set
\begin{align}
\cI_j = \left \{  2 \leq \ell \leq \min \cB_j : o_{T_{\ell-1}}(b_j(\sigma_\ell))\leq \bar{f}(T_{\ell})\right \}
\end{align}
where $t_j(\sigma_\ell)$ and $b_j(\sigma_\ell)$ are given by~\eqref{eq:tj_R1}  and \eqref{eq:bj}, respectively.  Set
\begin{align}
  i^*  = \left \{
    \begin{array}{ll}
        \max \cI_j,  &  \cI_j \neq \emptyset , \\
        1,                   & \mbox{otherwise.}
    \end{array}
    \right .
\end{align}
Let
\begin{align}
    \sigma^*(j) = \sigma_{i^*} .
 \label{eq:sigma_j_soln1}
\end{align}
If 
\begin{equation}
b_{j}\geq \frac{L}{\epsilon\rho}+\frac{T_M-\sigma_1}{\rho}+\frac{L}{C},
\label{eq:b_jlowerbound}
\end{equation}
where $\epsilon > 0$ is given by 
\begin{equation}
\epsilon=\min_{2\leq k \leq M}[f(T_{k-1})-f(T_k)],
\label{eq:epsilonappendix}    
\end{equation}  
then
\begin{equation}
  o_{T_{i^*-1}}(t)  \leq \bar{f}(T_{i^*-1}),
   ~~\forall~t\in[b_{j-1},b_j]. 
   \label{eq:canonical_constraintPreliminaryTheorem}
\end{equation}
\end{appxlemma}
By comparing~\eqref{eq:canonical_constraintPreliminaryTheorem} and~\eqref{eq:canonical_constraint}, we can see Lemma~\ref{appxlem:overshoot_soln1} guarantees satisfying the constraint in~\eqref{eq:canonical_constraint} only for one specific value $\gamma=T_{i^*-1}$ rather than $\forall\gamma\in[0,T]$.
Proof of Lemma~\ref{appxlem:overshoot_soln1} is based on the following three lemmas.

\begin{appxlemma}
\label{appxlem:maximumovershootbjsigmak}
Let $\cB_j$ be as defined in~\eqref{eq:cBj} and let $k=\min \cB_j$. Let assume $k>1$. Set $b_j=b_j(\sigma_k)$. Then
\begin{align}
     \argmax_{t\in[b_{j-1},b_j]} o_{T_{k-1}}(t)
      \in \{ b_{j-1}, b_j, \eta, \nu \}
\end{align}
where $\eta \in [b_{j-1}, t_j]$ and $\nu \in [t_j , b_j]$ 
are determined by
\begin{align}
    W(\eta; R_{\rm o}) = W(\nu; R_{\rm o}) = T_{k-1} .
    \label{eq:EtaandNu}
\end{align}
\end{appxlemma}
\begin{proof}
According to Proposition~\ref{prop:overshoot}, $O_{T_{k-1}}(t;R_{\rm o})$ is related to $O_{T_{k-1}}(b_{j-1};R_{\rm o})$ over the interval $t\in[b_{j-1},b_j]$ as follows,
\begin{align}
 O_{T_{k-1}}(t;R_{\rm o})=\left\lbrace\begin{array}{l}
  O_{T_{k-1}}(b_{j-1};R_{\rm o}) +\beta(b_{j-1},t,T_{k-1})    \\t\in[b_{j-1},t_j]
  \\
    O_{T_{k-1}}(t_j;R_{\rm o})+\alpha(t_{j},t,T_{k-1}) \\
     t\in[t_{j},b_j] 
 \end{array}\right.  
 \label{eq:TTildecases}
\end{align} 
We can have one of the two following cases based on $\tilde{s}_j$
\begin{case}
$\tilde{s}_j=s_j$.
\label{case:case1}
\end{case}
In this case $s_j>b_{j-1}$. According to~\eqref{eq:tj_R1} and~\eqref{eq:Stochastic Regulator eq 2}, $t_j=\tilde{s}_j$ and $W(t_j;R_{\rm o})=W(\tilde{s}_j;R_1)$. In this case using~\eqref{eq:Stochastic Regulator eq 1}-\eqref{eq:Stochastic Regulator eq 5} we have 
\begin{equation}
 W(t;R_{\rm o})=W(t;R_1)   ~~~~\forall ~t\in[b_{j-1},b_j]
 \label{eq:WRoeqWR1}
\end{equation}
If $T_{k}-T_{k-1}>\delta$, for $W(t;R_{\rm o})$ on the interval $t\in[b_{j-1},b_j]$ we can have one the five subcases shown depicted Fig.~\ref{fig:Lemma1Case1}. On the other hand, If $T_{k}-T_{k-1}=\delta$, then $T_{k-1}=\sigma_k$ and $W(t;R_{\rm o})$ on the interval $t\in[b_{j-1},b_j]$ will be like the four subcases shown in Figs.~\ref{fig:Lemma1Case1subcase2}-\ref{fig:Lemma1Case1subcase5}.

\begin{figure*}
    \centering
    \subfigure[subcase 1]{
        \centering
        \begin{adjustbox}{width=0.3\textwidth,height=1.7in}
        \begin{tikzpicture}
           \draw [thick,<->] (7.5,0)--(0,0)--(0,4);
            \node [left] at (0,4) {\begin{small}$W(t;R_{\rm o})$\end{small}};
            \node [left] at (0,3.3) {\begin{small}$\sigma_{k}$\end{small}};
            \node [left] at (0,1.3) {\begin{small}$T_{k-1}$\end{small}};
            \node [left] at (0,0.5) {\begin{small}$\sigma_{k-1}$\end{small}};
            \node [left] at (0,1.7) {\begin{small}$W(\tilde{s}_j;R_1)$\end{small}};
            \node [left] at (0,2.5) {\begin{small}$W(\tilde{s}_j;R_1)+\delta_j$\end{small}};
            \node [right] at (7.5,0) {$t$};
            \node [left,below] at (5,0) {\begin{small}$t_j=\tilde{s}_j$\end{small}};
            \node [right,below] at (5.9,0) {\begin{small}$b_j$\end{small}};
            \draw [dashed] (0.5,3.37)--(1.5,3);
            \draw (1.5,3)--(5,1.7)--(5.8,2.5);
            \draw [help lines] (5.8,0) -- (5.8,2.5);
            \draw [help lines] (0,1.7) -- (5,1.7);
            \draw [help lines] (5,0) -- (5,1.7);
            \draw [help lines] (0,2.5) -- (5.8,2.5);
            \draw [help lines] (0,3.3) -- (6.7,3.3);
            \draw [help lines] (6.7,1.3) -- (0,1.3);
            \draw [help lines] (0,0.5)--(6.7,0.5);
            \draw [fill](5,0) circle [radius=0.05];
            \draw [fill](5.8,0) circle [radius=0.05];
            \end{tikzpicture}
        \end{adjustbox}
        \label{fig:Lemma1Case1subcase1}
    }
    \subfigure[subcase 2]{
        \centering
        \begin{adjustbox}{width=0.3\textwidth,height=1.7in}
        \begin{tikzpicture}
           \draw [thick,<->] (7.5,0)--(0,0)--(0,4);
            \node [left] at (0,4) {\begin{small}$W(t;R_{\rm o})$\end{small}};
            \node [left] at (0,3.3) {\begin{small}$\sigma_{k}$\end{small}};
            \node [left] at (0,1.35) {\begin{small}$T_{k-1}$\end{small}};
            \node [left] at (0,0.5) {\begin{small}$\sigma_{k-1}$\end{small}};
            \node [left] at (0,1) {\begin{small}$W(\tilde{s}_j;R_1)$\end{small}};
            \node [left] at (0,1.8) {\begin{small}$W(\tilde{s}_j;R_1)+\delta_j$\end{small}};
            \node [right] at (7.5,0) {$t$};
            \node [left,below] at (5,0) {\begin{small}$t_j$\end{small}};
            \node [right,below] at (5.9,0) {\begin{small}$b_j$\end{small}};
            \node [right,below] at (4.08,0) {\begin{small}$\eta$\end{small}};
            \node [below] at (5.35,0) {\begin{small}$\nu$\end{small}};
            \draw [dashed] (0.5,2.45)--(1,2.3);
            \draw (1,2.3)--(5,1)--(5.8,1.8);
            \draw [help lines] (5.35,0) -- (5.35,1.35);
            \draw [help lines] (5.8,0) -- (5.8,1.8);
            \draw [help lines] (4.08,0) -- (4.08,1.3);
            \draw [help lines] (0,1.8) -- (5.8,1.8);
            \draw [help lines] (5,0) -- (5,1);
            \draw [help lines] (0,1) -- (5,1);
            \draw [help lines] (0,3.3) -- (6.7,3.3);
            \draw [help lines] (6.7,1.3) -- (0,1.3);
            \draw [help lines] (0,0.5)--(6.7,0.5);
            \draw [fill](5.35,0) circle [radius=0.05];
            \draw [fill](4.08,0) circle [radius=0.05];
            \draw [fill](5,0) circle [radius=0.05];
            \draw [fill](5.8,0) circle [radius=0.05];
            \end{tikzpicture}
        \end{adjustbox}
        \label{fig:Lemma1Case1subcase2}
    }
    \subfigure[subcase 3]{
        \centering
        \begin{adjustbox}{width=0.3\textwidth,height=1.7in}
        \begin{tikzpicture}
           \draw [thick,<->] (7.5,0)--(0,0)--(0,4);
            \node [left] at (0,4) {\begin{small}$W(t;R_{\rm o})$\end{small}};
            \node [left] at (0,3.3) {\begin{small}$\sigma_{k}$\end{small}};
            \node [left] at (0,1.2) {\begin{small}$T_{k-1}$\end{small}};
            \node [left] at (0,0.4) {\begin{small}$\sigma_{k-1}$\end{small}};
            \node [left] at (0,0.7) {\begin{small}$W(\tilde{s}_j;R_1)$\end{small}};
            \node [left] at (0,1.6) {\begin{small}$W(\tilde{s}_j;R_1)+\delta_j$\end{small}};
            \node [right] at (7.5,0) {$t$};
            \node [left,below] at (5,0) {\begin{small}$t_j$\end{small}};
            \node [below] at (6,0) {\begin{small}$b_j$\end{small}};
            \node [below] at (3.77,0) {\begin{small}$b_{j-1}$\end{small}};
            \node [below] at (5.7,0) {\begin{small}$\nu$\end{small}};
            \draw [dashed] (3.47,0.7)--(3.57,0.8);
            \draw (3.57,0.8)--(3.77,1)--(5,0.6)--(5.9,1.5);
            \draw [help lines] (5.7,0) -- (5.7,1.3);
            \draw [help lines] (5.9,0) -- (5.9,1.5);
            \draw [help lines] (3.77,0) -- (3.77,1);
            \draw [help lines] (0,1.5) -- (5.9,1.5);
            \draw [help lines] (5,0) -- (5,0.6);
            \draw [help lines] (0,0.6) -- (5,0.6);
            \draw [help lines] (0,3.3) -- (6.7,3.3);
            \draw [help lines] (6.7,1.3) -- (0,1.3);
            \draw [help lines] (0,0.5)--(6.7,0.5);
            \draw [fill](5.7,0) circle [radius=0.05];
            \draw [fill](3.77,0) circle [radius=0.05];
            \draw [fill](5,0) circle [radius=0.05];
            \draw [fill](5.9,0) circle [radius=0.05];
            \end{tikzpicture}
        \end{adjustbox}
        \label{fig:Lemma1Case1subcase3}
    }
        \newline
        
        \subfigure[subcase 4]{
        \centering
        \begin{adjustbox}{width=0.3\textwidth,height=1.7in}
        \begin{tikzpicture}
           \draw [thick,<->] (7.5,0)--(0,0)--(0,4);
            \node [left] at (0,4) {\begin{small}$W(t;R_{\rm o})$\end{small}};
            \node [left] at (0,3.3) {\begin{small}$\sigma_{k}$\end{small}};
            \node [left] at (0,1.4) {\begin{small}$T_{k-1}$\end{small}};
            \node [left] at (0,0.3) {\begin{small}$\sigma_{k-1}$\end{small}};
            \node [left] at (0,0.6) {\begin{small}$W(\tilde{s}_j;R_1)$\end{small}};
            \node [left] at (0,1.1) {\begin{small}$W(\tilde{s}_j;R_1)+\delta_j$\end{small}};
            \node [right] at (7.5,0) {$t$};
            \node [left,below] at (5,0) {\begin{small}$t_j=\tilde{s}_j$\end{small}};
            \node [right,below] at (5.8,0) {\begin{small}$b_j$\end{small}};
            \node [right,below] at (2.85,0) {\begin{small}$\eta$\end{small}};
            \draw [dashed] (0.5,2.05)--(1,1.9);
            \draw (1,1.9)--(5,0.6)--(5.5,1.1);
            \draw [help lines] (5.5,0) -- (5.5,1.1);
            \draw [help lines] (2.85,0) -- (2.85,1.3);
            \draw [help lines] (0,1.1) -- (5.5,1.1);
            \draw [help lines] (5,0) -- (5,0.6);
            \draw [help lines] (0,0.6) -- (5,0.6);
            \draw [help lines] (0,3.3) -- (6.7,3.3);
            \draw [help lines] (6.7,1.3) -- (0,1.3);
            \draw [help lines] (0,0.5)--(6.7,0.5);
            \draw [fill](2.85,0) circle [radius=0.05];
            \draw [fill](5,0) circle [radius=0.05];
            \draw [fill](5.5,0) circle [radius=0.05];
            \end{tikzpicture}
        \end{adjustbox}
        \label{fig:Lemma1Case1subcase4}
    }
    \subfigure[subcase 5]{
        \centering
        \begin{adjustbox}{width=0.3\textwidth,height=1.7in}
        \begin{tikzpicture}
           \draw [thick,<->] (7.5,0)--(0,0)--(0,4);
            \node [left] at (0,4) {\begin{small}$W(t;R_{\rm o})$\end{small}};
            \node [left] at (0,3.3) {\begin{small}$\sigma_{k}$\end{small}};
            \node [left] at (0,1.4) {\begin{small}$T_{k-1}$\end{small}};
            \node [left] at (0,0.3) {\begin{small}$\sigma_{k-1}$\end{small}};
            \node [left] at (0,0.6) {\begin{small}$W(\tilde{s}_j;R_1)$\end{small}};
            \node [left] at (0,1.1) {\begin{small}$W(\tilde{s}_j;R_1)+\delta_j$\end{small}};
            \node [right] at (7.5,0) {$t$};
            \node [below] at (5,0) {\begin{small}$t_j$\end{small}};
            \node [below] at (5.5,0) {\begin{small}$b_j$\end{small}};
            \node [below] at (3.77,0) {\begin{small}$b_{j-1}$\end{small}};
            \draw [dashed] (3.47,0.7)--(3.57,0.8);
            \draw (3.57,0.8)--(3.77,1)--(5,0.6)--(5.5,1.1);
            \draw [help lines] (5.5,0) -- (5.5,1.1);
            \draw [help lines] (3.77,0) -- (3.77,1);
            \draw [help lines] (0,1.1) -- (5.5,1.1);
            \draw [help lines] (5,0) -- (5,0.6);
            \draw [help lines] (0,0.6) -- (5,0.6);
            \draw [help lines] (0,3.3) -- (6.7,3.3);
            \draw [help lines] (6.7,1.3) -- (0,1.3);
            \draw [help lines] (0,0.5)--(6.7,0.5);
            \draw [fill](3.77,0) circle [radius=0.05];
            \draw [fill](5,0) circle [radius=0.05];
            \draw [fill](5.5,0) circle [radius=0.05];
            \end{tikzpicture}
        \end{adjustbox}
        \label{fig:Lemma1Case1subcase5}
    }
    \caption{Different cases of $W(t;R_{\rm o})$ on the interval $[b_{j-1},b_j]$ with $\sigma=\sigma_k$, $t_j=\tilde{s}_j$ and $T_k-T_{k-1}>\delta$.}
    \label{fig:Lemma1Case1}
\end{figure*}
According to~\eqref{eq:WRoeqWR1}, in subcase~\ref{fig:Lemma1Case1subcase1}, $W(t;R_{\rm o})>T_{k-1}$ for $\forall t\in[b_{j-1},b_j]$. Hence, using~\eqref{eq:TTildecases},~\eqref{eq:alpha_increment}, and \eqref{eq:beta_increment} we have 
\begin{equation}
o_{T_{k-1}}(t)=\frac{O_{T_{k-1}}(b_{j-1};R_{\rm o})+(t-b_{j-1})}{b_{j-1}+(t-b_{j-1})},
    \label{eq:RatioCase11}    
\end{equation}
for $\forall ~t\in[b_{j-1},b_j]$. As $O_{\gamma}(t;R_{\rm o})< t$ for $\forall \gamma\in[0,T]$, it can be easily verified that in this case
\begin{equation}
    \argmax_{t\in[b_{j-1},b_j]} o_{T_{k-1}}(t)=b_j.
    \label{eq:maxRatioCase11}
\end{equation}
On the other hand, for subcase~\ref{fig:Lemma1Case1subcase2}, as $W(t;R_{\rm o})>T_{k-1}$ for $\forall t\in\{[b_{j-1},\eta]\cup[\nu,b_{j}]\}$, where $\eta$ and $\nu$ are defined in~\eqref{eq:EtaandNu}, we have 
\begin{align}
o_{T_{k-1}}(t)=\left\lbrace\begin{array}{ll}
  \frac{O_{T_{k-1}}(b_{j-1};R_{\rm o})+(t-b_{j-1})}{b_{j-1}+(t-b_{j-1})}   & t\in[b_{j-1},\eta] \\
    \frac{O_{T_{k-1}}(\eta;R_{\rm o})}{\eta+(t-\eta)} & t\in[\eta,\nu]\\
    \frac{O_{T_{k-1}}(\eta;R_{\rm o})+(t-\nu)}{\nu+(t-\nu)} & t\in [\nu,b_j]
\end{array}\right.
\label{eq:RatioCase12}
\end{align}
Therefore, it can be easily verified that in this case 
\begin{equation}
    \argmax_{t\in[b_{j-1},b_j]} o_{T_{k-1}}(t)\in\{\eta,b_j\}.
    \label{eq:maxRatioCase12}
\end{equation}
For subcase~\ref{fig:Lemma1Case1subcase3}, as $W(t;R_{\rm o})>T_{k-1}$ for $\forall t\in[\nu,b_{j}]$, we have
\begin{align}
o_{T_{k-1}}(t)=
\left\lbrace\begin{array}{ll}
  \frac{O_{T_{k-1}}(b_{j-1};R_{\rm o})}{b_{j-1}+(t-b_{j-1})}   & t\in[b_{j-1},\nu] \\
    \frac{O_{T_{k-1}}(b_{j-1};R_{\rm o})+(t-\nu)}{\nu+(t-\nu)} & t\in[\nu,b_j]
\end{array}\right.
\label{eq:RatioCase13}
\end{align}
Therefore, it can be easily verified that in this case 
\begin{equation}
    \argmax_{t\in[b_{j-1},b_j]} o_{T_{k-1}}(t)\in\{b_{j-1},b_j\}.
    \label{eq:maxRatioCase13}
\end{equation}
For subcase~\ref{fig:Lemma1Case1subcase4}, as $W(t;R_{\rm o})>T_{k-1}$ for $\forall t\in[b_{j-1},\eta]$, we have
\begin{align}
o_{T_{k-1}}(t)=
\left\lbrace\begin{array}{ll}
  \frac{O_{T_{k-1}}(b_{j-1};R_{\rm o})+(t-b_{j-1})}{b_{j-1}+(t-b_{j-1})}   & t\in[b_{j-1},\eta] \\
    \frac{O_{T_{k-1}}(\eta;R_{\rm o})}{\eta+(t-\eta)} & t\in[\eta,b_j]
\end{array}\right.
\label{eq:RatioCase14}
\end{align}
Therefore, it can be easily verified that in this case 
\begin{equation}
    \argmax_{t\in[b_{j-1},b_j]} o_{T_{k-1}}(t)=\eta.
    \label{eq:maxRatioCase14}
\end{equation}
For subcase~\ref{fig:Lemma1Case1subcase5}, as $W(t;R_{\rm o})<T_{k-1}$ for $\forall t\in[b_{j-1},b_{j}]$, we have
\begin{align}
o_{T_{k-1}}(t)=
  \frac{O_{T_{k-1}}(b_{j-1};R_{\rm o})}{b_{j-1}+(t-b_{j-1})},
\label{eq:RatioCase15}
\end{align}
for $\forall t\in[b_{j-1},b_j]$. Therefore, it can be easily verified that in this case 
\begin{equation}
    \argmax_{t\in[b_{j-1},b_j]} \frac{O_{T_{k-1}}(t;R_{\rm o})}{t}=b_{j-1}.
    \label{eq:maxRatioCase15}
\end{equation}
On the other hand , when $T_k-T_{k-1}=\delta$, we have the subcases similar to the subcases~\ref{fig:Lemma1Case1subcase2}-\ref{fig:Lemma1Case1subcase5}. Therefore, we will have the same relations as~\eqref{eq:maxRatioCase12}-\eqref{eq:maxRatioCase15}.
\begin{case}
$\tilde{s}_j=b_{j-1}$
\label{case:case3}
\end{case}
In this case $s_j<b_{j-1}$. According to~\eqref{eq:tj_R1} and~\eqref{eq:Stochastic Regulator eq 2}, $t_j=\tilde{s}_j=b_{j-1}$ and $W(t_j;R_{\rm o})=W(\tilde{s}_j;R_1)$. In this case using~\eqref{eq:Stochastic Regulator eq 2}-\eqref{eq:Stochastic Regulator eq 5} we have 
\begin{equation}
 W(t;R_{\rm o})=W(t;R_1)   ~~~~\forall ~t\in[b_{j-1},b_j]
 \label{eq:WRoeqWR1Case3}
\end{equation}
If $T_k-T_{k-1}>\delta$, for $W(t;R_{\rm o})$ on the interval $t\in[b_{j-1},b_j]$ we can have one the four subcases shown depicted Fig.~\ref{fig:Lemma1Case3}. On the other hand, If $T_k-T_{k-1}=\delta$, then $T_{k-1}=\sigma_k$ and we can have one the two subcases shown in Fig.~\ref{fig:Lemma1Case3subcase3} and~\ref{fig:Lemma1Case3subcase4}.
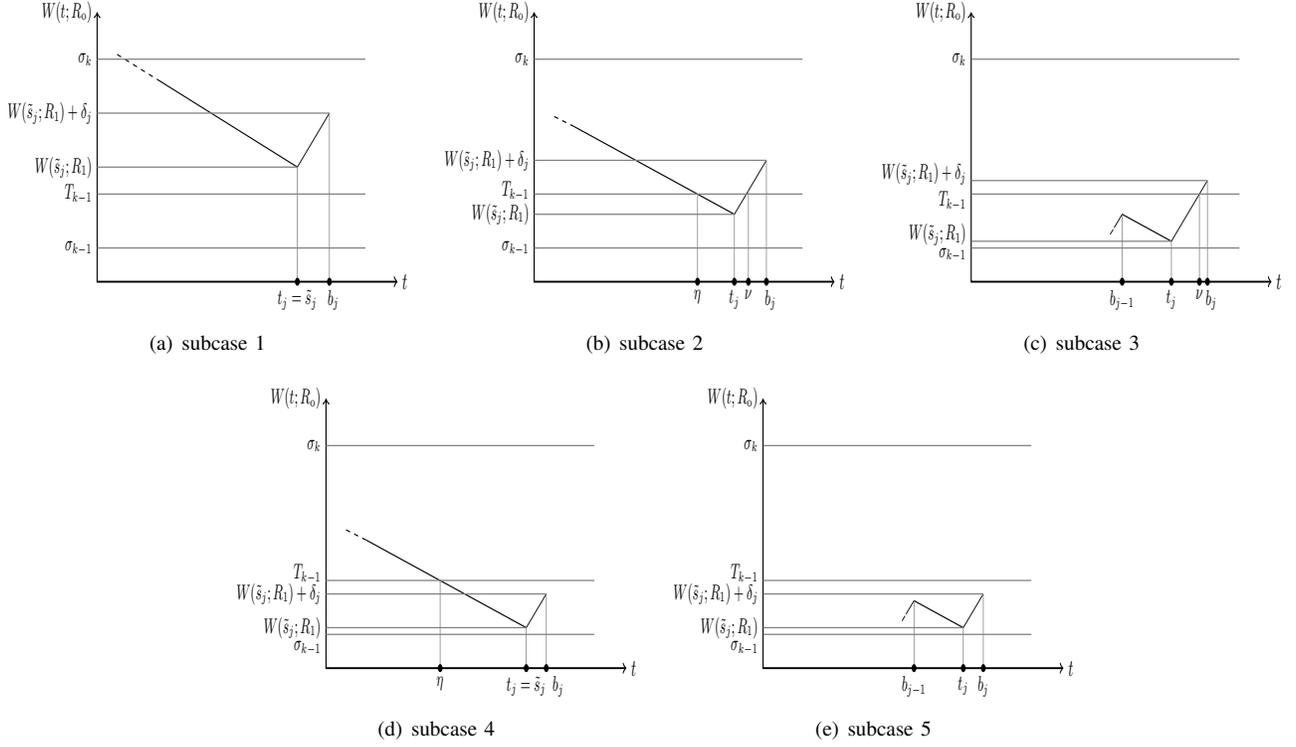
\begin{figure*}
    \centering
    \subfigure[subcase 1]{
        \centering
        \begin{adjustbox}{width=0.45\textwidth,height=1.7in}
        \begin{tikzpicture}
           \draw [thick,<->] (5,0)--(0,0)--(0,3.3);
           \node [left] at (0,3.3) {\begin{small}$W(t;R_{\rm o})$\end{small}};
            \node [left] at (0,2.8) {\begin{footnotesize}$W(\tilde{s}_j;R_1)+\delta_j$\end{footnotesize}};
            \node [left] at (0,2.1) {\begin{footnotesize}$W(\tilde{s}_j;R_{\rm 1})$\end{footnotesize}};
            \node [left] at (0,0.4) {\begin{footnotesize}$\sigma_{k-1}$\end{footnotesize}};
            \node [left] at (0,0.9) {\begin{footnotesize}$T_{k-1}$\end{footnotesize}};
            \node [left] at (0,2.5) {\begin{footnotesize}$\sigma_k$\end{footnotesize}};
            \node [right] at (5,0) {$t$};
            \node [left,below] at (2.8,0) {\begin{footnotesize}$t_j$\end{footnotesize}};
            \node [below] at (3.6,0) {\begin{footnotesize}$b_{j}$\end{footnotesize}};
            \draw [dashed] (2.1,1.55)--(2.4,1.8);
            \draw (2.4,1.8)--(3,2.3)--(3.6,2.8);
            \draw [help lines] (2.8,0) -- (2.8,2.1);
            \draw [help lines] (3.6,0) -- (3.6,2.8);
            \draw [help lines] (0,2.8) -- (3.6,2.8);
            \draw [help lines] (0,2.5) -- (4.8,2.5);
            \draw [help lines] (2.8,2.1) -- (0,2.1);
            \draw [help lines] (0,0.9)--(4.8,0.9);
            \draw [help lines] (0,0.4)--(4.8,0.4);
            \draw [fill](3.6,0) circle [radius=0.05];
            \draw [fill](2.8,0) circle [radius=0.05];
            \end{tikzpicture}
        \end{adjustbox}
        \label{fig:Lemma1Case3subcase1}
    }
    \subfigure[subcase 2]{
        \centering
        \begin{adjustbox}{width=0.45\textwidth,height=1.7in}
        \begin{tikzpicture}
           \draw [thick,<->] (5,0)--(0,0)--(0,3.3);
           \node [left] at (0,3.3) {\begin{small}$W(t;R_{\rm o})$\end{small}};
            \node [left] at (0,2.3) {\begin{footnotesize}$W(\tilde{s}_j;R_1)+\delta_j$\end{footnotesize}};
            \node [left] at (0,1.8) {\begin{footnotesize}$W(\tilde{s}_j;R_{\rm 1})$\end{footnotesize}};
            \node [left] at (0,0.4) {\begin{footnotesize}$\sigma_{k-1}$\end{footnotesize}};
            \node [left] at (0,0.9) {\begin{footnotesize}$T_{k-1}$\end{footnotesize}};
            \node [left] at (0,2.6) {\begin{footnotesize}$\sigma_k$\end{footnotesize}};
            \node [right] at (5,0) {$t$};
            \node [left,below] at (3,0) {\begin{footnotesize}$t_j$\end{footnotesize}};
            \node [below] at (3.6,0) {\begin{footnotesize}$b_{j}$\end{footnotesize}};
            \draw [dashed] (2.1,1.05)--(2.4,1.3);
            \draw (2.4,1.3)--(3,1.8)--(3.6,2.3);
            \draw [help lines] (3,0) -- (3,1.8);
            \draw [help lines] (3.6,0) -- (3.6,2.3);
            \draw [help lines] (0,2.3) -- (3.6,2.3);
            \draw [help lines] (0,2.6) -- (4.8,2.6);
            \draw [help lines] (3,1.8) -- (0,1.8);
            \draw [help lines] (0,0.9)--(4.8,0.9);
            \draw [help lines] (0,0.4)--(4.8,0.4);
            \draw [fill](3.6,0) circle [radius=0.05];
            \draw [fill](3,0) circle [radius=0.05];
            \end{tikzpicture}
        \end{adjustbox}
        \label{fig:Lemma1Case3subcase2}
    }
    \newline
    \subfigure[subcase 3]{
        \centering
        \begin{adjustbox}{width=0.45\textwidth,height=1.7in}
        \begin{tikzpicture}
           \draw [thick,<->] (5,0)--(0,0)--(0,3.3);
           \node [left] at (0,3.3) {\begin{small}$W(t;R_{\rm o})$\end{small}};
            \node [left] at (0,1.8) {\begin{footnotesize}$W(\tilde{s}_j;R_1)+\delta_j$\end{footnotesize}};
            \node [left] at (0,1.1) {\begin{footnotesize}$W(\tilde{s}_j;R_{\rm 1})$\end{footnotesize}};
            \node [left] at (0,0.4) {\begin{footnotesize}$\sigma_{k-1}$\end{footnotesize}};
            \node [left] at (0,1.4) {\begin{footnotesize}$T_{k-1}$\end{footnotesize}};
            \node [left] at (0,2.6) {\begin{footnotesize}$\sigma_k$\end{footnotesize}};
            \node [right] at (5,0) {$t$};
            \node [left,below] at (2.3,0) {\begin{footnotesize}$t_j$\end{footnotesize}};
            \node [below] at (3.1,0) {\begin{footnotesize}$b_j$\end{footnotesize}};
            \node [below] at (2.62,0) {\begin{footnotesize}$\nu$\end{footnotesize}};
            \draw [dashed] (1.6,0.55)--(1.9,0.8);
            \draw (1.9,0.8)--(2.5,1.3)--(3.1,1.8);
            \draw [help lines] (2.3,0) -- (2.3,1.1);
            \draw [help lines] (3.1,0) -- (3.1,1.8);
            \draw [help lines] (0,1.8) -- (3.1,1.8);
            \draw [help lines] (0,2.6) -- (4.8,2.6);
            \draw [help lines] (2.3,1.1) -- (0,1.1);
            \draw [help lines] (0,1.4)--(4.8,1.4);
            \draw [help lines] (0,0.4)--(4.8,0.4);
            \draw [help lines] (2.62,0)--(2.62,1.4);
            \draw [fill](3.1,0) circle [radius=0.05];
            \draw [fill](2.3,0) circle [radius=0.05];
            \draw [fill](2.62,0) circle [radius=0.05];
        \end{tikzpicture}
        \end{adjustbox}
        \label{fig:Lemma1Case3subcase3}
    }
        \subfigure[subcase 4]{
        \centering
        \begin{adjustbox}{width=0.45\textwidth,height=1.7in}
        \begin{tikzpicture}
           \draw [thick,<->] (5,0)--(0,0)--(0,3.3);
           \node [left] at (0,3.3) {\begin{small}$W(t;R_{\rm o})$\end{small}};
            \node [left] at (0,1.4) {\begin{footnotesize}$W(\tilde{s}_j;R_1)+\delta_j$\end{footnotesize}};
            \node [left] at (0,0.9) {\begin{footnotesize}$W(\tilde{s}_j;R_{\rm 1})$\end{footnotesize}};
            \node [left] at (0,0.4) {\begin{footnotesize}$\sigma_{k-1}$\end{footnotesize}};
            \node [left] at (0,1.8) {\begin{footnotesize}$T_{k-1}$\end{footnotesize}};
            \node [left] at (0,2.6) {\begin{footnotesize}$\sigma_k$\end{footnotesize}};
            \node [right] at (5,0) {$t$};
            \node [left,below] at (2.1,0) {\begin{footnotesize}$t_j$\end{footnotesize}};
            \node [below] at (2.7,0) {\begin{footnotesize}$b_j$\end{footnotesize}};
            \draw [dashed] (1.1,0.15)--(1.4,0.4);
            \draw (1.4,0.4)--(2.1,0.9)--(2.7,1.4);
            \draw [help lines] (2.1,0) -- (2.1,0.9);
            \draw [help lines] (2.7,0) -- (2.7,1.4);
            \draw [help lines] (0,1.4) -- (2.7,1.4);
            \draw [help lines] (0,2.6) -- (4.8,2.6);
            \draw [help lines] (2.1,0.9) -- (0,0.9);
            \draw [help lines] (0,1.8)--(4.8,1.8);
            \draw [help lines] (0,0.4)--(4.8,0.4);
            \draw [fill](2.7,0) circle [radius=0.05];
            \draw [fill](2.1,0) circle [radius=0.05];
            \end{tikzpicture}
        \end{adjustbox}
        \label{fig:Lemma1Case3subcase4}
    }
    \caption{Different cases of $W(t;R_{\rm o})$ on the interval $[b_{j-1},b_j]$ with $\sigma=\sigma_k$, $t_j=\tilde{s}_j=b_{j-1}$ and $T_k-T_{k-1}>\delta$.}
    \label{fig:Lemma1Case3}
\end{figure*}
As in subcase~\ref{fig:Lemma1Case3subcase1} and~\ref{fig:Lemma1Case3subcase2}, $W(t;R_{\rm o})>T_{k-1}$ for $\forall t\in[b_{j-1},b_j]$, we have 
\begin{equation}
o_{T_{k-1}}(t)=\frac{O_{T_{k-1}}(b_{j-1};R_{\rm o})+(t-b_{j-1})}{b_{j-1}+(t-b_{j-1})} 
 \label{eq:RatioCase31}
\end{equation}
for $\forall t\in[b_{j-1},b_j]$. Therefore, it can be easily verified that in this case 
\begin{equation}
    \argmax_{t\in[b_{j-1},b_j]} o_{T_{k-1}}(t)=b_j.
    \label{eq:maxRatioCase31}
\end{equation}
For subcase~\ref{fig:Lemma1Case3subcase3}, as $W(t;R_{\rm o})>T_{k-1}$ for $\forall t\in[\nu,b_{j}]$, we have 
\begin{align}
o_{T_{k-1}}(t)=\left\lbrace\begin{array}{ll}
   \frac{O_{T_{k-1}}(b_{j-1};R_{\rm o})}{b_{j-1}+(t-b_{j-1})}   &   t\in[b_{j-1},\nu] \\
    \frac{O_{T_{k-1}}(b_{j-1};R_{\rm o})+(t-\nu)}{\nu+(t-\nu)} &  t\in [\nu,b_j]
\end{array} \right.
 \label{eq:RatioCase32}
\end{align}
Therefore, it can be easily verified that in this case 
\begin{equation}
    \argmax_{t\in[b_{j-1},b_j]} o_{T_{k-1}}(t)\in\{b_{j-1},b_j\}.
    \label{eq:maxRatioCase32}
\end{equation}
For subcase~\ref{fig:Lemma1Case3subcase4}, as $W(t;R_{\rm o})<T_{k-1}$ for $\forall t\in[b_{j-1},b_{j}]$, we have 
\begin{equation}
o_{T_{k-1}}(t)=\frac{O_{T_{k-1}}(b_{j-1};R_{\rm o})}{b_{j-1}+(t-b_{j-1})} 
 \label{eq:RatioCase33}
\end{equation}
Therefore, it can be easily verified that in this case 
\begin{equation}
    \argmax_{t\in[b_{j-1},b_j]} o_{T_{k-1}}(t)=b_{j-1}.
    \label{eq:maxRatioCase33}
\end{equation}
On the other hand , when $T_k-T_{k-1}=\delta$, we have the subcases similar to the subcases~\ref{fig:Lemma1Case3subcase3}-\ref{fig:Lemma1Case3subcase4}. Therefore, we will have the same relations as~\eqref{eq:maxRatioCase32}-\eqref{eq:maxRatioCase33}.
\end{proof}

\begin{appxlemma}
\label{appxlem:maximumovershootbjsigmal}
Let $\cB_j$ be as defined in~\eqref{eq:cBj} and $k=\min \cB_j$.  Assume $k>1$. Let $\ell \in \{ 2, \ldots, k-1 \}$ and set $b_j = b_j(\sigma_\ell)$. Then
\begin{align}
     \argmax_{t\in[b_{j-1},b_j]} o_{T_{\ell-1}}(t)=b_j .
\end{align}
\end{appxlemma}

\begin{proof}
As it was mentioned before, $O_{T_{\ell-1}}(t;R_{\rm o})$ can be determined using $O_{T_{\ell-1}}(b_{j-1};R_{\rm o})$ over the interval $t\in[b_{j-1},b_j]$ according to~\eqref{eq:TTildecases}. Similarly, we can have one of the two following cases based on $\tilde{s}_j$:
\setcounter{case}{0}
\begin{case}
$\tilde{s}_j=s_j$
\label{case:case2}
\end{case}
In this case $s_j>b_{j-1}$. With $\sigma=\sigma_\ell$, $t_j$ is derived using~\eqref{eq:tj_R1}. According to~\eqref{eq:Stochastic Regulator eq 2}, $W(t_j;R_{\rm o})=\sigma_\ell$. Hence, in this case, when $T_k-T_{k-1}>\delta$, $W(t;R_{\rm o})$ is as shown in Fig.~\ref{fig:Lemma1Case2}. On the other hand, if $T_k-T_{k-1}=\delta$, $W(t;R_{\rm o})$ will be same as in Fig.~\ref{fig:Lemma1Case2}, except $T_{\ell-1}=\sigma_\ell$. In this case, as $W(t;R_{\rm o})\geq T_{\ell-1}$ for $\forall t\in [b_{j-1},b_j]$, we have
\begin{align}
o_{T_{\ell-1}}(t)=
  \frac{O_{T_{\ell-1}}(b_{j-1};R_{\rm o})+(t-b_{j-1})}{b_{j-1}+(t-b_{j-1})},
\label{eq:RatioCase2}
\end{align}
for $\forall t\in[b_{j-1},b_j]$. Therefore, it can be easily verified that in this case 
\begin{equation}
    \argmax_{t\in[b_{j-1},b_j]} o_{T_{\ell-1}}(t)=b_{j}.
    \label{eq:maxRatioCase2}
\end{equation}
\begin{figure}[t]
    \centering
        \begin{adjustbox}{width=0.9\columnwidth,height=1.7in}
        \begin{tikzpicture}
           \draw [thick,<->] (6.0,0)--(0,0)--(0,4);
            \node [left] at (0,4) {\begin{small}$W(t;R_{\rm o})$\end{small}};
            \node [left] at (0,3.3) {\begin{footnotesize}$\sigma_{k}$\end{footnotesize}};
            \node [left] at (0,1.1) {\begin{footnotesize}$W(t_j;R_{\rm o})=\sigma_{\ell}$\end{footnotesize}};
            \node [left] at (0,0.5) {\begin{footnotesize}$T_{\ell-1}$\end{footnotesize}};
            \node [left] at (0,1.7) {\begin{footnotesize}$W(t_j;R_{\rm o})+\delta_j$\end{footnotesize}};
            \node [left] at (0,2.5) {\begin{footnotesize}$W(\tilde{s}_j;R_1)$\end{footnotesize}};
            \node [right] at (6.0,0) {$t$};
            \node [below] at (2.28,0) {\begin{footnotesize}$\tilde{s}_j$\end{footnotesize}};
            \node [below] at (2.78,0) {\begin{footnotesize}$\tilde{a}_j$\end{footnotesize}};
            \node [below] at (5,0) {\begin{footnotesize}$b_j$\end{footnotesize}};
            \node [right,below] at (4.5,0) {\begin{footnotesize}$t_j$\end{footnotesize}};
            \draw (5.8,3.8)--(3.4,3.8)--(3.4,3)--(5.8,3)--(5.8,3.8);
            \draw (5.7,3.6)--(5,3.6);
            \draw [dashdotted](5.7,3.2)--(5,3.2);
            \node [left] at (5,3.6) {\begin{footnotesize}$W(t;R_{\rm o})$\end{footnotesize}};
            \node [left] at (5,3.2) {\begin{footnotesize}$W(t;R_{\rm 1})$\end{footnotesize}};
            \draw [dashed] (0.5,3.63)--(1.5,3);
            \draw [dashdotted] (2.28,2.5)--(2.78,3.2)--(5,1.8);
            \draw (1.5,3)--(4.5,1.1)--(5,1.8);
            \draw [help lines] (2.78,0) -- (2.78,3.2);
            \draw [help lines] (4.5,0) -- (4.5,1.1);
            \draw [help lines] (0,1.8) -- (5,1.8);
            \draw [help lines] (5,0) -- (5,1.8);
            \draw [help lines] (0,2.5) -- (5.8,2.5);
            \draw [help lines] (2.28,0) -- (2.28,2.5);
            \draw [help lines] (0,3.3) -- (3.3,3.3);
            \draw [help lines] (0,1.1) -- (5.7,1.1);
            \draw [help lines] (0,0.5)--(5.7,0.5);
            \draw [fill](2.78,0) circle [radius=0.05];
            \draw [fill](2.28,0) circle [radius=0.05];
            \draw [fill](5,0) circle [radius=0.05];
            \draw [fill](4.5,0) circle [radius=0.05];
            \end{tikzpicture}
        \end{adjustbox}
        \caption{$W(t;R_{\rm o})$, when $\sigma=\sigma_\ell$, $\tilde{s}_j=s_j$ and $T_k-T_{k-1}>\delta$.}
        \label{fig:Lemma1Case2}
    \end{figure}
\begin{case}
$\tilde{s}_j=b_{j-1}$
\label{case:case4}
\end{case}
In this case $s_j<b_{j-1}$. With $\sigma=\sigma_\ell$, $t_j$ is derived using~\eqref{eq:tj_R1}. According to~\eqref{eq:Stochastic Regulator eq 2}, $W(t_j;R_{\rm o})=\sigma_\ell$. Hence, in this case, when $T_k-T_{k-1}>\delta$, $W(t;R_{\rm o})$ is as shown in Fig.~\ref{fig:Lemma1Case4}. On the other hand, if $T_k-T_{k-1}=\delta$, $W(t;R_{\rm o})$ will be same as in Fig.~\ref{fig:Lemma1Case4}, except $T_{\ell-1}=\sigma_\ell$. In this case, as $W(t;R_{\rm o})\geq T_{\ell-1}$ for $\forall t\in [b_{j-1},b_j]$, we have
\begin{align}
o_{T_{\ell-1}}(t)=
  \frac{O_{T_{\ell-1}}(b_{j-1};R_{\rm o})+(t-b_{j-1})}{b_{j-1}+(t-b_{j-1})},
\label{eq:RatioCase4}
\end{align}
for $\forall t\in[b_{j-1},b_j]$. Therefore, it can be easily verified that in this case 
\begin{equation}
    \argmax_{t\in[b_{j-1},b_j]} o_{T_{\ell-1}}(t)=b_{j}.
    \label{eq:maxRatioCase4}
\end{equation}
\begin{figure}[t]
    \centering
        \begin{adjustbox}{width=0.9\columnwidth,height=1.7in}
        \begin{tikzpicture}
           \draw [thick,<->] (6.0,0)--(0,0)--(0,4);
            \node [left] at (0,4) {\begin{small}$W(t;R_{\rm o})$\end{small}};
            \node [left] at (0,3.3) {\begin{footnotesize}$\sigma_{k}$\end{footnotesize}};
            \node [left] at (0,1.1) {\begin{footnotesize}$W(t_j;R_{\rm o})=\sigma_{\ell}$\end{footnotesize}};
            \node [left] at (0,0.5) {\begin{footnotesize}$T_{\ell-1}$\end{footnotesize}};
            \node [left] at (0,1.7) {\begin{footnotesize}$W(t_j;R_{\rm o})+\delta_j$\end{footnotesize}};
            \node [left] at (0,2.5) {\begin{footnotesize}$W(\tilde{s}_j;R_1)$\end{footnotesize}};
            \node [right] at (6.0,0) {$t$};
            \node [below] at (2.28,0) {\begin{footnotesize}$\tilde{s}_j$\end{footnotesize}};
            \node [below] at (2.78,0) {\begin{footnotesize}$\tilde{a}_j$\end{footnotesize}};
            \node [below] at (5,0) {\begin{footnotesize}$b_j$\end{footnotesize}};
            \node [right,below] at (4.5,0) {\begin{footnotesize}$t_j$\end{footnotesize}};
            \draw (5.8,3.8)--(3.4,3.8)--(3.4,3)--(5.8,3)--(5.8,3.8);
            \draw (5.7,3.6)--(5,3.6);
            \draw [dashdotted](5.7,3.2)--(5,3.2);
            \node [left] at (5,3.6) {\begin{footnotesize}$W(t;R_{\rm o})$\end{footnotesize}};
            \node [left] at (5,3.2) {\begin{footnotesize}$W(t;R_{\rm 1})$\end{footnotesize}};
            \draw [dashed] (1.28,1.1)--(1.78,1.8);
            \draw [dashdotted] (2.28,2.5)--(2.78,3.2)--(5,1.8);
            \draw (1.78,1.8)--(2.28,2.5)--(4.5,1.1)--(5,1.8);
            \draw [help lines] (2.78,0) -- (2.78,3.2);
            \draw [help lines] (4.5,0) -- (4.5,1.1);
            \draw [help lines] (0,1.8) -- (5,1.8);
            \draw [help lines] (5,0) -- (5,1.8);
            \draw [help lines] (0,2.5) -- (5.8,2.5);
            \draw [help lines] (2.28,0) -- (2.28,2.5);
            \draw [help lines] (0,3.3) -- (3.3,3.3);
            \draw [help lines] (0,1.1) -- (5.7,1.1);
            \draw [help lines] (0,0.5)--(5.7,0.5);
            \draw [fill](2.78,0) circle [radius=0.05];
            \draw [fill](2.28,0) circle [radius=0.05];
            \draw [fill](5,0) circle [radius=0.05];
            \draw [fill](4.5,0) circle [radius=0.05];
            \end{tikzpicture}
        \end{adjustbox}
        \caption{$W_\rho(t;R_{\rm o})$, when $\sigma=\sigma_\ell$, $\tilde{s}_j=b_{j-1}$ and $T_k-T_{k-1}>\delta$.}
        \label{fig:Lemma1Case4}
    \end{figure}
\end{proof}

\begin{appxlemma}
\label{appxlem:bjlowerbound}
Let $\cB_j$ be as defined in~\eqref{eq:cBj} and let $k=\min \cB_j$. Let assume $k>1$ and assume that $b_j$ satisfies the following lower bound
\begin{equation}
b_{j}\geq \frac{L}{\epsilon\rho}+\frac{T_M-\sigma_1}{\rho}+\frac{L}{C},
\label{eq:b_jnewlowerbound}
\end{equation}
where $\epsilon > 0$ is given in~\eqref{eq:epsilonappendix}. 
Let $\ell \in \{ 2, \ldots, k \}$ and set $b_j = b_j(\sigma_\ell)$.  Then
\begin{align}
    o_{T_{\ell-1}}(b_j ) \leq \bar{f}(T_{\ell}) ,
    \label{eq:upper bound at complete departure for whole time interval}
\end{align}
implies 
\begin{align}
    o_{T_{\ell-1}}(t) \leq \bar{f}(T_{\ell-1}), ~~~~ \forall t \in [b_{j-1}, b_j] .
    \label{eq:lhs fraction}
\end{align}
\end{appxlemma}
\begin{proof}
 According to the definition~\eqref{eq:f_bar},
\begin{equation}
\bar{f}(T_{\ell-1})\geq \bar{f}(T_\ell) \text{~~for~~} \ell\in\{2,\ldots,M\}.
\label{eq:f(zeta)inequality}
\end{equation}
Therefore, for $\ell\in\{2,\ldots,k\}$ and $b_j=b_j(\sigma_\ell)$, then if
\begin{equation}
 \argmax_{t\in[b_{j-1},b_j]}o_{T_{\ell-1}}(t)=b_j,
 \label{eq:maxRatioattbj}
 \end{equation}
then, based on~\eqref{eq:f(zeta)inequality}, having~\eqref{eq:upper bound at complete departure for whole time interval} yields~\eqref{eq:lhs fraction} and no lower bound on $b_j$ is required.
 On the other hand, we will show for the cases that~\eqref{eq:maxRatioattbj} does not hold or
 \begin{equation}
 \argmax_{t\in[b_{j-1},b_j]}o_{T_{\ell-1}}(t)\ne b_j, 
 \label{eq:maxRationotatbj}
 \end{equation}
 if $b_j$ is greater than the lower bound in~\eqref{eq:b_jnewlowerbound}, then  having~\eqref{eq:upper bound at complete departure for whole time interval} yields~\eqref{eq:lhs fraction}. 
 As it was shown previously in Lemmas~\ref{appxlem:maximumovershootbjsigmak} and~\ref{appxlem:maximumovershootbjsigmal}, the only cases of having~\eqref{eq:maxRationotatbj} is when $b_j=b_j(\sigma_k)$. 
 When $b_j=b_j(\sigma_k)$ based on $\tilde{s}_j$ we can have two cases: 
\setcounter{case}{0}
\begin{case}
$\tilde{s}_j=s_j$
\end{case}
 This case is shown in Fig.~\ref{fig:Lemma1Case1}. 
As it is explained in Lemma~\ref{appxlem:maximumovershootbjsigmak}, in the four subcases~\ref{fig:Lemma1Case1subcase2}-\ref{fig:Lemma1Case1subcase5} we can have cases of having the maximum of the overshoot ratio function over the interval $[b_{j-1},b_j]$ at some $t\ne b_j$.
The overshoot ratio functions for these case, $o_{T_{k-1}}(t)$, are depicted in Fig.~\ref{fig:Lemma4Case1} for $t\in[b_{j-1},b_j]$. Theses figures are derived using~\eqref{eq:RatioCase12},~\eqref{eq:RatioCase13},~\eqref{eq:RatioCase14}, and~\eqref{eq:RatioCase15}. Note that, as it was mentioned in Lemma~\ref{appxlem:maximumovershootbjsigmak}, in subcases Fig.~\ref{fig:Lemma1Case1subcase2} and~\ref{fig:Lemma1Case1subcase3} we can have the  maximum of the overshoot function happening at $b_j$, these cases are however not considered in Fig.~\ref{fig:Lemma4Case1subcase1} and~\ref{fig:Lemma4Case1subcase3}, as if the overshoot ratio function is maximized at $b_j$, then~\eqref{eq:lhs fraction} holds for $\forall t\in[b_{j-1},b_j]$ and no lower bound in needed on $b_j$.
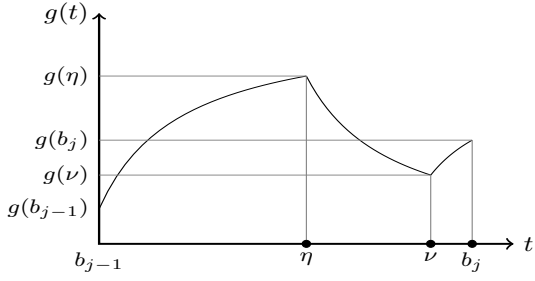
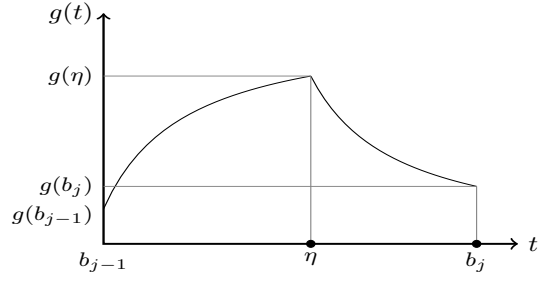
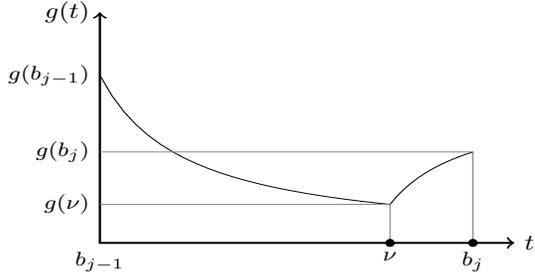
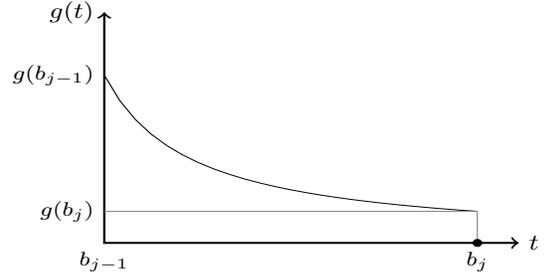
\begin{figure*}
    \centering
    \subfigure[Overshoot ratio function for subcase 2 in Fig.~\ref{fig:Lemma1Case1subcase2}]{
        \centering
        \begin{adjustbox}{width=0.4\textwidth,height=1.5in}
        \begin{tikzpicture}
            \draw [thick,<->] (5,0)--(0,0)--(0,3.3);
            \draw [domain=0:2.5] plot (\x, {(5+31.6*\x)/(10+10*\x)});
            \draw [domain=2.5:4] plot (\x, {30/(12.5+12*(\x-2.5))});
            \draw [domain=4:4.5] plot (\x, {(5+20*(\x-4))/(10+10*(\x-4))+0.483});
            \draw [help lines] (2.5,0) -- (2.5,2.4);
            \draw [help lines] (4,0) -- (4,0.9836);
            \draw [help lines] (4.5,0) -- (4.5,1.4830);
            \draw [fill](2.5,0) circle [radius=0.05];
            \draw [fill](4,0) circle [radius=0.05];
            \draw [fill](4.5,0) circle [radius=0.05];
            \draw [help lines] (0,2.4) -- (2.5,2.4);
            \draw [help lines] (0,0.9836) -- (4,0.9836);
            \draw [help lines] (0,1.4830) -- (4.5,1.4830);
            \node [below] at (0,0) {\begin{footnotesize}$b_{j-1}$\end{footnotesize}};
            \node [below] at (2.5,0) {\begin{footnotesize}$\eta$\end{footnotesize}};
            \node [below] at (4,0) {\begin{footnotesize}$\nu$\end{footnotesize}};
            \node [below] at (4.5,0) {\begin{footnotesize}$b_j$\end{footnotesize}};
            \node [left] at (0,0.5) {\begin{footnotesize}$g(b_{j-1})$\end{footnotesize}};
            \node [left] at (0,0.9836) {\begin{footnotesize}$g(\nu)$\end{footnotesize}};
            \node [left] at (0,1.4830) {\begin{footnotesize}$g(b_j)$\end{footnotesize}};
            \node [left] at (0,2.4) {\begin{footnotesize}$g(\eta)$\end{footnotesize}};
            \node [left] at (0,3.3) {\begin{small}$g(t)$\end{small}};
            \node [right] at (5,0) {\begin{small}$t$\end{small}};
            \end{tikzpicture}
        \end{adjustbox}
        \label{fig:Lemma4Case1subcase1}
    }
    \centering
    \subfigure[Overshoot ratio function for subcase 4 in Fig.~\ref{fig:Lemma1Case1subcase4}]{
        \centering
        \begin{adjustbox}{width=0.4\textwidth,height=1.5in}
        \begin{tikzpicture}
            \draw [thick,<->] (5,0)--(0,0)--(0,3.3);
            \draw [domain=0:2.5] plot (\x, {(5+31.6*\x)/(10+10*\x)});
            \draw [domain=2.5:4.5] plot (\x, {30/(12.5+12*(\x-2.5))});
            \draw [help lines] (2.5,0) -- (2.5,2.4);
            \draw [help lines] (4.5,0) -- (4.5,0.8219);
            \draw [fill](2.5,0) circle [radius=0.05];
            \draw [fill](4.5,0) circle [radius=0.05];
            \draw [help lines] (0,2.4) -- (2.5,2.4);
            \draw [help lines] (0,0.8219) -- (4.5,0.8219);
            \node [below] at (0,0) {\begin{footnotesize}$b_{j-1}$\end{footnotesize}};
            \node [below] at (2.5,0) {\begin{footnotesize}$\eta$\end{footnotesize}};
            \node [below] at (4.5,0) {\begin{footnotesize}$b_j$\end{footnotesize}};
            \node [left] at (0,0.4) {\begin{footnotesize}$g(b_{j-1})$\end{footnotesize}};
            \node [left] at (0,0.8219) {\begin{footnotesize}$g(b_j)$\end{footnotesize}};
            \node [left] at (0,2.4) {\begin{footnotesize}$g(\eta)$\end{footnotesize}};
            \node [left] at (0,3.3) {\begin{small}$g(t)$\end{small}};
            \node [right] at (5,0) {\begin{small}$t$\end{small}};
            \end{tikzpicture}
        \end{adjustbox}
        \label{fig:Lemma4Case1subcase2}
    }
    \newline
    \centering
    \subfigure[Overshoot ratio function for subcase 3 in Fig.~\ref{fig:Lemma1Case1subcase3}]{
        \centering
        \begin{adjustbox}{width=0.4\textwidth,height=1.5in}
        \begin{tikzpicture}
           \draw [thick,<->] (5,0)--(0,0)--(0,3.3);
            \draw [domain=0:3.5] plot (\x, {30/(12.5+12*\x)});
            \draw [domain=3.5:4.5] plot (\x, {(5+20*(\x-3.5))/(10+10*(\x-3.5))+0.0505});
            \draw [help lines] (3.5,0) -- (3.5,0.5505);
            \draw [help lines] (4.5,0) -- (4.5,1.3005);
            \draw [fill](3.5,0) circle [radius=0.05];
            \draw [fill](4.5,0) circle [radius=0.05];
            \draw [help lines] (0,0.5505) -- (3.5,0.5505);
            \draw [help lines] (0,1.3005) -- (4.5,1.3005);
            \node [below] at (0,0) {\begin{footnotesize}$b_{j-1}$\end{footnotesize}};
            \node [below] at (3.5,0) {\begin{footnotesize}$\nu$\end{footnotesize}};
            \node [below] at (4.5,0) {\begin{footnotesize}$b_j$\end{footnotesize}};
            \node [left] at (0,2.4) {\begin{footnotesize}$g(b_{j-1})$\end{footnotesize}};
            \node [left] at (0,0.5505) {\begin{footnotesize}$g(\nu)$\end{footnotesize}};
            \node [left] at (0,1.3005) {\begin{footnotesize}$g(b_j)$\end{footnotesize}};
            \node [left] at (0,3.3) {\begin{small}$g(t)$\end{small}};
            \node [right] at (5,0) {\begin{small}$t$\end{small}};
            \end{tikzpicture}
        \end{adjustbox}
        \label{fig:Lemma4Case1subcase3}
    }
    \centering
    \subfigure[Overshoot ratio function for subcase 5 in Fig.~\ref{fig:Lemma1Case1subcase5}]{
        \centering
        \begin{adjustbox}{width=0.4\textwidth,height=1.5in}
        \begin{tikzpicture}
           \draw [thick,<->] (5,0)--(0,0)--(0,3.3);
            \draw [domain=0:4.5] plot (\x, {30/(12.5+12*\x)});
            \draw [help lines] (4.5,0) -- (4.5,0.4511);
            \draw [fill](4.5,0) circle [radius=0.05];
            \draw [help lines] (0,0.4511) -- (4.5,0.4511);
            \node [below] at (0,0) {\begin{footnotesize}$b_{j-1}$\end{footnotesize}};
            \node [below] at (4.5,0) {\begin{footnotesize}$b_j$\end{footnotesize}};
            \node [left] at (0,2.4) {\begin{footnotesize}$g(b_{j-1})$\end{footnotesize}};
            \node [left] at (0,0.4511) {\begin{footnotesize}$g(b_j)$\end{footnotesize}};
            \node [left] at (0,3.3) {\begin{small}$g(t)$\end{small}};
            \node [right] at (5,0) {\begin{small}$t$\end{small}};
            \end{tikzpicture}
        \end{adjustbox}
        \label{fig:Lemma4Case1subcase4}
    }
    \caption{$o_{T_{k-1}}(t)$ for $\tilde{s}_j=s_j$ for the cases of $\argmax_{t\in[b_{j-1},b_j]}o_{T_{k-1}}\neq b_j$.}
    \label{fig:Lemma4Case1}
\end{figure*}

For the subcase Fig.~\ref{fig:Lemma4Case1subcase1} according to Fig.~\ref{fig:Lemma1Case1subcase2} we have
\begin{align*}
   T_{k-1}- W(\tilde{s}_j;R_{\rm o})\leq T_{k-1}-\sigma_{k-1}=\delta.
\end{align*}
Therefore,
\begin{align*}
    t_j-\eta\leq \frac{\delta}{\rho}.
\end{align*}
Similarly,
\begin{align*}
    \nu-t_j\leq \frac{\delta}{C-\rho}.
\end{align*}
Therefore,
\begin{align}
    \nu-\eta\leq \frac{\delta}{\rho}+\frac{\delta}{C-\rho}=\frac{L}{\rho}.
    \label{eq:nuetaabound}
\end{align}
 On the other hand, according to~\eqref{eq:upper bound at complete departure for whole time interval},~\eqref{eq:RatioCase12}, and Fig.~\ref{fig:Lemma4Case1subcase1}
\begin{align*}
    o_{T_{k-1}}(\nu)=\frac{O_{T_{k-1}}(\eta;R_{\rm o})}{\nu}\leq o_{T_{k-1}}(b_j)\leq \bar{f}(T_{k}).
\end{align*}
Therefore,
\begin{align}
\frac{O_{T_{k-1}}(\eta;R_{\rm o})}{\eta+(\nu-\eta)}\leq &\bar{f}(T_{k})
\nonumber\\
&\rightarrow o_{T_{k-1}}(\eta)
\leq \frac{\eta+(\nu-\eta)}{\eta}\bar{f}(T_{k})
\label{eq:Lemma4Case1Subcase1boundongeta}
\end{align}
In this subcase 
\begin{equation*}
    \argmax_{t\in[b_{j-1},b_j]} o_{T_{k-1}}(t)=\eta.
\end{equation*}
Therefore, we need to find a lower bound on $b_j$ such that $o_{T_{k-1}}(\eta)<\bar{f}(T_{k-1})$. 
 We know, 
\begin{align*}
    \frac{\bar{f}(T_{k})}{\eta}\leq \frac{\bar{f}(T_{k})}{b_{j-1}}.
\end{align*} 
Therefore, according to~\eqref{eq:Lemma4Case1Subcase1boundongeta},~\eqref{eq:nuetaabound}, and~\eqref{eq:epsilonappendix}, if 
\begin{align*}
   \frac{\bar{f}(T_{k})}{b_{j-1}}<\frac{\epsilon\rho}{L}
\end{align*}
then $o_{T_{k-1}}(\eta)<\bar{f}(T_{k-1})$. Therefore, if 
\begin{equation}
b_{j-1}> \frac{L}{\rho\epsilon}>\frac{L}{\rho\epsilon}\bar{f}(T_{k})   
\label{eq:Lemma4Case1Subcase1lowerboundonbj-1}
\end{equation}
then $o_{T_{k-1}}(\eta)<\bar{f}(T_{k-1})$. 
Since $T_M$ is chosen large enough such that $\cB_j\ne \emptyset$ for all $j$, 
we can assert that $W_\rho(b_{j-1};R_{\rm o})\leq T_M=\sigma_M+\delta$. On the
other hand, as $k =\min \cB_j > 1$, we have that $W_\rho(t_{j};R_{\rm o})\geq \sigma_1$. 
Using~\eqref{eq:Stochastic Regulator eq 1}, we have
\begin{align}
t_j-b_{j-1}= \frac{W(b_{j-1};R_{\rm o})-W(t_{j};R_{\rm o})}{\rho} \leq \frac{T_M-\sigma_1}{\rho} .
\end{align}
Therefore,
\begin{align}
b_{j}-b_{j-1} =\frac{L_j}{C}+t_j-b_{j-1}\leq \frac{T_M-\sigma_1}{\rho}+\frac{L}{C} . 
\label{eq:b_jb_j-1inequality}
\end{align}
Therefore, if 
\begin{equation}
b_{j}> \frac{L}{\rho\epsilon}+\frac{T_M-\sigma_1}{\rho}+\frac{L}{C} 
\label{eq:bjnewlowerboundcase1}
\end{equation}
then $o_{T_{k-1}}(\eta)<\bar{f}(T_{k-1})$.
Following the same arguments for the subcase Fig.~\ref{fig:Lemma4Case1subcase2}, we can have the same lower bound for $b_{j}$ as~\eqref{eq:bjnewlowerboundcase1}.

For the subcase Fig.~\ref{fig:Lemma4Case1subcase3}, following the same arguments we can show 
\begin{align*}
    t_j-b_{j-1}\leq \frac{\delta}{\rho},~~~~~\nu-t_j\leq \frac{\delta}{C-\rho}.
\end{align*}
Therefore,
\begin{align}
    \nu-b_{j-1}\leq \frac{\delta}{\rho}+\frac{\delta}{C-\rho}=\frac{L}{\rho}.
    \label{eq:nubj-1abound}
\end{align}
On the other hand, according to~\eqref{eq:upper bound at complete departure for whole time interval},~\eqref{eq:RatioCase13}, and Fig.~\ref{fig:Lemma4Case1subcase3}
\begin{align*}
    o_{T_{k-1}}(\nu)=\frac{O_{T_{k-1}}(b_{j-1};R_{\rm o})}{\nu}\leq o_{T_{k-1}}(b_j)\leq \bar{f}(T_{k}).
\end{align*}
Therefore,
\begin{align}
&\frac{O_{T_{k-1}}(b_{j-1};R_{\rm o})}{b_{j-1}+(\nu-b_{j-1})}\leq \bar{f}(T_{k})
\nonumber\\
&\rightarrow o_{T_{k-1}}(b_{j-1})
\leq \frac{b_{j-1}+(\nu-b_{j-1})}{b_{j-1}}\bar{f}(T_{k})
\label{eq:Lemma4Case1Subcase3boundongbj-1}
\end{align}
In this subcase 
\begin{equation*}
    \argmax_{t\in[b_{j-1},b_j]} o_{T_{k-1}}(t)=b_{j-1}.
\end{equation*}
Similarly, according to~\eqref{eq:Lemma4Case1Subcase3boundongbj-1},~\eqref{eq:nubj-1abound}, and~\eqref{eq:epsilonappendix}, if 
\begin{align*}
   \frac{\bar{f}(T_{k})}{b_{j-1}}<\frac{\epsilon\rho}{L}
\end{align*}
then $g(b_{j-1})<\bar{f}(T_{k-1})$. Therefore, if 
\begin{equation}
b_{j-1}> \frac{L}{\rho\epsilon}>\frac{L}{\rho\epsilon}\bar{f}(T_{k})   
\label{eq:Lemma4Case1Subcase4lowerboundonbj-1}
\end{equation}
then $g(b_{j-1})<\bar{f}(T_{k-1})$. Therefore the same lower bound on $b_j$ as~\eqref{eq:bjnewlowerboundcase1} will be achieved for this subcase. Following the same arguments for the subcase Fig.~\ref{fig:Lemma4Case1subcase4}, we can have the same lower bound on $b_{j}$. 
\begin{case}
$\tilde{s}_j=b_{j-1}$
\end{case}
This case is shown in Fig.~\ref{fig:Lemma1Case3}. As it is explained in Lemma~\ref{appxlem:maximumovershootbjsigmal}, in the two subcases~\ref{fig:Lemma1Case3subcase3} and~\ref{fig:Lemma1Case3subcase4} we can have cases of having the maximum of the overshoot ratio function, $o_{T_{k-1}}(t)$, over the interval $[b_{j-1},b_j]$ at $t= b_{j-1}$.
The overshoot ratio function for these two subcases is similar to Fig.~\ref{fig:Lemma4Case1subcase3} and~\ref{fig:Lemma4Case1subcase4} for $t\in[b_{j-1},b_j]$. Theses figures are derived using~\eqref{eq:RatioCase32} and~\eqref{eq:RatioCase33}. Note that, as it was mentioned in Lemma~\ref{appxlem:maximumovershootbjsigmal}, in subcase Fig.~\ref{fig:Lemma1Case3subcase3} we can have the maximum of the overshoot function happening at $b_j$, following the same argument as before, this case however is not considered in here.

For the subcase Fig.~\ref{fig:Lemma1Case3subcase3}, which its overshoot ratio function is depicted in Fig.~\ref{fig:Lemma4Case1subcase3}, by following the same argument as before, we have
\begin{align}
    \nu-b_{j-1}\leq \frac{\delta}{C-\rho}.
    \label{eq:nubj-1bound}
\end{align}
On the other hand, according to~\eqref{eq:upper bound at complete departure for whole time interval},~\eqref{eq:RatioCase33}, and Fig.~\ref{fig:Lemma4Case1subcase3}
\begin{align*}
    o_{T_{k-1}}(\nu)=\frac{O_{T_{k-1}}(b_{j-1};R_{\rm o})}{\nu}\leq o_{T_{k-1}}(b_j)\leq \bar{f}(T_{k}).
\end{align*}
Therefore,
\begin{align}
&\frac{O_{T_{k-1}}(b_{j-1};R_{\rm o})}{b_{j-1}+(\nu-b_{j-1})}\leq \bar{f}(T_{k})
\nonumber\\&\rightarrow o_{T_{k-1}}(b_{j-1})\leq \frac{b_{j-1}+(\nu-b_{j-1})}{b_{j-1}}\bar{f}(T_{k})
\label{eq:Lemma4Case2Subcase1boundongbj-1}
\end{align}
In this subcase 
\begin{equation*}
    \argmax_{t\in[b_{j-1},b_j]} o_{T_{k-1}}(t)=b_{j-1}.
\end{equation*}
Similarly, according to~\eqref{eq:Lemma4Case2Subcase1boundongbj-1},~\eqref{eq:nubj-1abound}, and~\eqref{eq:epsilonappendix}, if 
\begin{align*}
   \frac{\bar{f}(T_{k})}{b_{j-1}}<\frac{\epsilon(C-\rho)}{\delta}
\end{align*}
then $o_{T_{k-1}}(b_{j-1})<\bar{f}(T_{k-1})$. Therefore, if 
\begin{equation}
b_{j-1}> \frac{\delta}{\epsilon(C-\rho)}=\frac{L}{C\epsilon}>\frac{\delta}{\epsilon(C-\rho)}\bar{f}(T_{k})   
\label{eq:Lemma4Case2Subcase1lowerboundonbj-1}
\end{equation}
then $o_{T_{k-1}}(b_{j-1})<f(T_{k-1})$. As in this case $b_{j-1}=t_j$, the lower bound on $b_j$ in this case will be
\begin{equation}
    b_j>\frac{L}{C\epsilon}+\frac{L}{C}
\label{eq:bjnewlowerboundcase2}    
\end{equation}
Following the same arguments, we can have the same lower bound for $b_{j}$ as~\eqref{eq:bjnewlowerboundcase2} for the subcase Fig.~\ref{fig:Lemma1Case3subcase4}. Therefore, using~\eqref{eq:bjnewlowerboundcase1} and~\eqref{eq:bjnewlowerboundcase2} the lower bound for $b_j$ to ensure~\eqref{eq:lhs fraction} will be 
\begin{align}
b_j>& \max\{\frac{L}{\rho\epsilon}+\frac{T_M-\sigma_1}{\rho}+\frac{L}{C} , \frac{L}{C\epsilon}+\frac{L}{C}\} \nonumber\\
&=\frac{L}{\rho\epsilon}+\frac{T_M-\sigma_1}{\rho}+\frac{L}{C}.
\label{eq:Lemma4lowerboundonbj}
\end{align}
In our case study with $M=56$, the lower bound in~\eqref{eq:Lemma4lowerboundonbj} will be $b_j\geq 2.35\times 10^3$. Therefore, for $j\geq 220$, inequality~\eqref{eq:lhs fraction} holds. 
\end{proof}

\begin{proof}[Proof of Lemma~\ref{appxlem:overshoot_soln1}]
Let $\cB_j$ be as defined in Lemma~\ref{appxlem:overshoot_soln1} and let $k=\min\cB_j$. 
In lemma~\ref{appxlem:bjlowerbound} we showed if $k>1$, $\mathcal{I}_j\neq\emptyset$ and 
\begin{equation}
b_j> \frac{L}{\rho\epsilon}+\frac{T_M-\sigma_1}{\rho}+\frac{L}{C}.   
\end{equation}
Then 
\begin{equation}
\forall \ell\in \mathcal{I}_j:  ~~o_{T_{\ell-1}}(t)\leq \bar{f}(T_{\ell-1}),~~~~\forall~t\in[b_{j-1},b_j(\sigma_\ell)].    
\end{equation}
On the other hand, if $\mathcal{I}_j=\emptyset$ or $k=1$, which in turn means $\mathcal{I}_j=\emptyset$, then $\sigma^*(j)=\sigma_1$. But as $\bar{f}(T_0)=1$, therefore
\begin{equation}
o_{T_{0}}(t)\leq \bar{f}(T_{0}),~~~~\forall~t\in[b_{j-1},b_j(\sigma_1)].   \end{equation}
\end{proof}
In the next section we show if $t$ is sufficiently large, then the limited constraint in~\eqref{eq:canonical_constraintPreliminaryTheorem} can be extended to the desired constraint in~\eqref{eq:canonical_constraint}.

\subsection{Proof of Theorem~\ref{thm:overshoot_soln1}, Part II}
\label{appx_subsec:overshoot_soln1_II}

In this section we show in order to achieve the desired constraint in~\eqref{eq:canonical_constraint} rather than the preliminary one in~\eqref{eq:canonical_constraintPreliminaryTheorem}, we need to increase $b_j$ from the lower in bound in~\eqref{eq:b_jlowerbound} to sufficiently large values.
The proof of Theorem~\ref{thm:overshoot_soln1} is based on the next two lemmas.

\begin{appxlemma}
\label{appxlem:all_gamma}
The $(\sigma^*, \rho)$ regulator defined by
\eqref{eq:cBj}--\eqref{eq:basic_implementation} produces an output traffic
stream that satisfies
\begin{equation}
o_{T_i}(t)\leq \bar{f}({T_i})~~~~\text{for~~}\forall i\in\{1,\ldots,M\},
\label{eq:OvershootRatioBoundforallT}
\end{equation}
for sufficiently large~$t$.
\end{appxlemma}
\begin{proof}
In order to prove this lemma, we use Fig.~\ref{fig:WorkloadFluctuationNumericalExample} which shows the input workload, $W(t;R_{\rm i})$, and output workload, $W(t;R_{\rm o})$, for $t\in[2.5e4,3e4]$ for the numerical example in Section~\ref{sec:numerical} with $M=56$. The corresponding overshoot ratios, $o_{T_{i}}(t)$ for two values of $T_{16}$ and $T_{21}$ on the interval $[2.5e4,3e4]$ are shown in Fig.~\ref{fig:OvershootRatioViolationinTime}. As can be seen in Fig.~\ref{fig:OvershootRatioViolationinTime}, over the interval of $[2.5e4,3e4]$ there is violation of the constraints in~\eqref{eq:OvershootRatioBoundforallT} for some $t$, as 
\begin{equation*}
    o_{T_{17}}(t)>\bar{f}(T_{16}),~~~~\forall~t\in[t_3,t_6],
\end{equation*}
where $t_3=2.73e4$ and $t_6=2.87e4$ are shown in Fig.~\ref{fig:OvershootRatioViolationinTime}. By explaining what happens on the interval $[2.5e4,3e4]$ we can explain why this violation happens and how these violations are avoided when $t$ is sufficiently large. Note that, although this is just one specific example, it can act as a guideline and does not limit the scope of this proof. 
\begin{figure}
\centering
\includegraphics[width=0.95\columnwidth]{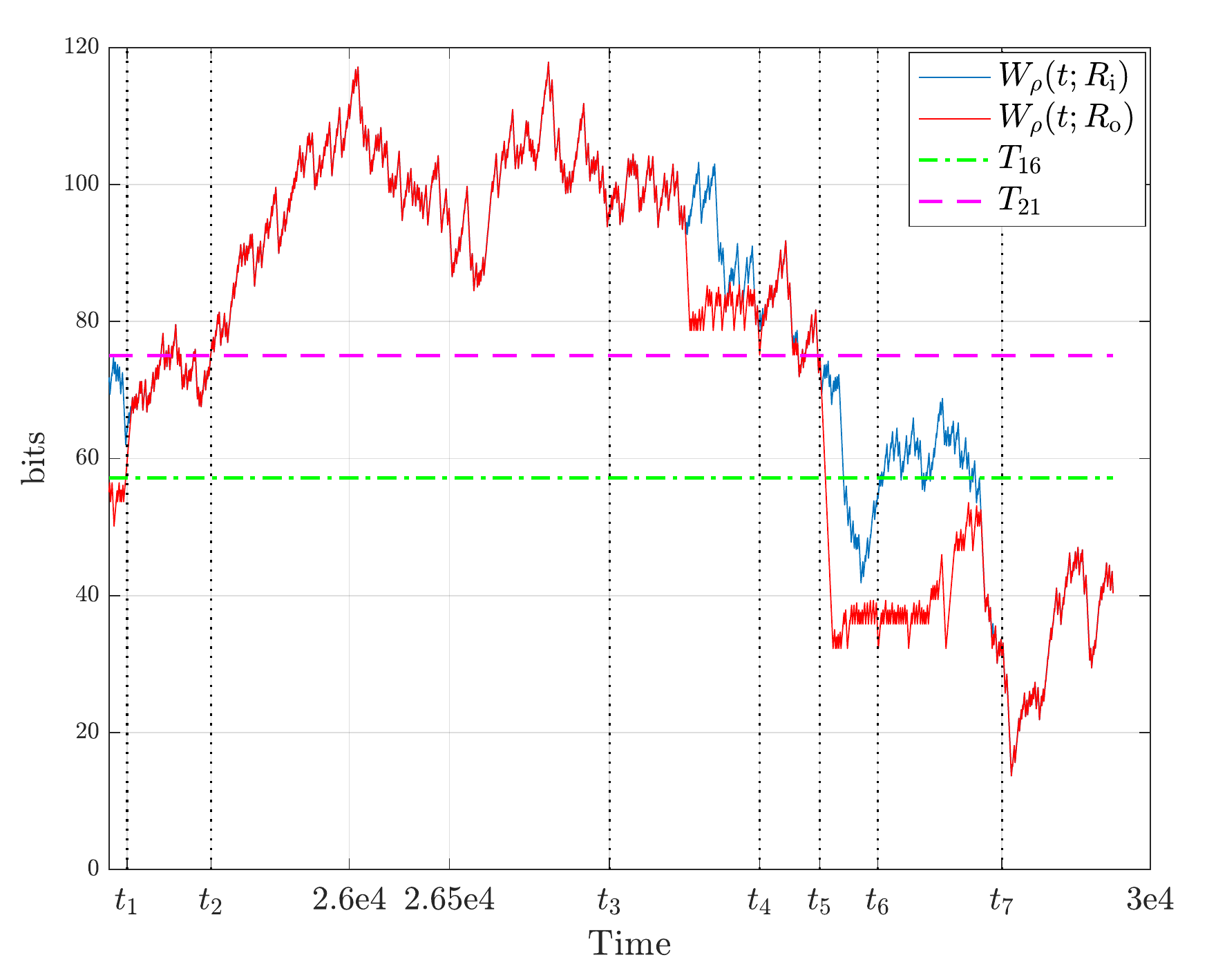}
\caption{Input and output workload for $t\in[2.5e4,3e4]$ for the example in Section~\ref{sec:numerical} for Algorithm~\ref{alg:Stochastic regulator} with $M=56$.}
\label{fig:WorkloadFluctuationNumericalExample}
\end{figure}
\begin{figure}
\centering
\includegraphics[width=0.95\columnwidth]{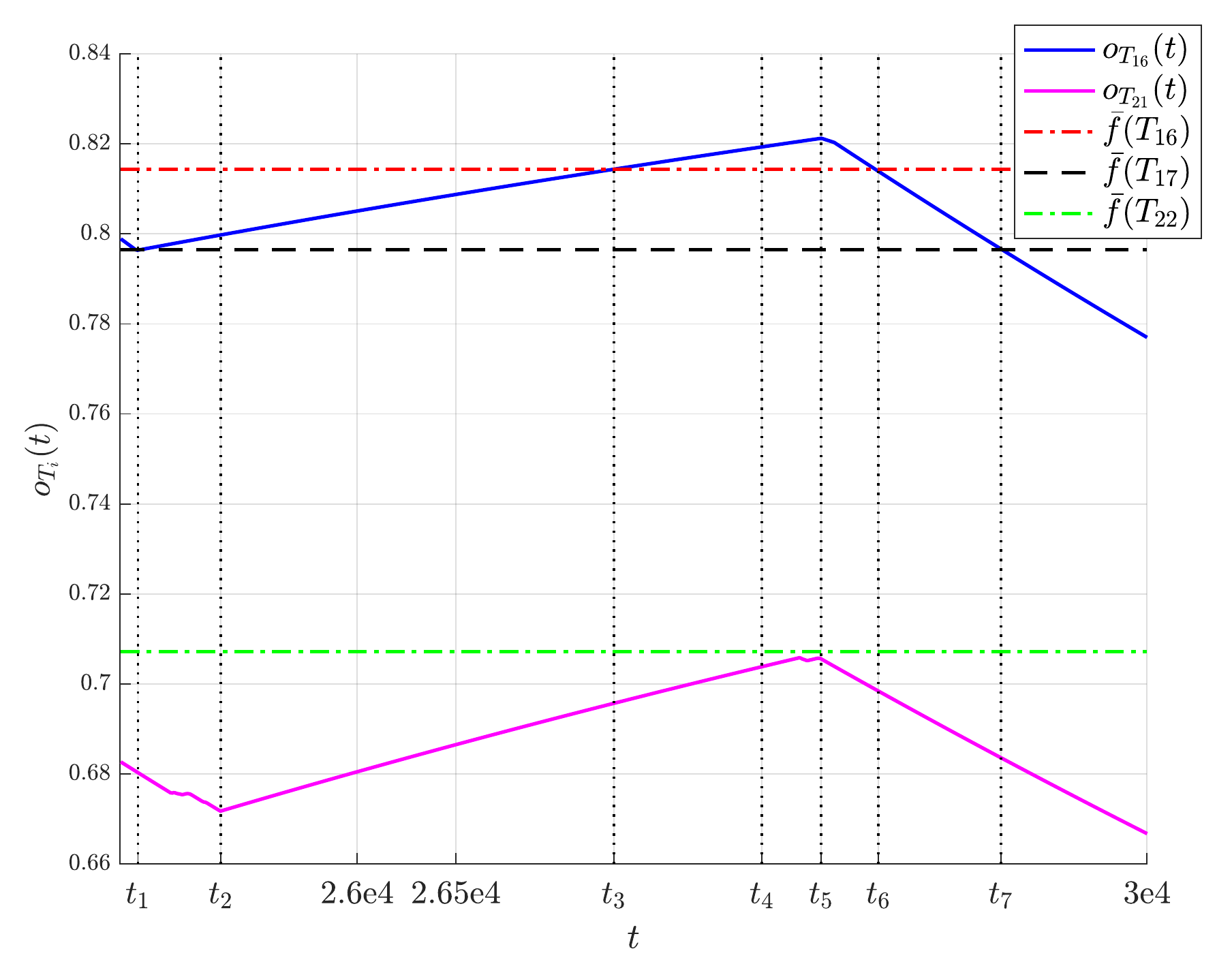}
\caption{Overshoot ratio, $o_{T_i}(t)$ for $t\in[2.5e4,3e4]$ and $T_i\in\{T_{16},T_{21}\}$ for the example in Section~\ref{sec:numerical} for Algorithm~\ref{alg:Stochastic regulator} with $M=56$.}
\label{fig:OvershootRatioViolationinTime}
\end{figure}

As it can be seen in Fig.~\ref{fig:WorkloadFluctuationNumericalExample}, at $t=t_1$ output workload increases above $T_{16}$, and after $t=t_5$ it decreases again to a level below $T_{16}$. According to Algorithm~\ref{alg:Stochastic regulator} and~\eqref{eq:Ij}, and as can be seen in Fig.~\ref{fig:OvershootRatioViolationinTime}, $o_{T_{16}}(t_1)\leq \bar{f}(T_{17})$. Therefore, at $t=t_1$, $\sigma^*(j)$ will be set to $\sigma^*(j)=\sigma_{17}$ and output workload will increase. As long as $o_{T_{16}}(b_j)\leq \bar{f}(T_{17})$ for $b_j>t_1$ and $W(\tilde{s}_j;R_1)<\sigma_{17}$, this process will continue and $\sigma^*(j)$ will be set to $\sigma^*(j)=\sigma_{17}$ till $W(\tilde{s}_j;R_1)$ is increased to $W(\tilde{s}_j;R_1)>\sigma_{17}$. At this point according to Algorithm~\ref{alg:Stochastic regulator} and~\eqref{eq:Ij}, $o_{T_{17}}(b_j)$  will be compared against $\bar{f}(T_{18})$ and if $o_{T_{17}}(t)\leq \bar{f}(T_{18})$, as in this example, $\sigma^*(j)$ will be set to $\sigma^*(j)=\sigma_{18}$. On the other hand, at $t=t_5$, $\min{\cB_j}=\sigma_{22}$. Therefore, $o_{T_{21}}(b_j)$  will be compared against $\bar{f}(T_{22})$ and if $o_{T_{21}}(b_j)> \bar{f}(T_{22})$, as in this example, $\sigma^*(j)$ will be set to a value less than $\sigma_{22}$. In this example $\sigma^*(j)$ is set to a $\sigma^*(j)=\sigma_{10}$ as $o_{T_{i-1}}(b_j(\sigma_i))> \bar{f}(T_{i})$ for $i=10,11,\ldots,22$. 

From the discussion above it can be understood when the output workload at the complete departure time, $W(b_j;R_{\rm o})$, increases above $\sigma_{i}$ for the $j$th packet, the overshoot ratio, $o_{T_{\ell}}(b_m(\sigma_{\ell+1}))$, will be compared against $\bar{f}(T_{\ell+1})$, for $\ell\in\{i,\ldots,M-1\}$ and $m>j$, as long as the output workload stays above $\sigma_{i}$. Therefore, during the interval that the workload is above $\sigma_{i}$, for all $m>j$ such that $W(b_m;R_{\rm o})>\sigma_i$, there is at least one $k\in\{i,\ldots,M-1\}$, such that $o_{T_{k}}(b_m(\sigma_{k+1}))\leq \bar{f}(T_{k+1})$. Based on this concepts we define threshold violation distance with respect to a threshold value, a bounding value and a traffic stream.
\begin{definition}
\label{def:distance}
Given a threshold value $\zeta >0$, a bounding value $\alpha>0$ and a traffic stream $R$,
{\em threshold violation distance} with respect to $R$, $\zeta$ and $\alpha$, is defined as the minimum time it takes such that the overshoot ratio reaches the bounding value $\alpha$. In other words,
\begin{equation}
    \label{eq:Dist}
    \text{Dist}_{\zeta,\alpha}(t;R):=\!\left\lbrace\begin{array}{ll}
        \!\hat{t}(\zeta)-t+\min_{R} dt, &
        \\
        \!\text{s.t:~~}o_\zeta(\hat{t}(\zeta)+dt)\!=\!\alpha,&\hspace{-2pt}\text{if}~o_\zeta(\hat{t}(\zeta))\leq \alpha, 
        \\
        \!0, & \hspace{-2pt}\text{otherwise},
    \end{array}
    \right.
\end{equation}
where $\hat{t}(\zeta)$ is defined as 
\begin{equation}
\hat{t}(\zeta)=  t+[\zeta-W_\rho(t;R)]^+/(C-\rho)
\label{eq:that}
\end{equation}
\end{definition}
Note that, in Definition~\ref{def:distance}, if the output workload is less than the threshold $\zeta$ then $\hat{t}(\zeta)$ will be the earliest time the output workload can increase to the threshold level $\zeta$. On the other hand, $\min_R dt$ is the minimum extra time the output workload needs to stay above $\zeta$ such that overshoot ratio with respect to $\zeta$ reaches the bounding value $\alpha$. Threshold violation distance for output traffic can be calculated using the following proposition.
\begin{proposition}
\label{prop:Distance}
For the output traffic if $o_\zeta(\hat{t}(\zeta))\leq \alpha$, then
\begin{align*}
\text{Dist}_{\zeta,\alpha}(t;R_{\rm o})=&\frac{\hat{t}(\zeta)-t}{1-\alpha}+\frac{\alpha t-O_\zeta(t;R_{\rm o})}{1-\alpha}\\
=&\frac{\hat{t}(\zeta)-t}{1-\alpha}+\frac{\alpha t-t o_\zeta(t)}{1-\alpha},
\end{align*}
where $\hat{t}(\zeta)$ can be derived according to~\eqref{eq:that} with $R$ replaced by $R_{\rm o}$.
\end{proposition}
\begin{proof}
According to~\eqref{eq:Dist} and~\eqref{eq:overshootratio}, if $o_\zeta(\hat{t}(\zeta))\leq \alpha$, then
\begin{equation}
 \frac{O_\zeta(t+\text{Dist}_{\zeta,\alpha}(t;R_{\rm o});R_{\rm o})}{t+\text{Dist}_{\zeta,\alpha}(t;R_{\rm o})}=\alpha   
 \label{eq:Disteq1}
\end{equation}
But clearly if $o_\zeta(\hat{t}(\zeta))\leq \alpha$, for the workload we need to have 
\begin{equation*}
W(t;R_{\rm o})>\zeta~~~~~\forall~t\in[\hat{t}(\zeta),t+\text{Dist}_{\zeta,\alpha}(t;R_{\rm o})].    
\end{equation*}
Therefore,
\begin{align}
&O_\zeta(t+\text{Dist}_{\zeta,\alpha}(t;R_{\rm o});R_{\rm o})=O_\zeta(\hat{t}(\zeta);R_{\rm o})+\text{Dist}_{\zeta,\alpha}(t;R_{\rm o})
\nonumber\\
&-(\hat{t}(\zeta)-t)=O_\zeta(t;R_{\rm o})+\text{Dist}_{\zeta,\alpha}(t;R_{\rm o})-(\hat{t}(\zeta)-t)
\label{eq:Disteq2}
\end{align}
Hence, using equations~\eqref{eq:that},~\eqref{eq:Disteq1} and~\eqref{eq:Disteq2}, Proposition~\ref{prop:Distance} can be derived.
\end{proof}
For the $j$th packet, we define $\text{dist}_j(T_i)$, for $1\leq i\leq M-1$, using definition of $\text{Dist}_{\zeta,\alpha}(t;R_{\rm o})$ for special values of $\zeta$, $\alpha$, and $t$ as follow:
\begin{align}
\text{dist}_j(T_i):=\left\lbrace\begin{array}{ll}
  \!\text{Dist}_{T_i,\Bar{f}(T_{i+1})}(b_j;R_{\rm o}),   & \text{for~~} i\in\mathcal{L}_j,\\
  \!\text{Dist}_{T_i ,\Bar{f}(T_{i})}(b_j;R_{\rm o}),  &  \text{for~~} i\in\mathcal{M}_j,
\end{array}\right.
\label{eq:dist}
\end{align}
where 
\begin{align}
    &\mathcal{L}_j=\{1\leq \ell\leq M-1:~T_\ell> \sigma^*(j) \},
    \label{eq:Lj}
    \\
    &\mathcal{M}_j=\{1\leq m\leq M-1:~T_m\leq \sigma^*(j) \}.
    \label{eq:Mj}
\end{align}
When the output workload is $W(b_j;R_{\rm o})$, $\text{dist}_j(T_i)$ for $i\in\mathcal{M}_j$, means the extra time the workload can be greater than $T_i$, such that the desired bound in~\eqref{eq:OvershootRatioBoundforallT} is violated. On the other hand, $\text{dist}_j(T_i)$ for $i\in\mathcal{L}_j$, means the time the workload can be greater than $T_i$, such that the constraint in~\eqref{eq:Ij} is violated. Note that, as $T_M$ is chosen large enough such that $W(t_j;R_{\rm o})$ is always less than $T_M$, therefore, $\text{dist}_j(T_M)$ can not be defined as the workload never goes beyond $T_M$.

For the example in Figs.~\ref{fig:WorkloadFluctuationNumericalExample} and~\ref{fig:OvershootRatioViolationinTime}, $\text{dist}_j(T_i)$ for $j=2191$ and $i=1,2\ldots,M-1$, is shown in Fig.~\ref{fig:dist}. For $j=2191$, $b_j$ is slightly less than $t_2$ In Fig.~\ref{fig:WorkloadFluctuationNumericalExample}. In Fig.~\ref{fig:dist}, $\text{dist}_{2191}(T_i)$ is shown in red if $T_i\leq\sigma^*(j)$ and is shown in blue if $T_i>\sigma^*(j)$. In other words,
 blue bars show how long the workload can stay above the corresponding $T_i$ according to Algorithm~\ref{alg:Stochastic regulator}. Red bars, however, show the longest time the workload can stay above the corresponding $T_i$ such that the desired upper bound at that $T_i$ is violated. Note that, if the blue bars are greater than the red bars for some $T_i$'s, then we can have the cases of the violations of the desired bound at the corresponding $T_i$'s for the red bars. This is actually the case in Fig.~\ref{fig:dist}. In this case,
 \begin{equation*}
     \text{dist}_{2191}(T_{21})=0.301e4>\text{dist}_{2191}(T_{16})=0.201e4
 \end{equation*}
 Therefore, as can be seen the workload is allowed to stay above $T_{21}$ according to Algorithm~\ref{alg:Stochastic regulator} on the interval $t\in[t_2,t_4]$, with $t_2=2.531e4$ and $t_4=2.805e4$. The length of this interval is $t_4-t_2=0.274e4$, which is greater $\text{dist}(t_1,T_{16})$. Therefore, although according to the Algorithm~\ref{alg:Stochastic regulator}, the output workload is allowed to stay above $T_{21}$, and no violation of~\eqref{eq:Ij} happens, the desired bound for $T_{16}$, however, as can be in seen in Fig.~\ref{fig:OvershootRatioViolationinTime}, is violated. 
\begin{figure}
\centering
\includegraphics[width=0.95\columnwidth]{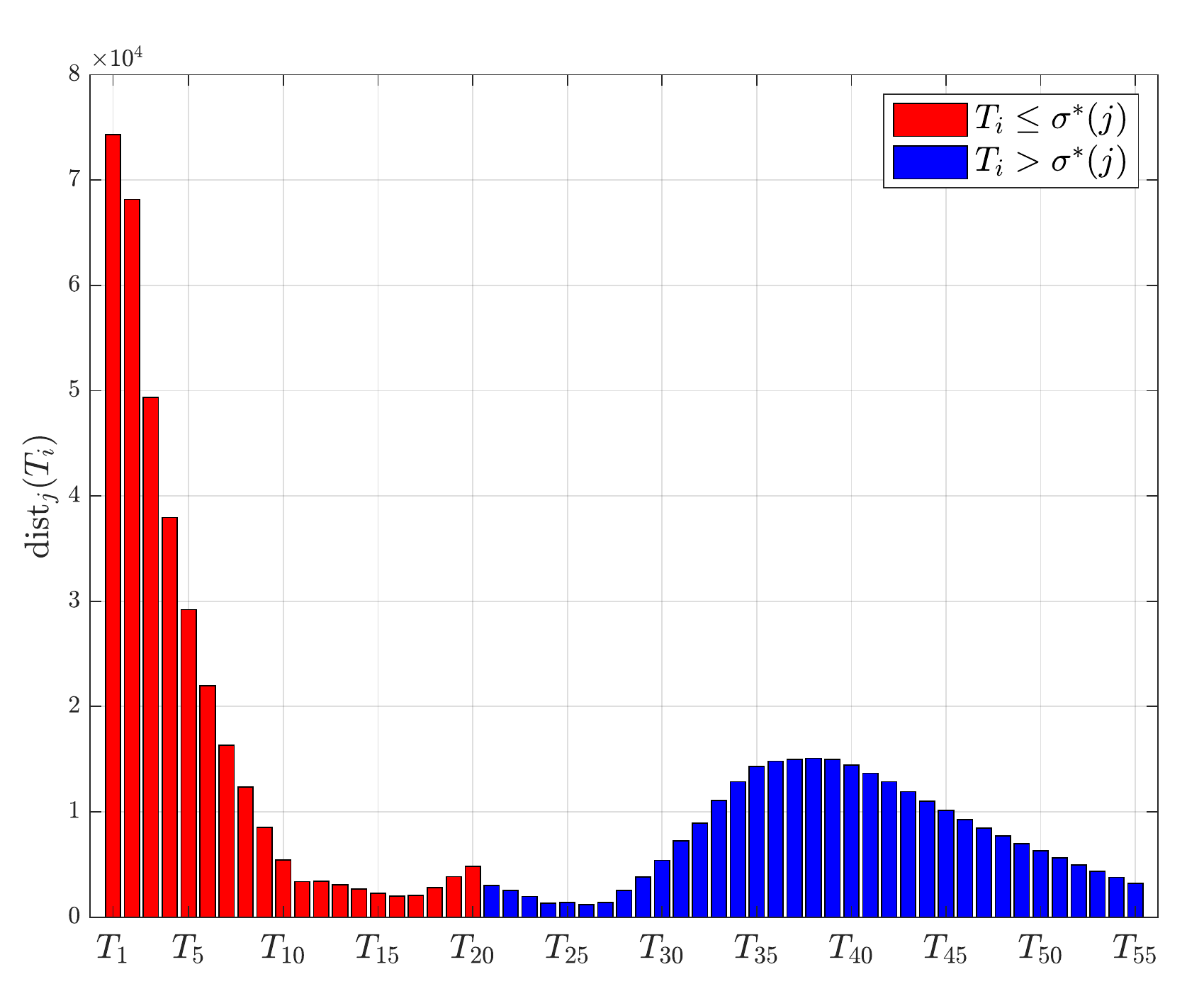}
\caption{$\text{dist}_j(T_i)$ for $i=1,2\ldots,M$ for the example in Section~\ref{sec:numerical} for Algorithm~\ref{alg:Stochastic regulator} with $M=56$, $j=2191$, $b_j=2.53e4$ and $\sigma^*(j)=71.54$.}
\label{fig:dist}
\end{figure}

On the other hand, for $j=8950$ and $b_j=11e4$, when $t$ is sufficiently large, $\text{dist}_{8950}(T_i)$ for $i=1,2\ldots,M$ is shown in Fig.~\ref{fig:distfinalt}. As we can see in Fig.~\ref{fig:distfinalt}, all the blue bars are less the red bars in this case.
\begin{figure}
\centering
\includegraphics[width=0.95\columnwidth]{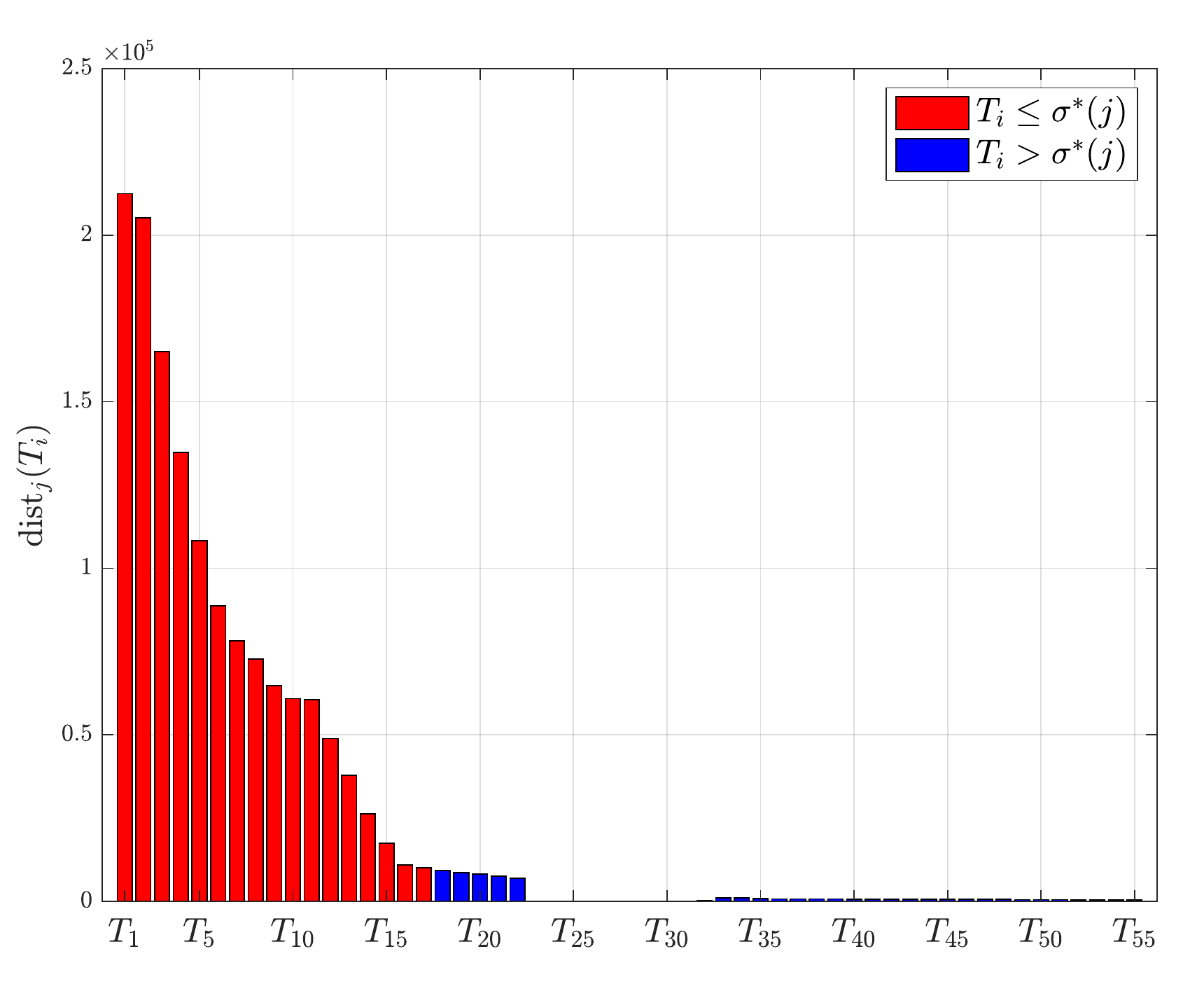}
\caption{$\text{dist}_j(T_i)$ for $i=1,2\ldots,M$ for the example in Section~\ref{sec:numerical} for Algorithm~\ref{alg:Stochastic regulator} with $M=56$, $j=8950$, $b_j=11e4$ and $\sigma^*(j)=60.82$.}
\label{fig:distfinalt}
\end{figure}

Based on the discussion for the specific example in Figs.~\ref{fig:dist} and~\ref{fig:distfinalt}, we can generalize these cases and present a sufficient condition on $\text{dist}_j(T_{i})$, for $i\in\{1,2,\ldots,M-1\}$, such that the desired bound in~\eqref{eq:OvershootRatioBoundforallT} is satisfied. If output workload is $W(b_j;R_{\rm o})$ and $b_{j}$ is sufficiently large enough, the sufficient condition to satisfy the desired constraint~\eqref{eq:OvershootRatioBoundforallT} is,
\begin{align}
&\text{dist}_j(T_{\ell}) \leq  \text{dist}_j(T_{m})~~\forall ~\ell\in\mathcal{L}_j,~\forall m\in\mathcal{M}_j,
\label{eq:sufficientconditionondist}
\end{align}
 if $\mathcal{M}_j\neq \emptyset$. Note that, if $\mathcal{M}_j= \emptyset$, then $\sigma^*(j)=\sigma_1$. In this case  $\bar{f}(\gamma)=1$ for $\gamma\in[T_0,T_1]$. Therefore, the desired bound of
 \begin{equation*}
 o_{T_0}(t)\leq \bar{f}(\gamma),     
 \end{equation*}
 will never be violated, independent of the duration of the interval that the workload stays above $T_0$. In the definition of $\text{dist}_j(T_i)$ and in the sufficient condition in~\eqref{eq:sufficientconditionondist}, we are just considering the complete departure times and we verify the sufficient condition at those moments. Ascertaining the sufficient condition at those complete departure times moment, however, can guarantee the desired condition in~\eqref{eq:OvershootRatioBoundforallT} is satisfied for all sufficiently large $t$. Because if after the departure of every packet we can assure the duration of the time that the workload stays above the $T_m$ for $\forall T_m\in\mathcal{M}_j$, is less than the time to violate the desired condition in~\eqref{eq:OvershootRatioBoundforallT}, then the desired condition in~\eqref{eq:OvershootRatioBoundforallT} is not only satisfied at the complete departure times, but also it is satisfied at all sufficiently large $t$. 

Note that, if $\mathcal{M}_j\neq \emptyset$, according to~\eqref{eq:Mj} and~\eqref{eq:Ij}, 
\begin{equation}
    \sigma^*(j)=\sigma_{\max\mathcal{M}_j+1}.
    \label{eq:Sigmastar_Mj}
\end{equation}
On the other hand, in the sufficient condition in~\eqref{eq:sufficientconditionondist}, if instead of all $m\in\mathcal{M}_j$, just $m=\max\mathcal{M}_j$ is considered the sufficient condition will be simplified as,
\begin{align}
&\text{dist}_j(T_{\ell}) \leq  \text{dist}_j(T_{m})~~\forall ~\ell\in\mathcal{L}_j,~ m=\max\mathcal{M}_j,
\label{eq:sufficientconditionondist1}
\end{align}
 if $\mathcal{M}_j\neq \emptyset$ and $b_j$ is sufficiently large enough.  This simplified sufficient condition can be explained as follows, when $\sigma^*(j)$ is set according to Algorithm~\ref{alg:Stochastic regulator} and according to~\eqref{eq:Sigmastar_Mj} and~\eqref{eq:Ij}, overshoot ratio at $b_j$ with respect to $T_{\max\mathcal{M}_j}$ is checked against $\Bar{f}(T_{\max\mathcal{M}_j+1})$ or 
 \begin{equation}
 o_{T_{\max\mathcal{M}_j}}(b_j)\leq \Bar{f}(T_{\max\mathcal{M}_j+1}).
 \label{eq:overshootRatioatmaxMj}
 \end{equation}
 If we make sure the duration of the time that the workload stays above $T_{\max\mathcal{M}_j}$ is less than the time to violate the desired upper bound or 
 \begin{equation}
 o_{T_{\max\mathcal{M}_j}}(j) \leq \Bar{f}(T_{\max\mathcal{M}_j+1})~~\text{for~}\forall t>b_j,   
 \end{equation}
 for all packets that $b_j$ is sufficiently large. Then the desired bound in~\eqref{eq:OvershootRatioBoundforallT} is never violated for sufficiently large $t$. It can easily be shown the sufficient condition in~\eqref{eq:sufficientconditionondist} and~\eqref{eq:sufficientconditionondist1} are equivalent.
 
 Using Proposition~\ref{prop:Distance}, we can simply $\text{dist}_j(T_i)$ in the two following cases:
\setcounter{case}{0}
\begin{case}
$i\in\mathcal{L}_j$
\end{case}
\begin{align}
\text{dist}_j(T_i)=
    \frac{b_j(\Bar{f}(T_{i+1})- o_{T_i}(b_j))}{1-\Bar{f}(T_{i+1})}+\frac{\hat{t}(T_i)-b_j}{1-\Bar{f}(T_{i+1})},  
\label{eq:distcase1}
\end{align} 
, if 
\begin{equation}
    b_j(\Bar{f}(T_{i+1})-o_{T_i}(b_j))>-(\hat{t}(T_i)-b_j)\Bar{f}(T_{i+1}).
    \label{eq:distcase1condition}
\end{equation} 
Otherwise, $\text{dist}_j(T_i)=0$. In~\eqref{eq:distcase1}, $\hat{t}(T_i)$ can be derived from~\eqref{eq:that}, with $\zeta=T_i$ and $t=b_j$.
\begin{case}
$i\in\mathcal{M}_j$
\end{case}
\begin{align}
\text{dist}_j(T_i)=
    \frac{b_j(\Bar{f}(T_{i})- o_{T_i}(t))}{1-\Bar{f}(T_{i})}+\frac{\hat{t}(T_i)-b_j}{1-\Bar{f}(T_{i})}, 
\label{eq:distcase2}
\end{align} 
, if 
\begin{equation}
    b_j(\Bar{f}(T_{i})-o_{T_i}(b_j))>-(\hat{t}(T_i)-b_j)\Bar{f}(T_{i}).
    \label{eq:distcase2condition}
\end{equation} 
Otherwise, $\text{dist}_j(T_i)=0$. In~\eqref{eq:distcase2}, $\hat{t}(T_i)$ can be derived from~\eqref{eq:that}, with $\zeta=T_i$ and $t=b_j$.

In order to show the desired condition in~\eqref{eq:OvershootRatioBoundforallT} is satisfied for sufficiently large values of $t$ if $\sigma^*(j)$ is chosen according to Algorithm~\ref{alg:Stochastic regulator}, we assume $b_j$ is sufficiently large and we show the sufficient condition in~\eqref{eq:sufficientconditionondist1} is satisfied. 

 According to~\eqref{eq:distcase1} and~\eqref{eq:distcase2}, if we have the following inequality, then the sufficient condition in~\eqref{eq:sufficientconditionondist1} is also satisfied,
\begin{align}
&\frac{b_j(\Bar{f}(T_{\ell+1})- o_{T_\ell}(b_j))}{1-\Bar{f}(T_{\ell+1})}+\frac{\hat{t}(T_\ell)-b_j}{1-\Bar{f}(T_{\ell+1})}
\nonumber \\
&\leq  \frac{b_j(\Bar{f}(T_{m})-o_{T_m}(b_j))}{1-\Bar{f}(T_{m})},
\nonumber\\
&\forall\ell\in\mathcal{L}_j,~ m=\max\mathcal{M}_j,~\mathcal{M}_j\neq \emptyset.
\label{eq:simplified-sufficient-inequality}
\end{align}
The second term in the LHS can be bounded according to~\eqref{eq:that} as follows,
\begin{align}
&\frac{\hat{t}(T_\ell)-b_j}{1-\Bar{f}(T_{\ell+1})}=\frac{[T_\ell-W_\rho(b_j;R_{\rm o})]^+}{(1-\Bar{f}(T_{\ell+1}))(C-\rho)}
\nonumber\\
&\leq \frac{T_{M-1}-W_\rho(b_j;R_{\rm o})}{(1-\Bar{f}(T_{M}))(C-\rho)}
\leq \frac{T_{M-1}-\sigma_1}{(1-\Bar{f}(T_{M}))(C-\rho)}:=c_0>0,
\label{eq:c0}
\end{align}
where we have used $W_\rho(t;R_{\rm o})> \sigma_1$ as $\mathcal{M}_j\neq \emptyset$. 
Therefore, we can simply the inequality in~\eqref{eq:simplified-sufficient-inequality} into a more conservative simplified inequality as follows,
\begin{align}
&\frac{b_j(\Bar{f}(T_{\ell+1})- o_{T_\ell}(b_j))}{1-\Bar{f}(T_{\ell+1})}+c_0
\leq  \frac{b_j(\Bar{f}(T_{m})-to_{T_m}(b_j))}{1-\Bar{f}(T_{m})},
\nonumber\\
&\forall\ell\in\mathcal{L}_j,~m=\max{\mathcal{M}_j},~\mathcal{M}_j\neq \emptyset.
\label{eq:simplified-sufficient-inequality2}
\end{align}
Note that, if the inequality~\eqref{eq:simplified-sufficient-inequality2} is satisfied, then inequality~\eqref{eq:simplified-sufficient-inequality} is also satisfied. By doing some manipulations we can reach the following inequality,
\begin{align}
\Bar{f}(T_{\ell+1})+\frac{c_2}{b_j c_1}-\frac{\bar{f}(T_m)-o_{T_m}(b_j)}{c_1}\leq o_{T_\ell}(b_j),
\label{eq:lovershootratiolowerbound}
\end{align}
where 
\begin{align}
    &c_1=\frac{1-\Bar{f}(T_m)}{1-\Bar{f}(T_{\ell+1})},
    \label{eq:c1}
    \\
    &c_2=c_0(1-\Bar{f}(T_m)).
    \label{eQ:c2}
\end{align}
In other words, for the sufficient condition in~\eqref{eq:sufficientconditionondist} to hold, the overshoot ratio, $o_{T_\ell}(b_j)$ for $\forall\ell\in\mathcal{L}_j$ should be higher than the lower bound specified in~\eqref{eq:lovershootratiolowerbound}. By considering the upper bound on $o_{T_m}(b_j)$ in~\eqref{eq:overshootRatioatmaxMj}, the lower bound in~\eqref{eq:lovershootratiolowerbound} can be simplified into a more conservative inequality as follows,
\begin{align}
\Bar{f}(T_{\ell+1})+\frac{c_2}{b_j c_1}-\frac{\epsilon_m}{c_1}\leq o_{T_\ell}(b_j),
\label{eq:lovershootratiolowerbound1}
\end{align}
where 
\begin{equation}
\epsilon_m:=\Bar{f}(T_m)- \Bar{f}(T_{m+1}).
\label{eq:epsilonm}
\end{equation}
Note that, if the lower bound in~\eqref{eq:lovershootratiolowerbound1} is satisfied, then the lower bound in~\eqref{eq:lovershootratiolowerbound} is also satisfied. 

For the numerical example in Figs.~\ref{fig:OvershootRatioViolationinTime}-\ref{fig:distfinalt}, the overshoot ratio, $o_{T_{16}}(t)$ and the lower bounds in~\eqref{eq:lovershootratiolowerbound} and~\eqref{eq:lovershootratiolowerbound1} are shown in Fig.~\ref{fig:OvershootRatioLowerBound}. As it was mentioned before, these lower bounds are sufficient conditions for the desired constraint in~\eqref{eq:OvershootRatioBoundforallT} to hold. The intervals on which the desired constraint is violated or,
\begin{equation}
 o_{T_{16}}(t)>\Bar{f}(T_{16}), 
 \label{eq:desired_constraint_violated_Interval}
\end{equation}
are shown in shaded blue areas.
Therefore, as it can be seen in Fig.~\ref{fig:OvershootRatioLowerBound}, there are some parts that these lower bounds are violated but the desired constraint in~\eqref{eq:OvershootRatioBoundforallT} is not violated. On the other hand, on the intervals that the desired constraint in~\eqref{eq:OvershootRatioBoundforallT} is violated, as it is shown the corresponding lower bounds are also violated. As it can be seen in Fig.~\ref{fig:OvershootRatioLowerBound}, we do not need a very large $b_j$ to satisfy the lower bound in~\eqref{eq:lovershootratiolowerbound}. However, for the more conservative lower bound in~\eqref{eq:lovershootratiolowerbound1}, a larger $b_j$ is necessary.
\begin{figure}
\centering
\includegraphics[width=0.95\columnwidth]{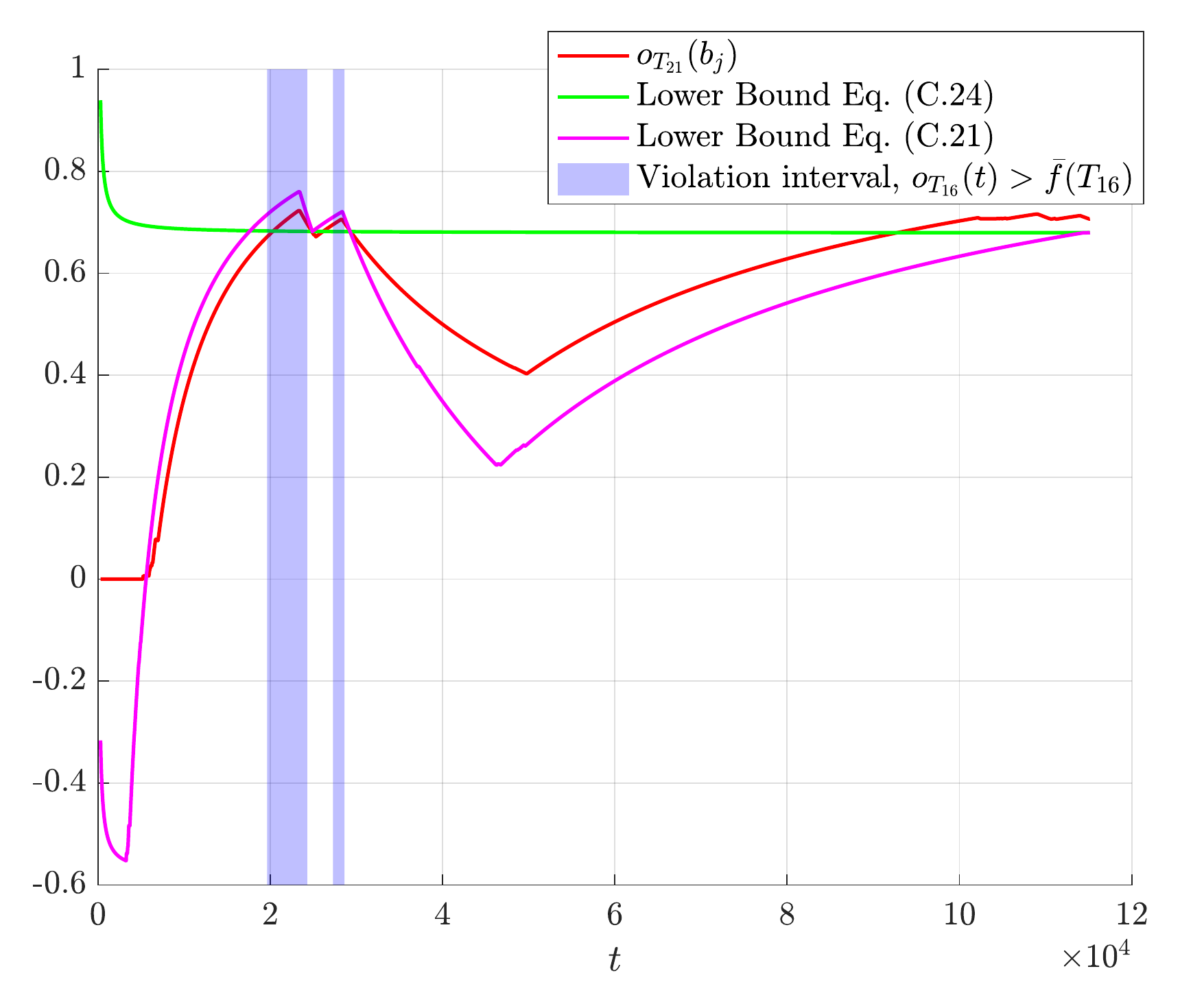}
\caption{$o_{T_{16}}(t)$ and the corresponding lower bounds in~\eqref{eq:lovershootratiolowerbound} and~\eqref{eq:lovershootratiolowerbound1}, for the example in Section~\ref{sec:numerical} for Algorithm~\ref{alg:Stochastic regulator} with $M=56$.}
\label{fig:OvershootRatioLowerBound}
\end{figure}

Now we show when $b_j$ is sufficiently large, the lower bound in~\eqref{eq:lovershootratiolowerbound1} holds. Let define the event $\xi_j(T_{\ell})$, for $\ell\in\{1,2,\ldots,M-1\}$ and the $k$th packet as,
\begin{equation}
 \xi_k(T_{\ell}):=\{W(\tilde{s}_k;R_1)\geq T_\ell \cap W(t_k;R_{\rm o})<T_\ell\}
 \label{eq:xi_j_ell}
\end{equation}
In other words, when the event $\xi_k(T_{\ell})$ occurs, the $k$th packet is delayed enough such that the output workload becomes less than $T_\ell$. The event $\xi_k(T_{\ell})$ occurs when 
\begin{equation}
    o_{T_\ell}(b_k(\sigma_{\ell+1}))>f(T_{\ell+1}).
\end{equation}
It can be easily shown between input traffic overshoot ratio, the internal traffic overshoot ration and the output traffic overshoot ratio we have the following relation
\begin{equation}
    o_{T_{\ell}}(t)\leq \frac{O_{T_{\ell}}(t;R_1)}{t}\leq \frac{O_{T_{\ell}}(t;R_{\rm i})}{t}
    \label{eq:OvershootRatioInequalities}
\end{equation}
Note that, due to ergodicity and stationarity of the input and internal traffic,
we have
\begin{align*}
 &\frac{O_{T_{\ell}}(t;R_{\rm i})}{t}\sim\Pr\{W_\rho(t;R_{\rm i})\geq T_\ell\},
 \\
&\frac{O_{T_{\ell}}(t;R_{\rm 1})}{t}\sim\Pr\{W_\rho(t;R_{\rm 1})\geq T_\ell\} 
\end{align*}
Therefore, it can be shown if the probability of the input traffic being greater than $T_\ell$ is greater than $\Bar{f}(T_\ell)$, then the probability of the event $\xi_j(T_\ell)$ is greater than zero. In other words,
\begin{equation}
    \text{if~~}\Pr\{W_\rho(t;R_{\rm i})\geq T_\ell\}\geq \Bar{f}(T_\ell) ~~\rightarrow~~\Pr\{\xi_j(T_\ell)\}>0.
\end{equation}
Let define $\Tilde{t}_{T_\ell}(t)$ for $\ell\in\{1,2,\ldots,M-1\}$ as the last time before $t$ that event $\xi_k(T_\ell)$ happened. In other words,
\begin{equation}
    \Tilde{t}_{T_\ell}(t)=\max \{b_k\leq t : \mathbf{1}_{\xi_k(T_\ell)}=1 \},
\end{equation}
for $\ell\in\{1,2,\ldots,M-1\}$, where 
\begin{equation}
 \mathbf{1}_{A}=\left\lbrace\begin{array}{ll}
   1   & \text{if event}~A~\text{occurs}, \\
   0   & \text{if event}~A~\text{does not occur}.
 \end{array}\right.   
\end{equation}
As the output workload is stationary and ergodic and the probability of the event is greater than zero, therefore, the interval between consecutive occurs of the events $\xi_k(T_\ell)$ is bounded. Next we show for the lower bound in~\eqref{eq:lovershootratiolowerbound1} to hold, $b_j-\Tilde{t}_{T_\ell}(b_j)$ should have an upper bound. In other words, we find the minimum value of $b_j-\Tilde{t}_{T_\ell}(b_j)$, such that the lower bound in~\eqref{eq:lovershootratiolowerbound1} is violated and then we verify that when $b_j$ is sufficiently large, $b_j-\Tilde{t}_{T_\ell}(b_j)$ will be always less than this minimum value.

Note that,
\begin{equation}
    W(t;R_{\rm o})\leq T_\ell,~~\forall t\in[\Tilde{t}_{T_\ell}(b_j),b_j].
    \label{eq:workloadoIntervalttildet}
\end{equation}
Therefore, according to~\eqref{eq:OvershootRatioInequalities}, in order to find the minimum value of $b_j-\Tilde{t}_{T_\ell}(b_j)$, such that the lower bound in~\eqref{eq:lovershootratiolowerbound1} is violated we consider
\begin{align}
&o_{T_\ell}(b_j)=\Bar{f}(T_{\ell+1})+\frac{c_2}{b_j c_1}-\frac{\epsilon_m}{c_1},
\\
&o_{T_\ell}(\Tilde{t}_{T_\ell}(b_j))=\bar{f}(T_{\ell+1}).
\end{align}
Therefore, according to~\eqref{eq:overshootratio} and~\eqref{eq:workloadoIntervalttildet}, 
\begin{equation}
    b_j-\Tilde{t}_{T_\ell}(b_j)=b_j\frac{\epsilon_m-c_2/b_j}{ c_1\Bar{f}(T_{\ell+1})}=O(b_j).
    \label{eq:MinimumConsecutiveDistancetoViolate}
\end{equation}
Therefore, the minimum time interval that needs to pass between $\Tilde{t}_{T_\ell}(b_j)$ and $b_j$, such that the lower bound in~\eqref{eq:lovershootratiolowerbound1} is violated is linearly proportional to $b_j$. But as $b_j$ increases the time interval between consecutive occurrences of the event $\xi_k(T_\ell)$, will be less than $O(b_j)$ with probability 1. Note that, $b_j-\Tilde{t}_{T_\ell}(b_j)$ is less the time interval between consecutive occurrences of the event $\xi_k(T_\ell)$. Therefore, $b_j-\Tilde{t}_{T_\ell}(b_j)$ will be always less then the upper bound derived in~\eqref{eq:MinimumConsecutiveDistancetoViolate}. Hence, the lower bound in~\eqref{eq:lovershootratiolowerbound1} is always met for sufficiently large values of $b_j$.

For the numerical example in Figs.~\ref{fig:OvershootRatioViolationinTime}-\ref{fig:distfinalt}, $b_j-\Tilde{t}_{T_{21}}(b_j)$ and the higher bound in~\eqref{eq:MinimumConsecutiveDistancetoViolate} are shown in Fig.~\ref{fig:DeclineEventPassedTimeUpperBound}. As it can be seen, when $b_j$ is sufficiently large, $b_j-\Tilde{t}_{T_{21}}(b_j)$ will be bounded by the higher bound in~\eqref{eq:MinimumConsecutiveDistancetoViolate}.
\begin{figure}
\centering
\includegraphics[width=0.95\columnwidth]{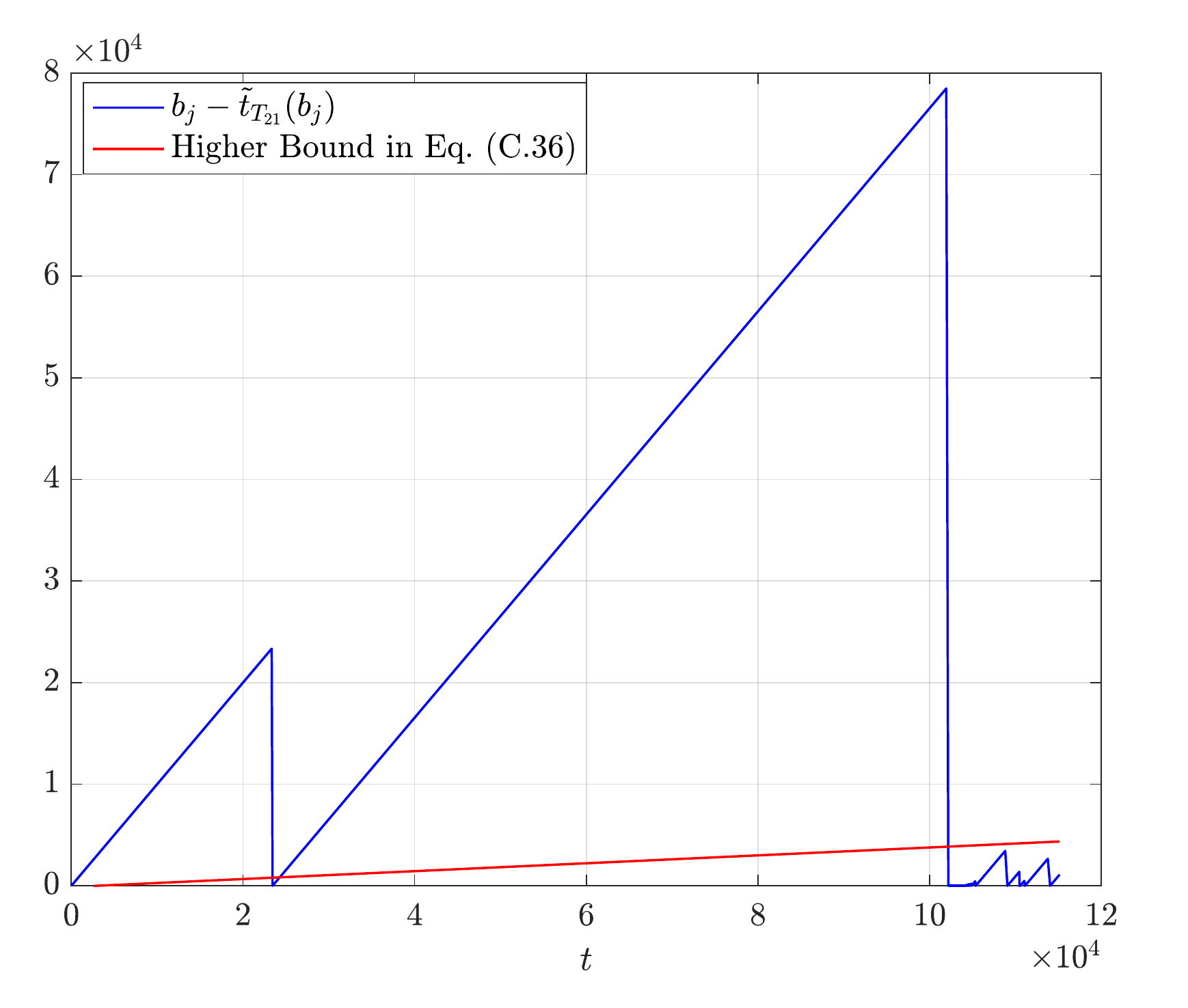}
\caption{$b_j-\Tilde{t}_{T_{21}}(b_j)$ and the higher bound in~\eqref{eq:MinimumConsecutiveDistancetoViolate} for the example in Section~\ref{sec:numerical} for Algorithm~\ref{alg:Stochastic regulator} with $M=56$.}
\label{fig:DeclineEventPassedTimeUpperBound}
\end{figure}
As it was mentioned in Appendix~\ref{appx_subsec:overshoot_soln1_I}, for this numerical example the lower bound on $b_j$, in order to satisfy the preliminary constraint in~\eqref{eq:canonical_constraintPreliminaryTheorem}, is $b_j\geq 2.35\times 10^3$. In this example, in order to satisfy the desired bound in~\eqref{eq:canonical_constraint}, however, the lower bound is increased to $b_j\geq 10^5$. 
\end{proof}

\begin{proof}[Proof of Theorem~\ref{thm:overshoot_soln1}]
In Lemma~\ref{appxlem:all_gamma}, we showed if $t$ is sufficiently large, then in a $(\sigma^*,\rho)$ traffic regulator defined by~\eqref{eq:cBj}--\eqref{eq:basic_implementation}, 
\begin{equation}
o_{T_i}(t)\leq \bar{f}({T_i})~~~~\text{for~~}\forall i\in\{1,\ldots,M\}.
\end{equation}
Therefore, using the same argument as in Appendix~\ref{appx:overshoot_soln2}, and using Corollary~\ref{appxcorollary:1} and the definition of $\bar{f}(\gamma)$ as~\eqref{eq:f_bar} for the case $M=M_{\max}$, or using Corollary~\ref{appxcorollary:2} and the definition of $\bar{f}(\gamma)$ as~\eqref{eq:f_barnew} for the case $M<M_{\max}$ we can show 
\begin{align}
o_{\gamma}(t) \leq  f(\gamma),
~~\forall~t\in[b_{j-1},b_{j}(\sigma^*(j))],~~ \forall \gamma \in [T_1, T].
\end{align}
\end{proof}
In practice the sufficiently large $t$ constraint for. Algorithm~\ref{alg:Stochastic regulator} is reasonable as we are approximating the overshoot probability with the overshoot ratio in~\eqref{eq:Prob_approx}, and this approximation is asymptotically valid. 
In Algorithm~\ref{alg:Stochastic regulator} we need to compute the index set $\cI_j$ in \eqref{eq:Ij}. The process for computing set $\cI_j$ is depicted in Fig.~\ref{fig:Thm1procedure} for two possible cases: 1) $\tilde{s}_j=s_j$ in Figs.~\ref{fig:Thm1procedureCase1a}--\ref{fig:Thm1procedureCase1c} and $\tilde{s}_j=b_{j-1}$ in Figs.~\ref{fig:Thm1procedureCase2a}--\ref{fig:Thm1procedureCase2c}. In Fig.~\ref{fig:Thm1procedure}, $k=\min \cB_j$. In the first step of computing $\cI_j$ according to~\eqref{eq:Ij}, 
$\ell$ is set to $k$ as in Figs.~\ref{fig:Thm1procedureCase1a} and~\ref{fig:Thm1procedureCase2a}. In these cases according to~\eqref{eq:tj_R1}, $t_j=\tilde{s}_j$. Then $o_{T_{\ell-1}}(b_j)$  is determined using Proposition~\ref{prop:overshoot}.  If the condition in~\eqref{eq:Ij} holds for $\ell=k$, then $k \in \cI_j$. Therefore, $\sigma^*$ will be set as $\sigma^*=\sigma_k$ and the algorithm will terminate at this step. Otherwise, in the next step we set $\ell=k-1$ as in Figs.~\ref{fig:Thm1procedureCase1b} and~\ref{fig:Thm1procedureCase2b}. In these cases, $t_j$ will be determined according to~\eqref{eq:tj_R1}. Again $O_{T_{\ell-1}}(b_j;R_{\rm o})$ will be determined using Proposition~\ref{prop:overshoot} and the condition in~\eqref{eq:Ij} is checked for $\ell=k-1$. If $k-1 \in \cI_j$, then using the same argument as before
we set $\sigma^*=\sigma_{k-1}$ and the algorithm will terminate at this step. Otherwise these steps are continued
as shown in Figs.~\ref{fig:Thm1procedureCase1c} and~\ref{fig:Thm1procedureCase2c} and the same process is repeated. If $\cI_j$ is determined to be empty, then we set $\sigma^*=\sigma_1$. 

\input{figs/Algorithm_Ilustration}

\section*{Acknowledgments}
The authors thank Prof.\ Yariv Ephraim for helpful comments and discussions on this work.


\end{document}